\documentclass[a4paper,oneside,11pt]{article}
\usepackage{amsmath,amsfonts,amscd,amssymb}
\usepackage{longtable,geometry}
\usepackage[english]{babel}
\usepackage[utf8]{inputenc}
\usepackage[active]{srcltx}
\usepackage[T1]{fontenc}
\usepackage{graphicx}
\usepackage{bbm}
\usepackage{enumitem}
\usepackage{MnSymbol}
\usepackage{stmaryrd}
\usepackage{nicefrac}
\usepackage{calrsfs}
\usepackage{hyperref}
\usepackage{xcolor}
\usepackage{bm}

\usepackage{cancel} 

\def\ket#1{|#1\rangle}
\def\bra#1{\langle #1|}

\usepackage{mathtools}

\def\Z{\mathbb{Z}}                          
\def\R{\mathbb{R}}

\newtheorem{theorem}{Theorem}[section]

\newtheorem{corollary}[theorem]{Corollary}
\newtheorem{lemma}[theorem]{Lemma}
\newtheorem{proposition}[theorem]{Proposition}
\newtheorem{definition}[theorem]{Definition}

\newtheorem{assumption}[theorem]{Assumption}
\newtheorem{remark}[theorem]{Remark}

\newcommand{\be}[1]{\begin{equation}\label{#1}}
\newcommand{\ee}{\end{equation}}
\numberwithin{equation}{section}

\newcommand{\ba}[1]{\begin{align}\label{#1}}
\newcommand{\ea}{\end{align}}
\numberwithin{equation}{section}

\newcommand{\ben}{\begin{equation*}}
\newcommand{\een}{\end{equation*}}
\numberwithin{equation}{section}

\newenvironment{proof}[1][\relax]
  {\paragraph{Proof\ifx#1\relax\else~of #1\fi}}%
  {~\hfill$\square$\par\bigskip}

\newcommand{\calB}{\mathcal{B}}

\newcommand{\calE}{\mathcal{E}}
\newcommand{\calF}{\mathcal{F}}

\newcommand{\calK}{\mathcal{K}}
\newcommand{\calL}{\mathcal{L}}
\newcommand{\calM}{\mathcal{M}}

\newcommand{\calR}{\mathcal{R}}
\newcommand{\calS}{\mathcal{S}}
\newcommand{\calT}{\mathcal{T}}

\newcommand{\calV}{\mathcal{V}}

\newcommand{\calX}{\mathcal{X}}

\newcommand{\bbA}{\mathbb{A}}

\newcommand{\bbC}{\mathbb{C}}

\newcommand{\bbI}{\mathbb{I}}

\newcommand{\bbR}{\mathbb{R}}

\newcommand{\bbT}{\mathbb{T}}

\newcommand{\bbZ}{\mathbb{Z}}
\def\Z{\mathbb{Z}}                          
\def\be{\begin{equation}}
\def\ee{\end{equation}}

\newcommand{\ep}{\varepsilon}

\newcommand{\n}{\mathbf{n}}
\newcommand{\m}{\mathbf m}

\newcommand{\rk}[1]{\bgroup\color{red}%
  \par\medskip\hrule\smallskip%
  \noindent\textbf{#1}%
  \par\smallskip\hrule\medskip\egroup}

\title{Marginal triviality of the scaling limits \\  of critical 4D  Ising and $\phi_4^4$ models}
\author{Michael Aizenman\thanks{\texttt{aizenman@princeton.edu} Departments of Physics and Mathematics, Princeton University}\ \ and Hugo Duminil-Copin\thanks{
\texttt{duminil@ihes.fr} Institut des Hautes \'Etudes Scientifiques and Universit\'e de Gen\`eve
}}
\date{25 January 2021}

\begin{document}
\maketitle
\begin{abstract}
We prove that the scaling limits of spin fluctuations in four-dimensional Ising-type models with nearest-neighbor ferromagnetic interaction at or near the critical point are Gaussian. 
A similar statement is proven for the  
$\lambda \phi^4$ fields over $\mathbb{R}^4$ with a lattice ultraviolet cutoff, in the limit of infinite volume and vanishing lattice spacing.
 The proofs are enabled by the models' random current representation, in which the correlation functions' deviation from Wick's law is expressed in terms of intersection probabilities of random currents with  sources at distances which are large on the model's  lattice scale.   Guided by the analogy with random walk intersection amplitudes, the analysis focuses on the improvement of the so-called tree diagram bound by a logarithmic correction term, which is derived here through multi-scale analysis.  
 \end{abstract}

\section{Introduction}

The results presented below address questions pertaining to two distinct research agendas: one aims at Constructive Field Theory and the other at the understanding of the critical behavior in Statistical Mechanics.   
While these two goals are somewhat different 
the questions and the answers are related. We start with their brief presentation. 

\subsection{Constructive Quantum Field Theory and Functional Integration}\label{sec:1.2}

Quantum field theories with local interaction  play an important role in the physics discourse, where they appear in  subfields ranging from high energy to condensed matter physics.  
The mathematical challenge of proper formulation of this concept led to programs of Constructive Quantum Field Theory (CQFT).  A path  towards that goal was charted through the  proposal to define quantum fields as operator valued distributions whose essential properties are  formulated as the Wightman axioms~\cite{Wig56}.  Wightman's reconstruction theorem  allows one to recover this structure from the collection of the corresponding correlation functions, defined over 
the Minkowski space-time.  By the Osterwalder-Schrader theorem~\cite{OS73,OS75}, correlation functions 
with the required properties may potentially be obtained through analytic continuation from those of random distributions defined over the corresponding Euclidean space that meet a number of conditions: suitable analyticity, permutation symmetry, Euclidean covariance, and reflection-positivity.

Seeking natural candidates for such \emph{Euclidean fields}, one ends up with the task of constructing  probability averages over random distributions $\Phi(x)$, for which the expectation value of  functionals   $F(\Phi)$ would  have properties fitting the  formal expression  
\be \label{Phi_EV}
\langle F(\Phi) \rangle  \approx \frac1{\rm norm}   \int  F(\Phi) \exp[-H(\Phi)]  \prod_{x\in \R^d} d \Phi(x) ,
\ee 
where  
 $ H(\Phi) $ is the Hamiltonian. In this context, it seems natural to consider expressions of the form
  \be \label{eq:FI} 
H(\Phi):\approx  (\Phi, A  \Phi) + \int_{\R^d}   P(\Phi(x)  )  \, dx
\ee
with  $(\Phi, A  \Phi)$ a positive definite and reflection-positive quadratic form, and 
$P(\Phi(x)  )$ a polynomial (or a more general function) whose terms of  order $\Phi(x)^{2k}$ are interpreted heuristically as representing $k$-particle interactions.   An example of a quadratic form with the above properties  (at $K, b>0$) and also rotation invariance is  
\be \label{free_field}
  (\Phi, A  \Phi):=  \int_{\R^d}  \left( K \vert \nabla \Phi \vert^2(x) + b \vert \Phi(x) \vert^2 \right)   \, dx.
\ee

The functionals $F(\Phi)$ to which \eqref{Phi_EV} is intended to apply include  the smeared averages  
\be\label{eq:functional}
T_f(\Phi) :=   \int_{\R^d} f(x) \Phi(x) dx
\ee
 associated with continuous functions of compact support $f\in C_0(\R^d)$.    
By linearity, the expectation values of products of such variables take the form
\be \label{eq:Schwinger} 
\langle \prod_{j=1}^{n} T_{f_j}(\Phi) \rangle   :=   \int_{(\R^d)^n}  d x_1 \dots d x_n  \, S_n(x_1,\dots, x_n) \,\prod_{j=1}^n  f(x_j), 
\ee 
with  $S_n(x_1,\dots, x_n)$  characterizing the  probability measure on the space of distribution which corresponds to the expectation value $\langle  --  \rangle$.    This is summarized by saying that 
 in a distributional sense 
\be
\langle \prod_{j=1}^{n} \Phi(x_j) \rangle = S_n(x_1,\dots, x_n)    \, , 
\ee
with $S_n$ referred to as the {\em Schwinger functions} of the corresponding euclidean field theory.  

A relatively simple class of Euclidean fields are the Gaussian fields, for which  $H$ contains  only quadratic terms.   Gaussian fields (whether reflection-positive or not) are alternatively characterized by having their structure determined  by just  the two-point function, with  the $2n$-point Schwinger functions computable through  Wick's law:
\be \label{triv}
S_{2n}(x_1,\dots,x_{2n})  =  \sum_{\pi}  \prod_{j=1}^n S_2(x_{\pi(2j-1)},x_{\pi(2j)})\, =: \, \mathcal {G}_n[S_2](x_1,\dots,x_{2n}) \, ,
\ee 
where $\pi$ ranges over pairing permutations of $\{1,\dots,2n\}$.   
The field theoretic interpretation of \eqref{triv} is the absence of interaction.  Due to that, and to their algebraically simple structure, such fields have been referred to as {\em trivial}.  

When interpreting   \eqref{Phi_EV}, one quickly encounters  a  number of problems.    Even in the generally understood case of the Gaussian \emph{free field},  with   $H$ consisting of just the quadratic term \eqref{free_field},  Equation \eqref{Phi_EV} is not to be taken  literally as   the measure is supported by non-differentiable functions for which the integral in the exponential is almost surely divergent.  

A natural step to tackle next seems to be the addition of  the lowest order even term, i.e. $\lambda \Phi^4$.   
However, in dimensions $d>1$, the free field is no longer a random function but a random distribution which even locally is unbounded.  Thus such simple looking proposals lead to additional divergences, whose severity increases with the dimension.  

The heuristic ``renormalization group'' approach to the problem  by K. Wilson~\cite{Wil71} indicates that in low enough dimensions, specifically $d<4$ for $\lambda \Phi^4$,  the problem could be tackled through cutoff-dependent counter-terms.  Partially successful attempts to carry such a project rigorously have been the focus of a substantial body of works.    The means employed have included: counter-terms, which are allowed to depend on regularizing cutoffs, scale decomposition, renormalization group flows, the theory or regularity structures~\cite{Hai14}, etc.  

A natural starting point towards such a construction of  a  $\Phi^4_d$  functional integral \eqref{Phi_EV} is to  regularize it with a pair of cutoffs: at the short distance (\emph{ultraviolet}) scale and the large distance (\emph{infrared}) scale.    A lattice version of that is the  restriction of  $\Phi(\cdot)$ to  the vertices of a finite graph with the vertex set 
\be \label{eq:a_R} 
\mathcal V_{a,R} = (a \Z)^d \cap \Lambda_R,  \qquad \Lambda_R :=[-R,R]^d\,.
\ee 
For the corresponding finite collection of variables $\{\Phi(x)\}_{x\in \mathcal V_{a,R}}$ the Hamiltonian \eqref{eq:FI} is initially interpreted in terms of the  Riemann-sum style discrete analog of the integral expressions.  Moments of $\Phi(x)$ are to be accompanied by lower order counter-terms.  In particular,  the fourth power addition takes the form
\be
 P(\Phi(x)) =\lambda \Phi^4\,-c(\lambda,a,R)\Phi^2 \,, 
\ee
The cutoffs are removed, through the limit   $R\nearrow \infty$ 
followed by $a\searrow 0$.  Parameter such as  $c(\lambda,a,R)$ are allowed to be adjusted in the process, so as to stabilize the Schwinger functions $S_n(x_1,.., x_n)$ on the continuum limit scale. 

The constructive field theory   
program  has yielded non-trivial scalar field theories over $\R^2$ and $\R^3$~\cite{BryFroSpe82,GliJaf73,GRS75,OS75}.   (Here we do not discuss here gauge field theories, 
cf.~\cite{JW}). 
However, the  progression of constructive results was halted when it was proved  that for dimensions $d>4$ the attempt to construct  $\Phi^4_d$  with 
\be\label{Phi_even}
\lim_{|x-y| \to \infty} S_2(x,y)  \ =\  0 
\ee 
by the method outlined above (in essence: taking the scaling limit of the lattice models at $\beta \leq \beta_c$) 
yields only  Gaussian fields~\cite{Aiz82,Fro82}.

Various partial results have indicated that the same may hold true for the critical dimension $d=4$  
(cf.~\cite{AizGra83, ACF83, BauBrySla14,GawKup85, HarTas87}),
 however a sweeping statement such as proved for $d>4$ has remained  open.   In this work we address this case.   

For clarity let us note that, like the no-go statements of~\cite{Aiz82,Fro82},  the results  presented here do not involve explicit computations of the counterterms along the above scheme.  Instead, they are based on dimension-dependent relations among the Schwinger functions which may emerge in any such limit.

\subsection{Statement of the main result}

The probability measures which correspond to \eqref{Phi_EV} with the lattice and finite volume cutoffs \eqref{eq:a_R} take the form of a statistical-mechanics Gibbs equilibrium state average   
\be \label{Phi_SM}
\langle F(\phi) \rangle  = \frac 1 {\rm norm}   \int  F(\phi) \exp{[-H(\phi)]}  \prod_{x\in \Lambda_R} \rho( d \phi_x) ,
\ee 
with a Hamiltonian  $H(\phi)$ and an a-priori measure $\rho(d \phi)$   of the form
\be \label{H}
H(\phi) = -\sum_{\{x,y\}\subset \Lambda_R}J_{x,y}\,\phi_x \phi_y    \, , \qquad \rho( d \phi_x) =  e^{- \lambda \phi_x^4 + b \phi^2_x}  d \phi_x \,,  
\ee 
where $d\phi_x$ is the Lebesgue measure on $\R$ and $J_{x,y}$ is zero for non-nearest neighbour vertices, and $J\ge0$ otherwise.   To keep the notation simple, the basic variables are written here as they appear from the perspective of the lattice but our attention is focused on the correlations  at distances of the order of $L$, with 
\be 
1 \ll L \ll R \, .
\ee 
In terms of the scaling limit discussed above, $a$ is equal to $1/L$.

A point of fundamental importance is that since the interaction through which the field variables are correlated is local (nearest neighbor on the lattice scale), for  the field correlations functions to exhibit non-singular variation on the scales $L\gg1$, the system's parameters $(J, \lambda, b)$ need to be very close to the critical  manifold, along which the correlation length of the lattice system diverges\footnote{The scaling limit of a correlation function with exponential decay which on the lattice scale is of a fixed correlation length results in a white noise distribution in the limit.}. 

Quantities whose joint distribution we track  in the scaling limit are based on the collections of random variables of the form
\be \label{def_Tf_scaled}
T_{f,L}(\phi) :=  \frac{1}{\sqrt{\Sigma_L}}\sum_{x\in \Z^d} f(x/L ) \, \phi_x   \,,
\ee
where $f$ ranges over compactly supported continuous functions, whose collection is denoted $C_0(\R^d)$, and $\Sigma_L$ denotes the variance of the sum of spins over the box of size $L$, i.e. 
\be
\Sigma_L:=\big\langle\big(\sum_{x\in\Lambda_L}\phi_x\big)^2\big\rangle.
\ee

\begin{definition}  A discrete system as described  above, parametrized by $(J,\lambda,b, R, L)$, {\em converges in distribution},  in the double limit $\lim_{L \to \infty} \lim_{R/L\to \infty}$ (with a possible  restriction to a subsequence along which also the other parameters are  allowed to vary) if 
for any finite collection of test functions $f\in C_0(\R^d)$ the joint distributions of the random variables $\{T_{f,L}(\phi)\}$ converge.
\end{definition} 

Through a standard probabilistic construction,  the limit can be presented as a random field $\Phi$, to whose weighted averages $T_f(\Phi)$ the above variables converge in distribution.   We omit here the detailed discussion of this point\footnote{By the Kolmogorov extension theorem, one may start by selecting  sequences of the parameter values so as to establish convergence in distribution for a countable collection of test functions $f$, which is dense in $C_0(\R^d)$, and then use the uniform local integrability of the rescaled correlation function and of the limiting Schwinger functions, to extend the statement by continuity arguments to all $f\in C_0(\R^d)$.  One may then  recast the limiting variables as associated with a single random $\Phi$, as in \eqref{eq:functional}.}, 
but remark that for the models considered here the construction  is simplified by i) the exclusion of  delta functions $\delta (x)$ and their derivatives from the family of considered test functions, and ii) the uniform local integrability of the  rescaled correlation functions (before and at the limit).  This important  condition is implied in the present case by the \emph{infrared bound}, which is presented below in Section~\ref{sec:3.2}.

Our main result concerning the euclidean field theory is the following.   
   \begin{theorem}[Gaussianity of $\Phi^4_4$]\label{thm:gaussian phi4}
For  dimension $d=4$,  any random field reachable by the above constructions, and satisfying \eqref{Phi_even}, is a generalized Gaussian process. \end{theorem}

Let us mention that the precise asymptotic behaviour of scaling limits of lattice models which start from sufficiently small perturbations of the Gaussian free field, i.e.~small enough $\lambda$,  have been obtained through rigorous renormalization techniques   \cite{BauBrySla14,FMRS87,GawKup85,HarTas87}.  In comparison, our result also covers  arbitrarily ``hard'' $\phi^4$ fields.  However, we do not currently provide comparable analysis of the convergence in terms of the exact scale of the logarithmic corrections, and the exact expression for the covariance of the limiting Gaussian field.

Let us also note that what from the perspective of constructive field theory may be regarded as disappointment is a positive and constructive result  from the perspective of statistical mechanics.   The theoreticians' goal there is to understand the critical behavior in models which lie beyond the reach of exact solutions.   The proven gaussianity of the limit  is therefore also a constructive result.

\subsection{The statistical mechanics perspective} 

Statistical mechanics provides a general approach for studying the behaviour of extensive systems of a divergent number of degrees of freedom.  Among the theoretically gratifying observations in this field  has been the discovery of ``universality''.  The term means that some of the key features of phase diagrams, and critical behavior (including the critical exponents),  appear to be the same across broad classes of systems of rather different microscopic structure.   This has accorded 
 relevance to studies of the phase transitions in drastically streamlined mathematical models. 
The  ferromagnetic Ising spin model 
 to which we turn next are among  the earliest, and most studied  such systems.  

An intuitive explanation of universality is that the large scale behavior of models of rich short scale structure  is described by statistical field theories for which there are far fewer  options.    A heuristic perspective on this phenomenon is provided by the renormalization group theory, c.g.  ~\cite{Wil71}.   In particular, the mechanism underlying the simplicity of the scaling limit is related to simplicity of the critical exponents, which means that for $d\geq 4$ they assume their mean field values. Rigorous results for the latter (though still partial, in terms of logarithmic corrections) were presented in \cite{Sok79, AizFer86}.  

The Ising  spin model on $\Lambda \subset \Z^d$ has as its basic variables a collection of $\pm 1$ valued variables $\{\sigma_x\}_{x\in \Lambda}$, and a Hamiltonian (the energy function)  of the form
\be \label{Ising_H}
H_{\Lambda,J,h}(\sigma):=-\sum_{\{x,y\}\subset \Lambda}J_{x,y}\,\sigma_x\sigma_y \  - \  \sum_{x\in \Lambda} h \sigma_x \, .   
\ee 

The model's finite volume Gibbs equilibrium state $\langle \cdot\rangle_{\Lambda,J,h,\beta}$ at {\em inverse temperature} $\beta\ge0$ is the probability measure under which the expectation value of any function  $F:\{\pm1\}^\Lambda\rightarrow \R$ is  given by
\begin{equation}
\label{eq:Gibbs}
\langle F\rangle_{\Lambda,J,h,\beta}:=\frac{1}{Z(\Lambda,J,h,\beta)}\sum_{\sigma\in\{\pm1\}^\Lambda}F(\sigma)\exp[-\beta H_{\Lambda,J,h}(\sigma)],  \end{equation}
where the normalizing factor $Z(\Lambda,J,h,\beta)$ is the model's {\em partition function}. 
Infinite volume Gibbs states on $\Z^d$, which we shall denote by $\langle\cdot\rangle_{J,h,\beta}$,  are defined through suitable limits (over  sequences $\Lambda_n \nearrow \Z^d$) of the above.

We focus here on the {\em  nearest neighbor ferromagnetic} interaction (n.n.f.) 
\be
J_{x,y}   = \begin{cases}   J & \|x-y\| =1 \\  
0 & \mbox{otherwise} 
\end{cases}
\ee 
 with $J>0$.  In dimensions $d>1$, this model exhibits a line of first-order phase transitions (in the  plane of the model's thermodynamics parameters $(\beta,h)$) along the line $h=0$, $\beta \in (\beta_c(d),\infty)$.   The line terminates at the critical point 
 $(\beta_c,0)$ at which the model's correlation length diverges.   Our discussion concerns the scaling limits at, or near, this point.   Since the phase transition occurs at zero magnetic field, we restrict the discussion to $h=0$ and will omit $h$ from the notation.


Away from the critical point the model's truncated correlation functions  
decay exponentially fast \cite{AizBarFer87,DumGosRao18}.  This leads to the definition of the {\em correlation length} $\xi(\beta)$ as:
\be 
\xi(\beta):=\lim_{n\rightarrow\infty}-n/\log \langle\sigma_0;\sigma_{n{\bf e}_1}\rangle_{\beta}\quad \mbox{(with ${\bf e}_1 = (1,0,...,0)$)}. 
\ee
The correlation length is proven to be finite for any $\beta < \beta_c$ \cite{AizBarFer87} and divergent in the limit $\beta \nearrow \beta_c$~\cite{Sim80}.    At the critical point $\xi(\beta_c)=+\infty$ 
 as  the decay of the 2-point function slows to a power-law (see \cite{Sim80} and  the  discussion around Corollary~\ref{cor:SL}).   
 
At this point, one may notice the similarity between the Ising model's Gibbs equilibrium distribution \eqref{eq:Gibbs} and  the discretized functional integral \eqref{Phi_SM}.  Furthermore,  in view of the probability measures' relation
\be \label{Phi4_Ising}
 \frac{1}{2} [\delta(\sigma-1) + \delta(\sigma+1)] \, d\sigma  = 2 \lim_{\lambda\to \infty} e^{-\lambda (\phi^2-1)^2} d \phi / \text{Norm}(\lambda)
 \ee   
 the Ising spin's a-priori (binary) distribution can be viewed as the ``hard''  limit of the $\phi^4$ measure. 
Hence included in Theorem~\ref{thm:gaussian phi4} is the statement that for $d=4$ any scaling limit of the critical Ising model is Gaussian.  

 However, our analysis flows in the opposite direction.  In essence, the argument is structured as follows
 \begin{enumerate} 
 \item   deploying methods which take advantage of the Ising systems' structure,  the stated results are first  proven for the n.n.f. Ising model (in four dimensions);
 \item the analysis is  adapted to the model's extension, in which each spin is  replaced by a block average of `elemental' Ising spins with an intrablock ferromagnetic coupling;
 \item through  weak limits the statement is extended to systems of variables whose a-priori single spin distribution belongs to the Griffiths-Simon (G-S) class.   \end{enumerate} 
Included in the G-S  class (defined below) are the $\Phi^4$ measure 
$\rho(d \varphi) $ of \eqref{H}.
\medbreak

To reduce the repetition,  some of the relevant relations are presented below in a form which may not be the simplest for n.n.f. but is suitable  for the model's  generalized version.   However  in the rest of this section we focus on the n.n.f. case. 

%

As it is known, and made explicit in Section~\ref{sec:4.3}, for Ising models a bellwether for Gaussian behaviour at large distances  is the asymptotic validity of Wick's law at the level of the four-point  function~\cite{Aiz82,New75}.   The deviation is expressed in the  {\em Ursell function} 
\begin{equation} 
U_4^{\beta}(x,y,z,t)  := \langle\sigma_x\sigma_y\sigma_z\sigma_t\rangle_\beta - \left[ \langle\sigma_x\sigma_y\rangle_\beta\langle\sigma_z\sigma_t\rangle_\beta+\langle\sigma_x\sigma_z\rangle_\beta\langle\sigma_y\sigma_t\rangle_\beta+\langle\sigma_x\sigma_t\rangle_\beta\langle\sigma_y\sigma_z\rangle_\beta    \right] 
\end{equation}
the relevant  question being whether      
$ U_4(x,y,z,t)/ \langle\sigma_x\sigma_y \sigma_z\sigma_t\rangle_\beta$   vanishes asymptotically for quadruples of sites at large distances, of comparable order between the pairs.

Gaussianity of the scaling limits for $d>4$ was previously established  through the combination of the {\em tree diagram bound} of~\cite{Aiz82}:
 \begin{align} |U_4^\beta(x,y,z,t)|  &\le  2  \sum_{u\in\Z^d} \langle \sigma_u\sigma_x\rangle_\beta  
 \langle \sigma_u\sigma_y\rangle_\beta  
  \langle \sigma_u\sigma_z\rangle_\beta     
   \langle \sigma_u\sigma_t\rangle_\beta \label{tree} 
\end{align} 
and the \emph{Infrared Bound} of~\cite{FroSimSpe76,GliJaf73} 
\be
\langle \sigma_x \sigma_y\rangle_{\beta_c}  \leq  \frac{C}{|x-y|^{d-2}}\,.
\ee  

At the heuristic level,  the triviality of the scaling limit for $d>4$ is indicated by the following dimension counting.  Assume that at $\beta_c$ the two-point function is of comparable values for pairs of sites at similar distances (which is false for $\beta \neq \beta_c$ at distances much larger than $\xi(\beta)$).  Then, for  quadruples of points at mutual distances of order $L$, the sum in the tree diagram bound \eqref{tree}  contributes a factor $L^d$ while the summand has two extra correlation function factors, in comparison to 
$\langle\sigma_x\sigma_y\sigma_z\sigma_t\rangle_\beta
$, each factor dominated by $1/L^{d-2}$.  
This suggests that $U_4(x,y,z,t)$ in comparison to the full correlation functions may be of the order $O(L^{4-d})$, which for $d>4$   vanishes in the limit $L\to \infty$.   Up to numerous technical details this is the essence of the argument presented  in \cite{Aiz82,Fro82}.  
However, the above estimate is clearly  inconclusive for $d=4$.


The key advance presented here is the following improvement of the tree diagram bound.  The multiplicative  factor by which it improves ~\eqref{tree} is derived through a multi scale analysis which is of relevance at the marginal dimension $d=4$.  

\begin{theorem}[Improved tree diagram bound inequality]\label{thm:improved tree bound simple} For the n.n.f.~Ising model in dimension $ d=4$, there exist $c,C>0$ such that for every $\beta\le \beta_c$, every $L\le \xi(\beta)$ and every $x,y,z,t\in\Z^d$ at a distance larger than $L$ of each other,
\be\label{eq:improved tree bound}
|U_4^\beta(x,y,z,t)|\le \frac C{B_L(\beta)^c}\sum_{u\in\Z^4} \langle \sigma_u\sigma_x\rangle_\beta  
 \langle \sigma_u\sigma_y\rangle_\beta  
  \langle \sigma_u\sigma_z\rangle_\beta     
   \langle \sigma_u\sigma_t\rangle_\beta,
\ee
where $B_L(\beta)$ is the {\em bubble diagram} truncated at a distance $L$ defined by the formula
\be \label{B_L}
B_L(\beta):=\sum_{x\in\Lambda_L}\langle\sigma_0\sigma_x\rangle_\beta^2.
\ee
\end{theorem}

For a heuristic insight on the implications of this improvement for $d=4$, one may consider separately the two following scenarios:   the two-point function $\langle\sigma_0\sigma_x\rangle_\beta$ may be  roughly of the order  $L^{2-d}$ (meaning that the Infrared Bound is saturated up to constant), or it may be much smaller.  In the first case (which is conjectured to hold when $d=4$),  $B_L(\beta)$ is of order $\log L$, so that the improved tree diagram bound indicates that  $|U_4|/S_4 = O(\log L)^{-c}$, and thus is asymptotically negligible.  In the second case (which is not the one expected to hold), already the unadulterated  tree diagram bound \eqref{tree} suffices.

We derive \eqref{eq:improved tree bound} making extensive use of  the Ising model's random current representation that is presented in Section~\ref{sec:1.5}.  It enables combinatorial identities 
through which the   
deviations from Wick's law can be expressed in terms of intersection probabilities of the random clusters which link pairwise the specified source points.  

Beyond the four point function, the full statement of the scaling limit's gaussianity  is established here through the following estimate of the characteristic function of smeared averages of spins. 

\begin{proposition}\label{prop:gaussian b}
There exist $c,C>0$ such that for the n.n.f.~Ising model on $\Z^4$, every $\beta\le\beta_c$, every $L\le\xi(\beta)$, and  test function $f \in C_0(\R^4)$,  
\be \label{eq:gauss}
 \Big|\,\big\langle \exp[z\,T_{f,L}(\sigma)-\tfrac{z^2}2\langle T_{f,L}(\sigma)^2\rangle_\beta ]\big\rangle_\beta\,-\,1\,\Big|~\le~ \frac{C\|f\|_\infty^4r_f^{12}}{(\log L)^c}\,z^4,
\ee
with $\| f\|_\infty := \max\{|f(x)|:x\in \R^4\}$ and $r_f$ the diameter of the function's support. 
\end{proposition}

The claimed gaussianity follows since  (by the Infrared Bound, applied on the left-hand side) for any non-negative continuous function $f \nequiv 0$ with bounded support,
\be
C\,r_f^2 \| f\|^2_\infty   
 \geq  \langle T_{f,L}(\sigma)^2\rangle_\beta  \ge  c_f>0,
\ee 
uniformly in $\beta\le\beta_c$ and $L$, 
 we get that  for $L\gg 1$ the distribution of 
$T_{f,L}(\sigma)$ is approximately   Gaussian  of variance $\langle T_{f,L}(\sigma)^2\rangle_\beta$.

 \paragraph{Organization of the proof:}
 The result proven here is unconditional.   However, to better convey the argument's structure, we  first establish the claimed result for the scaling limits of critical models $(\beta=\beta_c)$  under the auxiliary assumption that the two-point function behaves regularly on all scales, in  a sense defined below.  We then present an unconditional proof for $\beta\le\beta_c$ in which we add to the above analysis the proof that the two-point function is regular on a sufficiently large collection of distance scales, up to the correlation length $\xi(\beta)$.   
 
 \paragraph{Organization of the article:}  In the next section, we present the Griffiths-Simon construction of random variables which can be obtained as local aggregates of ferromagnetically coupled Ising spins.  It yields a  useful  link  between the $\phi^4$ and Ising variables.
 Following that, in Section~\ref{sec:1.5} we present the basics of Ising models' random current representation, and  the intuition based on random walk intersection probabilities.   Section~\ref{sec:2} contains a conditional proof of the improved tree diagram bound at criticality,  derived under a  power-law decay assumption on the two-point function.   Next, as a preparation for the unconditional proof, in Section~\ref{sec:3} we present some relevant properties of Ising model's two-point function. These estimates are stated and proved in the context of systems of real valued variables with the single-spin distribution in the afore mentioned Griffiths-Simon class. Included there are mostly known but also some new results.   Section~\ref{sec:4} contains the unconditional proof of our main results for the Ising model.  Section~\ref{sec:6} is devoted to its extension to the Griffiths-Simon class.   The appendix contains some auxiliary technical statements that are of independent interest.
  
\section{The Griffiths-Simon class of measures} \label{sec:GS}

The discrete approximations of the $\varphi^4$ functional integral and the Gibbs states of an Ising model are not only  analogous, as explained above, but are actually related.    

In one direction one has  \eqref{Phi4_Ising} and the implications mentioned next to it.  However, in this work we shall make use of another relation, which permits us to apply tools which are initially developed for general Ising models to the study of the $\varphi^4$ functional integral.    This relation is based on a construction which was initiated by  Griffiths~\cite{Gri69}, and advanced further by Simon-Griffiths~\cite{SimGri73}.     

\begin{definition} \label{def:rho}
A probability measure on  $\rho(d\varphi) $ on $\R$ is said to belong to the Griffiths-Simon (GS) class if either  of the following conditions is satisfied 
\\  
\indent $1)$   the expectation values with respect to $\rho$ can be presented as 
\be 
\int F(\varphi) \, \rho(d\varphi)  =   \sum_{\underline{\sigma}\in \{-1,1\}^N} F(\alpha \sum_{n=1}^N b_n \sigma_n) e^{\sum_{n,m=1}^N  K_{n,m} \sigma_n \sigma_m} / \text{Norm.}
\ee 
with some $\{b_n\}\subset \R$, and $K_{n,m}  \geq 0$.   \\  
\indent $2) $   $\rho$ can be presented as a (weak) limit of probability measures of  the above type, and is of sub-gaussian growth: 
\be \label{sub_gauss}
\int e^{|\varphi |^\alpha} \rho(d \varphi ) < \infty \quad \mbox{for some $\alpha> 2$.}
\ee 
\end{definition} 
A random variable is said to be of  Griffiths-Simon type if its probability distribution is in the GS class.

The construction  (1) was employed  by 
Griffiths~\cite{Gri69} 
for a proof  that the Ising model's Lee-Yang property as well as the Griffiths correlation inequalities hold also for a broader class of similar models with other notable spin variables.  Subsequently, Simon and Griffiths~\cite{SimGri73} pointed out that upon taking weak limits this can be extended to cover alsothe $\phi^4$ a-priori  measures, spelled in \eqref{H}.  

More specifically, a finite collection of the variables $\{\varphi_x\}_{x\in \Lambda}$ with the a-priori measure $\rho(d\varphi) = e^{-\lambda \varphi^4 + b\varphi^2} d\varphi/\text{norm} $   
can be produced as the $N\to \infty$ limit \emph{(in distribution)}  of the collection of the block averages of elemental Ising spins $\{\sigma_{x,n}\}$ (the dots  in Fig.~\ref{fig_blocks} )
\be
\varphi_x^{(N)} =   \alpha_N(\lambda, b) \sum_{n=1}^N \sigma_{x,n}
\ee 
 under the ``ultra-local'' coupling (which is to be added to the intersite interaction $H$ of \eqref{H})  
\be
H_{\rm inner}  = -  \frac{g_N(\lambda,b)}{N }  \sum_{x\in \Lambda} 
\sum_{n,m} \sigma_{x,n} \sigma_{x,m}
\ee 
with suitably adjusted $(\alpha_N, g_N)$.  Their exact values are not important for our discussion, but let us note that   $H_{\rm inner}$ is a mean field interaction and thus it is easy to see that for each 
$(\lambda,b)$ with $\lambda \neq 0$: 
$g_{N} (\lambda,b)$ tends to $1$ as $N$ tends to infinity, at a $(\lambda,b)$ dependent rate.  

In this representation, any system 
of $\phi^4$ variables  associated with the sites of a graph $\mathcal V$, and coupled through the graph's edges, is presentable as the limit ($N\to \infty$) of a system of constituent Ising spins associated withe the  Cartesian graph product   $\Z^d \times \mathcal {K}_N$, with $\mathcal {K}_N$ denoting the complete graph of $N$ vertices.

\begin{figure}
\begin{center}
\includegraphics[width = 0.40\textwidth]{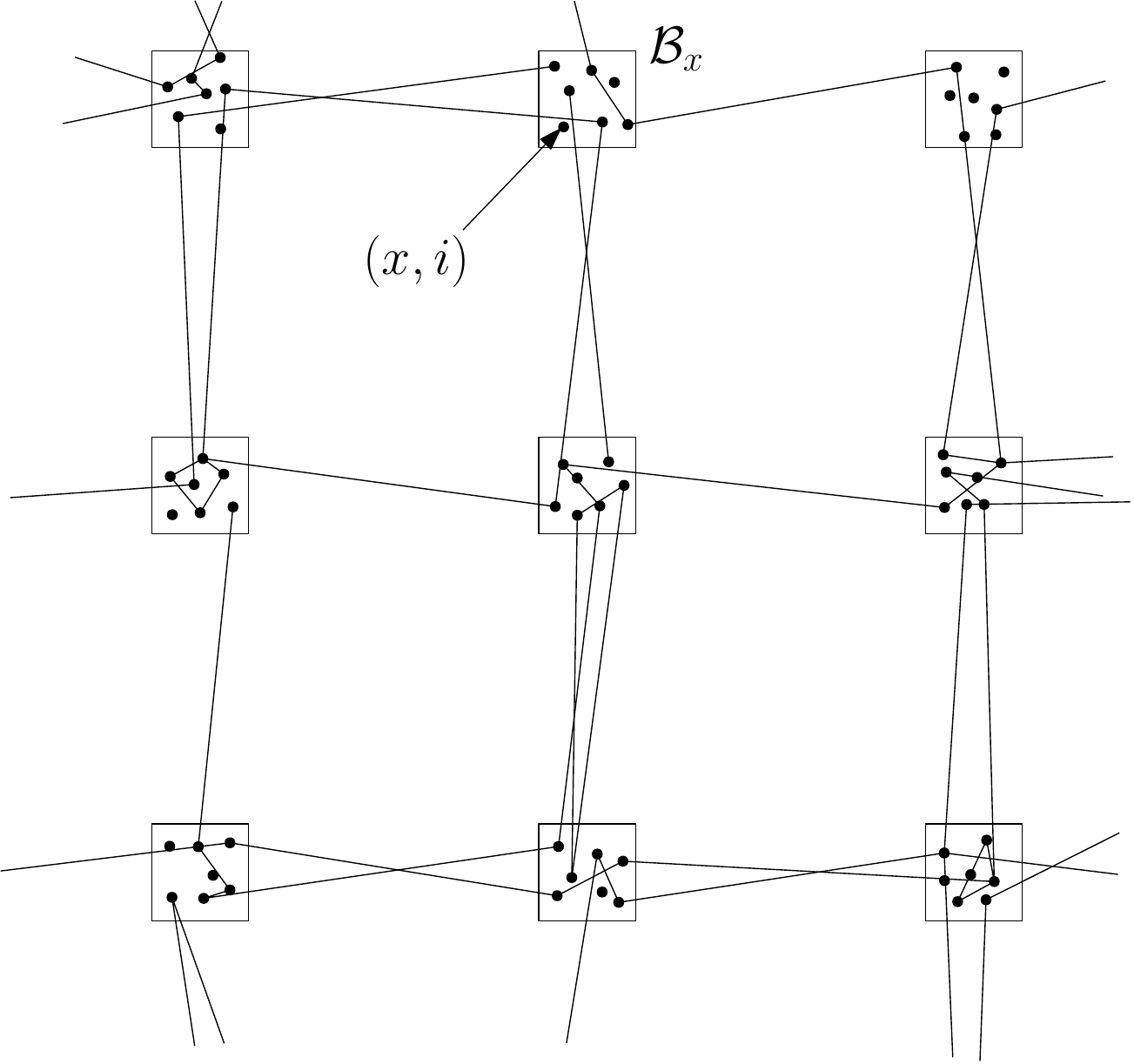}
\caption{The decorated graph, in which the sites $x\in \Lambda$  of a graph of interest are replaced by ``blocks'' 
$\calB_x$ of sites indexed as $(x,n)$.   The Ising ``constituent spins'' $\sigma_{x,n}$ are coupled pairwise through intra-block  couplings $ \delta_{x,y} K_{n,m}$ and inter-block couplings $J_{x,y}$.  The depicted lines indicate a  possible realization of the corresponding random current.}
\label{fig_blocks}
\end{center}
\end{figure}

\section{Random current intersection probabilities}\label{sec:1.5}

\subsection{Definition and switching lemma}

Starting with the Ising model, in this section we briefly introduce its random current representation, which allows to express the model's subtle correlation effects  in more tangible stochastic geometric terms.   The utility of the random current representation is enhanced by the combinatorial symmetry  expressed in its \emph{switching lemma}, which enables to structure some of the essential truncated correlations in terms guided by the analysis of the intersection properties of the traces of random walks.

\begin{definition}A {\em current} configuration $\n$ on $\Lambda$ is an integer-valued function defined over unordered pairs $(x,y)\in \Lambda$.  
The current's set of {\em sources} is defined as the set
\be 
\partial\n \, := \, \{ x\in \Lambda :\, (-1)^{\sum_{y\in \Lambda} \n(x,y)} = -1 \}.
\ee  
For a given Ising model on $\Lambda$, we associate to a current configuration  the {\em weight} 
\be 
w(\n) =w_{\Lambda,J,\beta}(\n) :=\prod_{\{x,y\}\subset V}\frac{\displaystyle(\beta J_{x,y})^{\n(x,y)}}{\n(x,y)!}\,. 
\ee
\end{definition}

Starting from Taylor's expansion
\be
\exp(\beta J_{x,y}\sigma_x\sigma_y)=\sum_{\n(x,y)\ge0} \frac{(\beta J_{x,y}\sigma_x\sigma_y)^{\n(x,y)}}{\n(x,y)!},
\ee
one can see that the Ising model's partition function (defined below~\eqref{eq:Gibbs}) can be expressed in terms of the corresponding random current: 
\begin{equation}\label{eq:8}
Z (\Lambda,\beta)=2^{|\Lambda|}\sum_{\n:\partial\n=\emptyset} w(\n).
\end{equation}
Furthermore, the spin-spin correlation functions can be represented as 
\begin{equation}\label{eq:erg}
\langle\prod_{x\in A}\sigma_x\rangle_{\Lambda,\beta}=\frac{\displaystyle\sum_{\n:\, \partial\n=A}w(\n)}{\displaystyle\sum_{\n:\partial\n=\emptyset}w(\n)}\,.
\end{equation}
At this point, it helps to note that any configuration with $\partial\n=\emptyset$, i.e.~without sources, can be viewed as the edge count of a multigraph which is decomposable into a union of loops.  In contrast, any configuration with $\partial\n=A$, such as the one appearing in the numerator of \eqref{eq:erg}, can be viewed as describing the edge count of a multigraph which is decomposable into a collection of loops and of paths connecting pairwise the sources, i.e. sites of $A$.  
In particular, 
a configuration with 
$\partial\n = \{u,v\}$ can be viewed as giving the ``flux numbers'' of a family  of loops together with a path from $u$ to $v$.  Thus, the random current representation allows to present the spin-spin correlation as the effect on the partition function of a loop system with the addition of a path linking the two sources.   In these terms, the spin-spin correlation 
$\langle  \sigma_{ x_1} \cdots \,\sigma_{ x_{2n}} \rangle_\beta$  represents the sum of the multiplicative effect of the introduction of $n$ paths pairing the sources.

Connectivity properties of currents   play a significant role in our analysis.  To express those we shall employ  the following terminology and notation.  
\begin{definition}
i) We say that $x$ is {\em connected} to $y$ (in $\n$), and denote the event by $x\stackrel{\n}\leftrightarrow y$,  if there exists a path of vertices $x=u_0,u_1,\dots, u_k=y$ with $\n(u_i,u_{i+1})>0$ for every $0\le i<k$.  We say that $x$ is connected to a set $S$ if it is connected to a vertex in $S$.  \\ 
ii) The {\em cluster} of $x$, denoted by ${\bf C}_\n(x)$, is the set of vertices connected to $x$ in $\n$.\\ 
iii) For a set of vertices $B$, we denote by $\calF_B$  the set of $\n$ satisfying that  there exists a sub-current  $\m\le \n$ such that $\partial\m=B$.
\end{definition}
Some of the most powerful properties of the random current representation are best seen when considering pairs of random currents and using the following lemma.  

\begin{lemma}[Switching lemma]\label{lem:switching} For any $A,B\subset \Lambda$ and any function $F$ from the set of currents into $\bbR$, 
\begin{equation}\sum_{\substack{
\n_1:\partial\n_1=A\\
\n_2:\partial\n_2=B}}F(\n_1+\n_2)w(\n_1)w(\n_2)=\sum_{\substack{
\n_1:\partial\n_1=A\Delta B\\
\n_2:\partial\n_2=\emptyset}}F(\n_1+\n_2)w(\n_1)w(\n_2)\bm 1_{{\n_1+\n_2}\in \calF_B}.
\end{equation}
where $A\Delta B $ denotes the symmetric difference of sets, $A\Delta B:=(A\setminus B)\cup(B\setminus A)$. 
\end{lemma}  

The switching lemma appeared as a combinatorial identity in Griffiths-Hurst-Sherman's derivation of the GHS inequality~\cite{GHS70}.  Its greater potential  for the geometrization of the correlation functions was developed in \cite{Aiz82}, and works which followed. In this paper, we employ two generalizations of this useful identity.  In the first,  the two currents $\n_1$ and $\n_2$ need not be defined on the same graph (see \cite[Lemma~2.2]{AizDumSid15} for details).  The second 
 will involve a slightly more general switching statement, which was used in several occasions in the past (cf.  \cite[Lemma 2.1]{AizDumTasWar18}  and reference therein).
 
It should be recognized that other stochastic geometric representations of spin correlations and/or interactions can be found (e.g.~the Symanzik representation of the $\phi^4$ action~\cite{Sym69}, and the BFS random walk representation of the correlation functions \cite{BryFroSpe82}).  It is conceivable that the overall strategy could be applied also through other means.   However we find the random current representation particularly useful for our purpose.

\subsection{Representation of Ursell's four-point function}

The switching lemma enables one to rewrite spin-spin correlation ratios in terms of probabilities of  events expressed in terms of the random currents.  The first of these is the relation  
\begin{align}\label{eq:orgaf}\frac{\langle\sigma_A\rangle_{\Lambda,\beta}\langle\sigma_B\rangle_{\Lambda,\beta}}{\langle\sigma_A\sigma_B\rangle_{\Lambda,\beta}}&:={\bf P}^{A\Delta B,\emptyset}_{\Lambda,\beta}[\n_1+\n_2\in \calF_B],\end{align}
for which we denote by 
$\mathbf P_{\Lambda,\beta}^{A} \left(\n \right)$ the probability distribution on random currents constrained by the source condition $\partial \n = A$, or more explicitly   
\begin{equation} \label{eq:Pr}
\mathbf P_{\Lambda,\beta}^{A} \left(\n \right)  :=   
\frac{2^{|\Lambda|} w(\n)} {\langle \prod_{x\in A} \sigma_x \rangle_{\Lambda,\beta} \, Z(\Lambda,\beta) }  \bbI [\partial \n = A],
\end{equation} 
and by $\mathbf P_{\Lambda,\beta}^{A_1,\dots,A_i} $ we denote the  law of an independent family of currents $(\n_1,\dots,\n_i)$
\begin{equation} \label{eq:2Pr}
\mathbf P_{\Lambda,\beta}^{A_1,\dots,A_i} :=  \mathbf P_{\Lambda,\beta}^{A_1} \otimes\dots\otimes \mathbf P_{\Lambda,\beta}^{A_i}. \end{equation}  
For two-point sets we may write $A=xy$ instead of $\{x,y\}$.  

As we will  also work  with the infinite volume Gibbs measures, let us note that random currents and the switching lemma admit a generalization to infinite volume\footnote{The extension of the switching lemma to $\Z^d$ is straightforward for $\beta\le\beta_c$ since then $\n_1+\n_2$ does not contain infinite paths of positive currents, almost surely under ${\bf P}^{A,B}_\beta$.  For $\beta<\beta_c$ this is implied by the discussion of  \cite{Aiz82} for $\beta<\beta_c$, and for  $\beta=\beta_c$ it follows from the continuity result of \cite{AizDumSid15} for $\beta=\beta_c$.}.  
Existing continuity results \cite{AizDumSid15} permit to extend  \eqref{eq:orgaf} to  the infinite volume, expressed in terms of the weak limits of the random current measures ${\bf P}^A_{\Lambda_n,\beta}$ and ${\bf P}^{A_1,\dots,A_i}_{\Lambda_n,\beta}$,  in the limit $\Lambda_n \nearrow \bbZ^d$.   The limiting statement is similar to \eqref{eq:orgaf} but without   the finite volume subscript $\Lambda$:   
\begin{align}\label{eq:orga}\frac{\langle\sigma_A\rangle_{\beta}\langle\sigma_B\rangle_{\beta}}{\langle\sigma_A\sigma_B\rangle_{\beta}}&={\bf P}^{A\Delta B,\emptyset}_{\beta}[\n_1+\n_2\in \calF_B].\end{align}
Combining \eqref{eq:orga} for the different values of the product of spin-spin correlations leads to
\begin{equation}\label{eq:U4} U_4^{\beta}(x,y,z,t)=-2\langle \sigma_x\sigma_y\rangle_{\beta}\langle \sigma_z\sigma_t\rangle_{\beta}\,{\bf P}_{\beta}^{xy,zt}[{\mathbf C}_{\n_1+\n_2}(x)\cap {\mathbf C}_{\n_1+\n_2}(z)\ne \emptyset]\, .\end{equation}
This equality is of fundamental importance to the question discussed here.  It  was the basis of the analysis of \cite{Aiz82}, and is the starting point for our discussion.

By \eqref{eq:U4},  the relative magnitude of the deviation of the four-point function $\langle\sigma_x\sigma_y\sigma_z\sigma_t\rangle_\beta$ from the Gaussian law (i.e.~the discrepancy in Wick's formula) 
is  bounded in terms of intersection properties of the two clusters that link the indicated sources pairwise:
\be \label{U/S}
\frac{{|U^{\beta}_4(x,y,z,t)|}}{\langle\sigma_x\sigma_y\sigma_z\sigma_t\rangle_\beta}\ \leq \ 
2 {\bf P}_{\beta}^{xy,zt}[{\mathbf C}_{\n_1+\n_2}(x)\cap {\mathbf C}_{\n_1+\n_2}(z)\ne \emptyset]\, .
\ee 

The random sets  ${\mathbf C}_{\n_1+\n_2}(x)$ and $ {\mathbf C}_{\n_1+\n_2}(z)$ are not independently distributed.  However \eqref{U/S} can be further simplified through a  monotonicity property of random currents.  As  proved in \cite{Aiz82},  and recalled here in the Appendix, the probability of an intersection can only increase upon the two sets' replacement by a pair of  independently distributed  clusters defined through the addition of two sourceless currents: 
\be {\bf P}^{xy,zt}_\beta[{\mathbf C}_{\n_1+\n_2}(x)\cap {\mathbf C}_{\n_1+\n_2}(z)\ne \emptyset]\le {\bf P}^{xy,zt,\emptyset,\emptyset}_\beta[{\mathbf C}_{\n_1+\n_3}(x)\cap {\mathbf C}_{\n_2+\n_4}(z)\ne \emptyset]\,.
\ee
This leads to the simpler upper bound in which the two random sets are independent:
 \begin{align}\label{eq:inter} |U_4^\beta(x,y,z,t)|   & \leq    2\langle \sigma_x\sigma_y\rangle_\beta\langle \sigma_z\sigma_t\rangle_\beta{\bf P}^{xy,zt,\emptyset,\emptyset}_\beta[ {\mathbf C}_{\n_1+\n_3}(x)\cap{\mathbf C}_{\n_2+\n_4}(z)\ne \emptyset] \,.  
 \end{align}

Bounding the intersection probability by the expected number of intersection sites and applying the switching lemma leads directly to the tree diagram bound \eqref{tree}.   However, as was explained above, to tackle the marginal dimension $d=4$ one needs to improve on that.  

While ${\mathbf C}_{\n_1+\n_3}(x)$ and ${\mathbf C}_{\n_2+\n_4}(z)$ are bulkier and exhibit less independence  than simple random walks linking the sources $\{x,y\}$ and $\{z,t\}$, the analogy is of help in guiding the intuition towards useful estimate strategies.  
 In particular, it is classical that in dimension  $d=4$    the  probability that the traces of two random walks  starting at  distance 
 $L $ of each other intersect,  tends to $0$ (as $1/\log L$, see \cite[(2.8)]{Aiz85} and \cite{Law13}), but  nevertheless  the expected number of  points of  intersection remains of order $\Omega(1)$.    The discrepancy is explained by the fact that although the  intersections occur rarely, the conditional expectation of the number of intersection sites, conditioned on there being at least one, diverges logarithmically in $L$.  The thrust of our analysis will be to establish  similar behaviour in the system considered here.  More explicitly, we will prove that the conditional expectation of the clusters' intersection size, conditioned on it being non-empty, grows at least as  $(\log L)^c$. 

 The analysis of clusters' intersection properties is  more difficult than that of the paths of simple random walks for at least two reasons:
  \begin{itemize}
 \item Missing information on the two-point function: Most analyses of intersection properties of random walks involve 
 estimates on the Green function.  In our system its role is to some extent taken by the two-point spin-spin correlation function.  However, unlike the former case we do not a priori know the two-point function's exact order of magnitude  (though a good one-sided inequality is provided by the Infrared Bound).  This  raises a difficulty  that we address  by studying the regularity properties of the two-point function in Section~\ref{sec:3}. 
 \item The lack of a simple Markov property:
in one way or another, the analysis of intersections for random walks involves the random walk's Markov property.  Among its other applications, the walk's renewal property facilitates de-correlating the walks' behaviour at different places.   
In comparison, the random current clusters exhibit only a multidimensional domain Markov property.   
 One of the main contributions of this paper will  be to show a mixing property of random currents which will enable us to bypass the difficulty raised by the lack of a renewal property. 
  \end{itemize}
We expect that both the regularity estimates and the mixing properties established here 
are of independent interest, and may  be of help in studies of the model also in three dimensions.

 \section{A conditional improvement of the tree diagram bound for $\beta = \beta_c$}
\label{sec:2}

To better convey the strategy by which the tree diagram bound is improved,  we start with a \emph{conditional proof} of \eqref{eq:improved tree bound} for the Ising model on $\Z^4$ at criticality (i.e.~when $\beta=\beta_c$),  under the following assumption on  the model's two-point function. 
The removal of this assumption will raise substantial problems which are presented in the sections that follow.  
Below, $|\cdot|$ denotes the infinity-norm
\be
|x|:=\max\{|x_i|,1\le i\le d\}.
\ee

\begin{assumption}[Power-law decay] 
There exist $\eta$ and $c,C\in(0,\infty)$ such that for every $x\in \Z^d$,
\begin{align}\label{eq:S}
\frac{c}{|x|^{d-2+\eta}}\le \langle\sigma_0\sigma_x\rangle_{\beta_c}\le \frac{C}{|x|^{d-2+\eta}}.
\end{align}
\end{assumption}
The Infrared Bound \eqref{eq:IB} guarantees that $\eta\geq 0$ in any dimension $d>2$.
Note that if $\eta>0$ for $d=4$, then $B_L(\beta_c)$ is bounded uniformly in $L$ in which case the tree diagram bound implies the improved one.  Thus, under this assumption the case requiring attention is just  $\eta=0$ (which is the generally expected value).

\subsection{Intersection clusters} \label{sec:2.1}  

Our starting point is \eqref{eq:inter} in which $U_4^{\beta_c}$ is bounded by the probability of intersection of two independently distributed clusters ${\bf C}_{\n_1+\n_3}(x)$ and ${\bf C}_{\n_2+\n_4}(z)$, of which $\n_1$ and $\n_2$  include paths linking pairwise widely separated sources,  $\partial \n_1=\{x,y\}$ and $\partial \n_2= \{ z,t\}$.  Introduce the notation 
\be 
\calT := {\bf C}_{\n_1+\n_3}(x) \cap  {\bf C}_{\n_2+\n_4}(z),
\ee  and let $|\calT|$ be the set's cardinality. 
The tree diagram bound corresponds to  the first moment estimate: 
\be 
{\bf P}^{xy,zt,\emptyset,\emptyset}_{\beta_c}[|\calT|>0] \leq {\bf E}^{xy,zt,\emptyset,\emptyset}_{\beta_c}[|\calT|],
\ee  
in which  the intersection probability is bounded  by the intersection set's expected size.  

Although the  set $\calT$  is less tractable than the intersection of a pair of  Markovian random walks, their intuitive example provides a useful guide.  The intersection of the traces of two simple random walks in dimension $d=4$  has a Cantor-set like structure.   Guided by this analogy, and taking advantage of the switching lemma, we show that  conditioned on the event that  $u$ belongs to $\calT$,  the intersection  $|\calT|$ is typically very large.  This is in line with our expectation that the vertices in the intersection set occur in large (disconnected) clusters, causing  
the expected size of  $|\calT|$ to be much larger than the probability of it being non-zero.  

Below and in the rest of this article, we introduce the annulus of sizes $k\le n$ and the boundary of a box as follows:
\be
{\rm Ann}(k,n):=\Lambda_n\setminus\Lambda_{k-1}\quad\text{and}\quad\partial\Lambda_n:={\rm Ann}(n,n)
\ee
(cf. Fig.~\ref{fig:2}).

In the proof, we apply the following deterministic covering lemma, which  links the number of points in a set $\mathcal X\subset\Z^d$  with the number of concentric annuli of the form 
$u+{\rm Ann}(\ell_k,\ell_{k+1})$,  with $u\in \mathcal X$, which it takes to cover $\mathcal X$. 
To state it we denote, for any (possibly finite) increasing sequence of lengths $\mathcal L = (\ell_k)$, every $u\in \Z^d$,  and every integer $K$, 
\be 
{\bf M}_u(\mathcal X;\mathcal L, K) \:= \text{card}\{k\le K\,:\, \mathcal X\cap [u+{\rm Ann}(\ell_k,\ell_{k+1})]\neq \emptyset\}\,   
\ee 
(cf. Fig.~\ref{fig:2}).  
\begin{lemma} (Annular covering) \label{covering} 
In the above notation, for any sequence  $\mathcal L = (\ell_k)$ with  $\ell_1\ge 1$ and 
$ \ell_{k+1} \geq 2 \ell_k$
\be
|\mathcal X|\ge 2^{\min\{\mathbf M_u(\mathcal X;\mathcal L, K)/5\,:\,u\in\mathcal X\}}.
\ee
\end{lemma} 

\begin{proof}  
It suffices to show that if  $|\mathcal X|< 2^r$ for some $r$, then there exists a site $u\in \mathcal X$ for which
${\bf M}_u(\mathcal X;\mathcal L, K) <5r$.

We prove the following stronger statement: For every set $\calX$ containing the origin and every $K$, if 
$|\mathcal X\cap \Lambda_{\ell_{K}}|< 2^r$, then there exists $u\in \calX\cap \Lambda_{\ell_K}$ with $M_u(\mathcal X;\mathcal L,K)<5r$.

The assertion is obviously true for $r=1$ as one can pick $u$ to be the origin.  Next,  consider the case of $r>1$  assuming the statement holds for all smaller values.  If the intersection of $\mathcal X$ and $\Lambda_{\ell_{K-1}}$ is reduced to the origin,  then $M_0(\mathcal X;\mathcal L,K)\le 2$ (only the annuli ${\rm Ann}(\ell_l,\ell_{l+1})$ with $l$ equal to $K-1$ or $K$ can intersect $\mathcal X$) as required so we now assume that this is not the case.
Consider $0\le k\le  K-2$ maximal such that there exists $u \in\mathcal X $ with $\ell_k<|u|\le \ell_{k+1}$.  

Since $\mathcal X\cap\Lambda_{\ell_{k-1}}$ and  $\mathcal X\cap(u+\Lambda_{\ell_{k-1}})$ are disjoint  (we use that $\ell_{k}\ge 2\ell_{k-1}$), one of the two sets has cardinality strictly smaller than $2^{r-1}$. Assume first that it is $\mathcal X\cap\Lambda_{\ell_{k-1}}$. The induction hypothesis implies the existence of $v\in \mathcal X\cap\Lambda_{\ell_{k-1}}$ such that 
\be
{\bf M}_v(\mathcal X;\calL,k-1) <5(r-1).
\ee  
By our choice of $k$, every site in $\mathcal X$ is either in $\Lambda_{\ell_{k+1}}$ or outside of $\Lambda_{\ell_{K-1}}$.  This implies that only the annuli ${\rm Ann}(\ell_l,\ell_{l+1})$ with $l$ equal to $k$, $k+1$, $K-2$, $K-1$ or $K$ can intersect $\mathcal X$, so that
\begin{equation}
{\bf M}_v(\mathcal X;  \mathcal L, K) \le {\bf M}_v(\mathcal X;\mathcal L, k-1) +5 < 5r \,.  
\end{equation}
If it is $\mathcal X\cap(u+\Lambda_{\ell_{k-1}})$ which has small cardinality, simply translate the set by $u$ and apply the same reasoning. The distance between the vertex $v$ obtained by the procedure and $0$ is at most $\ell_{k-1}+\ell_k\le \ell_K$, so that the claim follows in this case as well.
\end{proof} 
 
In  the following conditional statement,  we    denote by  $\mathcal L_\alpha$ a sequence of integers defined recursively so that 
$\ell_{k+1}=\ell_k^{\alpha}$ with a specified $\alpha>1$ and $\ell_0$ a large enough integer. 

\begin{proposition}[Conditional intersection-clustering bound]\label{prop:non isolated} 
Under the assumption that the Ising model on $\Z^4$ satisfies \eqref{eq:S}  with $\eta =0$  and restricting to $\alpha>3^8$: there exist $\ell_0=\ell_0(\alpha)$ and $\delta=\delta(\alpha)>0$ such that for every $K> 2$ and every $u,x,y,z,t\in\Z^4$ with mutual distance between $x,y,z,t$ larger than $2\ell_{K}$,
\be
{\bf P}^{ux,uz,uy,ut}_{\beta_c}[{\bf M}_u(\calT; \mathcal L_\alpha, K)< \delta K]\le 2^{-\delta K}.
\ee 
\end{proposition}

Before deriving this estimate, which is proven in the next section, let us show how it leads to the improved tree diagram bound.

\begin{figure}
\begin{center}
\includegraphics[width = 0.50\textwidth]{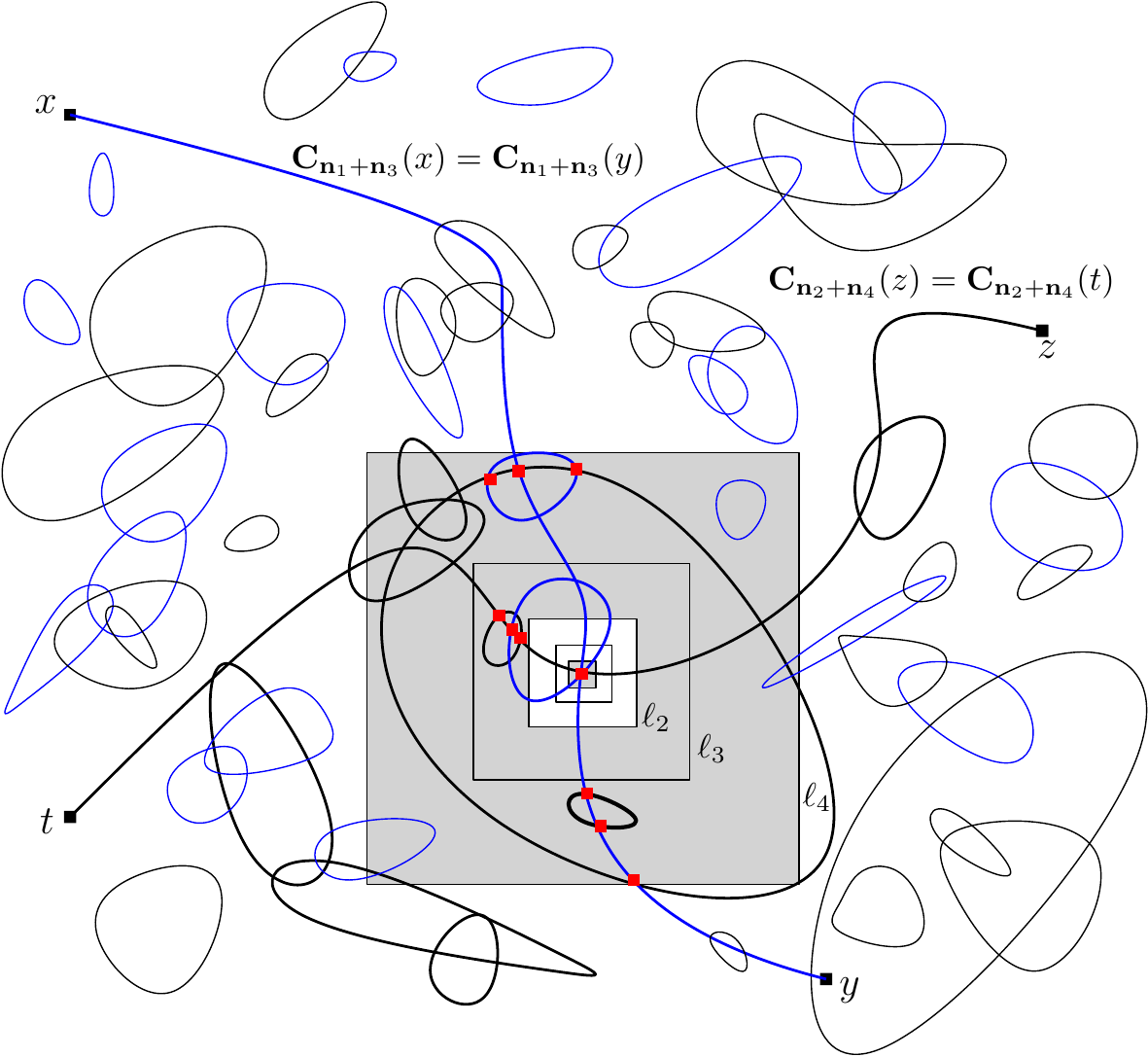}
\caption{The two (duplicated)-currents $\n_1+\n_3$ and $\n_2+\n_4$ in blue and black respectively. The clusters of $x$ (or equivalently $y$) in $\n_1+\n_3$ and $z$ (or equivalently $t$) in $\n_2+\n_4$ are depicted in bold. The red vertices are the elements of the intersection $\calT$. We illustrated the annuli around one element, denoted $u$, of $\calT$ and draw them in gray when an intersection occurs. Here, we therefore have $M_u(\calT;\calL,5)=3$ since three annuli contain an intersection.}
\label{fig:2}
\end{center}
\end{figure}
\bigbreak
\noindent {\bf  Proof of Theorem~\ref{thm:improved tree bound simple} under the assumption \eqref{eq:S}}.   As the discussion is limited here to  $\beta = \beta_c$, we omit it  from the notation.  If $\eta>0$ the bubble diagram is finite and hence  the desired statement is already contained in the tree diagram bound~\eqref{tree}.  Focus then on the case
 $\eta=0$, for which the bubble diagram diverges logarithmically. 
Fix $\alpha>3^8$ and let $\ell_0$ and $\delta$ be given by Proposition~\ref{prop:non isolated}.  Since $x,y,z,t$ are at mutual distances at least $L$, there exists  $c=c(\alpha)>0$ such that one may pick 
\begin{equation}\label{eq:choice K}K=K(L)\ge c\log\log L\end{equation}
in such a way that $L\ge 2\ell_K$.

Using Lemma~\ref{covering}, then the switching lemma, and finally Proposition~\ref{prop:non isolated}, we get 
\begin{align} 
{\bf P}^{xy,zt,\emptyset,\emptyset}[0<|\calT|< 2^{\delta K/5}]
&\le \sum_{u\in \bbZ^4}{\bf P}^{xy,zt,\emptyset,\emptyset}[u\in \calT,{\bf M}_{u}(\calT; \mathcal L_\alpha, K)< \delta K]\nonumber\\
&= \sum_{u\in \bbZ^4}\frac{\langle\sigma_u\sigma_x\rangle\langle\sigma_u\sigma_y\rangle\langle\sigma_u\sigma_z\rangle\langle\sigma_u\sigma_t\rangle}{\langle\sigma_x\sigma_y\rangle\langle\sigma_z\sigma_t\rangle}{\bf P}^{ux,uz,uy,ut}[{\bf M}_{u}(\calT; \mathcal L_\alpha, K)< \delta K]\nonumber \\
&\le 2^{-\delta K} \sum_{u\in \bbZ^4}\frac{\langle\sigma_u\sigma_x\rangle\langle\sigma_u\sigma_y\rangle\langle\sigma_u\sigma_z\rangle\langle\sigma_u\sigma_t\rangle}{\langle\sigma_x\sigma_y\rangle\langle\sigma_z\sigma_t\rangle}.  \label{eq:p2}
\end{align}
For the larger values of $|\calT|$, the Markov inequality and the switching lemma give
\begin{align} {\bf P}^{xy,zt,\emptyset,\emptyset}[|\calT|\ge2^{\delta K/5}]
&~\leq~ 2^{-\delta K/5}{\bf E}^{xy,zt,\emptyset,\emptyset}[|\calT|] \notag \\ 
&~=~2^{-\delta K/5}\sum_{u\in \bbZ^4}\frac{\langle\sigma_u\sigma_x\rangle\langle\sigma_u\sigma_y\rangle\langle\sigma_u\sigma_z\rangle\langle\sigma_u\sigma_t\rangle}{\langle\sigma_x\sigma_y\rangle\langle\sigma_z\sigma_t\rangle}.\label{eq:p1}
\end{align} 
Adding \eqref{eq:p2} and \eqref{eq:p1} gives an improved tree diagram bound which, in view of \eqref{eq:choice K} and of the logarithmic divergence of $B_L(\beta_c)$ implied by $\eta=0$, yields \eqref{eq:improved tree bound}.  
\mbox{} \hfill\ensuremath{\square}

\subsection{Derivation of the conditional intersection-clustering bound (Proposition~\ref{prop:non isolated})} \label{sec:2.2}

The intuition underlying  the conditional intersection-clustering bound and the choice of $\ell_k$ is guided by the aforementioned  example of simple random walks. In dimension 4, the traces of two independent random walks starting at the origin intersect in an annulus of the form ${\rm Ann}(n,n^{\alpha})$ with probability at least $c(\alpha)>0$ uniformly in $n$.   Since the paths traced by these random walks within different annuli are roughly independent, one may expect the number of annuli among the $K$ first ones in which the paths intersect to be,  with large probability, of the order of $\delta K$.   

However, in the case considered  here, the clusters 
of $u$ in $\n_1+\n_3$ and $\n_2+\n_4$ do not have the renewal structure of  Markovian random walks.   We shall compensate for that in two steps:
\begin{itemize}[noitemsep,nolistsep]
\item[(i)] \emph{reformulate the intersection property},   
\item[(ii)] \emph{derive an asymptotic mixing statement.}  
\end{itemize}

For the first step, let $I_k$ 
be the event  (with $I$ standing for intersection)  that there exist unique clusters of ${\rm Ann}(\ell_{k},\ell_{k+1})$ in $\n_1+\n_3$ and $\n_2+\n_4$ crossing the annulus from the inner boundary to the outer boundary and that these two clusters are intersecting.   
Lemma~\ref{lem:intersection} presents the statement that   the probability that the event occurs and that these clusters intersect, is bounded away from 0 uniformly in $k$. 
 
 Note that  the annuli ${\rm Ann}(\ell_{k},\ell_{k+1})$ are wide enough so that  sourceless currents will typically have no radial crossing, and when such crossings are forced by the placement of sources (for instance when one source, is at the common center  of a family of nested annuli and the other at a distant site outside), in each annulus there will most likely be only one crossing cluster.   It then follows that all the crossing clusters of $\n_1+\n_3$  belong to the $\n_1+\n_3$ cluster of the sources, and a similar property holds for  the crossing clusters of 
 $\n_2+\n_4$.   
 
For the second step, we prove that 
events observed within sufficiently separated annuli are roughly independent.  The exact assertion  is presented below in  
Proposition~\ref{prop:mixing multi simplify} and will be the crux of the whole paper. 

Following is the first of these two statements. 

\begin{lemma}[Conditional intersection-clustering property]\label{lem:intersection}
Assume \eqref{eq:S} holds for the Ising model on $\Z^4$ with $\eta =0$. For $\alpha>3^4$, there exist $\ell_0=\ell_0(\alpha)$ and $c=c(\alpha,\ell_0)>0$ such that for every $x,z\notin \Lambda_{2\ell_{k+1}}$, 
\be
{\bf P}^{0x,0z,\emptyset,\emptyset}_{\beta_c}[I_k]\ge c.
\ee
\end{lemma}

The main ingredient in the proof is a second moment method on the number of intersections in ${\rm Ann}(\ell_{k},\ell_{k+1})$ of the clusters of the origin in $\n_1+\n_3$ and $\n_2+\n_4$. A second  part of the proof is devoted to the uniqueness of the clusters crossing the annulus.   This  makes the event under consideration measurable in terms of the currents within just the specified annulus, allowing us to apply the mixing property for the proof of Proposition~\ref{prop:non isolated}, which follows further below. 

\begin{proof}Drop $\beta_c$ from the notation. 
Fix $\alpha>3^4$ and set $\ep>0$ so that $\alpha>(1+\ep)(3+\ep)^4$. The constants $c_i$ below depend on $\ep$ only. Introduce the intermediary integers $n\le m\le M\le N$ satisfying 
\be 
n\ge \ell_{k}^{3+\ep}\ ,\  m\ge n^{3+\ep}\ ,\  M\ge m^{1+\ep}\ ,\  N\ge M^{3+\ep}\ ,\ \ell_{k+1}\ge N^{3+\ep}.
\ee
 We start by proving that 
 $\calM:={\bf C}_{\n_1+\n_3}(0)\cap {\bf C}_{\n_2+\n_4}(0)\cap {\rm Ann}(m,M)$
  is non-empty with positive probability by applying a second-moment method on $|\calM|$. Namely, the switching lemma (more precisely \eqref{eq:prop2b}) and \eqref{eq:S} imply   that 
\begin{align}{\bf E}^{0x,0z,\emptyset,\emptyset}[|\calM|] &=\sum_{v\in {\rm Ann}(m,M)}{\bf P}^{0x,\emptyset}[v\stackrel{\n_1+\n_2}\longleftrightarrow 0]{\bf P}^{0z,\emptyset}[v\stackrel{\n_1+\n_2}\longleftrightarrow 0]\nonumber\\
&=\sum_{v\in {\rm Ann}(m,M)}\frac{\langle\sigma_0\sigma_v\rangle\langle\sigma_v\sigma_x\rangle}{\langle\sigma_0\sigma_x\rangle}\frac{\langle\sigma_0\sigma_v\rangle\langle\sigma_v\sigma_z\rangle}{\langle\sigma_0\sigma_z\rangle}\nonumber\\ 
&\geq  \, c_1 (B_M-B_{m-1}) \, \geq  \, c_2 \log (M/m)\, .\end{align} 
On the other hand, we find that 
\begin{align}
{\bf E}^{0x,0z,\emptyset,\emptyset}[|\calM|^2]&=\sum_{v,w\in {\rm Ann}(m,M)}{\bf P}^{0x,\emptyset}[v,w\stackrel{\n_1+\n_2}\longleftrightarrow 0]{\bf P}^{0z,\emptyset}[v,w\stackrel{\n_1+\n_2}\longleftrightarrow 0].
\end{align}
Now, by a delicate application of the switching lemma and a monotonicity argument   we have  the following inequality (stated and proven as  
 Proposition~\ref{prop:3}  in the Appendix),
 \begin{equation}
{\bf P}^{0x,\emptyset}[v,w\stackrel{\n_1+\n_2}\longleftrightarrow 0]\le \frac{\langle\sigma_0\sigma_v\rangle\langle\sigma_v\sigma_w\rangle\langle\sigma_w\sigma_x\rangle}{\langle\sigma_0\sigma_x\rangle}+\frac{\langle\sigma_0\sigma_w\rangle\langle\sigma_w\sigma_v\rangle\langle\sigma_v\sigma_x\rangle}{\langle\sigma_0\sigma_x\rangle}\,. 
\end{equation}
Together with \eqref{eq:S}, this gives 
\begin{align} 
{\bf E}^{0x,0z,\emptyset,\emptyset}[|\calM|^2] 
\ \leq \, C_3 (B_M-B_{m-1}) \, B_{2M} & \leq \,  C_4 (\log M)^2.\end{align}
The second moment (or Cauchy-Schwarz) inequality, 
and the bound $M\ge m^{1+\ep}$ thus imply  
\begin{equation} \label{eq:M}
{\bf P}^{0x,0z,\emptyset,\emptyset}[\calM\ne \emptyset]\, \geq \, 
\frac{{\bf E}^{0x,0z,\emptyset,\emptyset}[|\calM|]^2}{{\bf E}^{0x,0z,\emptyset,\emptyset}[|\calM|^2] } \, \geq \,  c_5>0.\end{equation}
At this stage, one may feel that the main point of the lemma was established: we showed that with uniformly positive probability the clusters of 0 in $\n_1+\n_3$ and $\n_2+\n_4$ intersect  in ${\rm Ann}(m,M)$. However,  to conclude the argument we need to establish the uniqueness, with large probability,  of the crossing cluster in $\n_1+\n_3$ (the same then holds true for $\n_2+\n_4$).  
This part of the proof is slightly more technical and may be omitted in a first reading.  
It is here that we shall need $\alpha$ to be large enough.

To prove the uniqueness of crossings, we employ the notion of the current's {\em backbone}\footnote{We mentioned that a current $\n$ with sources $x$ and $y$ can be seen as the superposition of one path from $x$ to $y$ and loops. The backbone
$\Gamma(\n)$ is an appropriate choice of such a path induced by an ordering of the edges. Again, we refrain ourselves from providing more details here and refer to the relevant literature for details on this notion.}, on which more can be found in  \cite{Aiz82,AizBarFer87,Dum16,Dum17,DumTas15}. If the event $\{\calM\ne \emptyset\}$ occurs but not $I_k$, then one of the following four events must occur (see e.g.~Fig.~\ref{fig:3}): \begin{itemize}[noitemsep]
\item[$F_1:=$] the backbone $\Gamma(\n_1)$ of $\n_1$ does two successive crossings of ${\rm Ann}(\ell_{k},n)$;
\item[$F_2:=$] $\n_1+\n_2$ contains a cluster crossing ${\rm Ann}(n,m)\setminus\Gamma(\n_1)$;
\item[$F_3:=$] $\n_1+\n_2$ contains a cluster crossing ${\rm Ann}(M,N)\setminus \Gamma(\n_1)$;
\item[$F_4:=$] the backbone $\Gamma(\n_1)$ of $\n_1$  does two successive crossings of ${\rm Ann}(N,\ell_{k+1})$.  
\end{itemize}
We bound the probabilities of these events separately.  
For $F_1$ to occur, the backbone $\Gamma(\n_1)$ must do a zigzag: to go from 0 to a vertex $v\in\partial\Lambda_n$, then to a vertex $w\in\partial\Lambda_{\ell_k}$, and finally to $x$. 
The {\em chain rule} for backbones (see e.g.~\cite{AizBarFer87}) combined with 
the assumed condition \eqref{eq:S}, jointly imply that 
\begin{equation}\label{eq:bound}{\bf P}^{0x,\emptyset}[F_1]\le \sum_{\substack{v\in \partial\Lambda_{n}\\w\in \partial\Lambda_{\ell_{k}}}}\frac{\langle\sigma_0\sigma_v\rangle\langle\sigma_v\sigma_w\rangle\langle\sigma_w\sigma_x\rangle}{\langle\sigma_0\sigma_x\rangle}\le C_6 n^3\ell_{k}^3 n^{-4}\le C_7\ell_{k}^{-\ep}.\end{equation}
To bound the probability of $F_2$, condition on $\Gamma(\n_1)$. The remaining current in $\n_1$ is a sourceless current with depleted coupling constants (see \cite{AizBarFer87,Dum16,Dum17} for details on this type of reasoning). The probability that some $v\in\partial\Lambda_n$ and $w\in \partial\Lambda_m$ are connected in $\Z^4\setminus \Gamma(\n_1)$ to each other can then be bounded by $\langle\sigma_v\sigma_w\rangle\langle\sigma_v\sigma_w\rangle'$ where the $\langle\cdot\rangle'$ denotes an Ising measure with depleted coupling constants (the depletion depends on $\Gamma(\n_1)$ and the switching lemma concerns one current with depletion and one without; we refer to \cite{AizDumSid15} for the statement and proof of the switching lemma in this context, and some applications). At the risk of repeating ourselves, we refer to \cite{AizBarFer87} for an illustration of this line of reasoning. The Griffiths inequality \cite{Gri67} implies that this probability is bounded by $\langle\sigma_v\sigma_w\rangle^2$, which together with \eqref{eq:S}, immediately leads to the following sequence of inequalities: 
\begin{align}\label{eq:bound sourceless}
{\bf P}^{0x,\emptyset}[F_2]&\le\sum_{\substack{v\in \partial\Lambda_n\\w\in \partial\Lambda_m}} \langle\sigma_v\sigma_w\rangle^2\le C_8n^{-\ep}.
\end{align}
The event $F_4$   is bounded  similarly to $F_1$, and $F_3$ similarly to 
$F_2$.
For $\ell_0=\ell_0(\ep)$ large enough  the sum of the four probabilities does not exceed half of the constant $c_5$ in \eqref{eq:M}, and the main statement follows.  
\end{proof}

\begin{figure}
\begin{center}
\includegraphics[width = 0.50\textwidth]{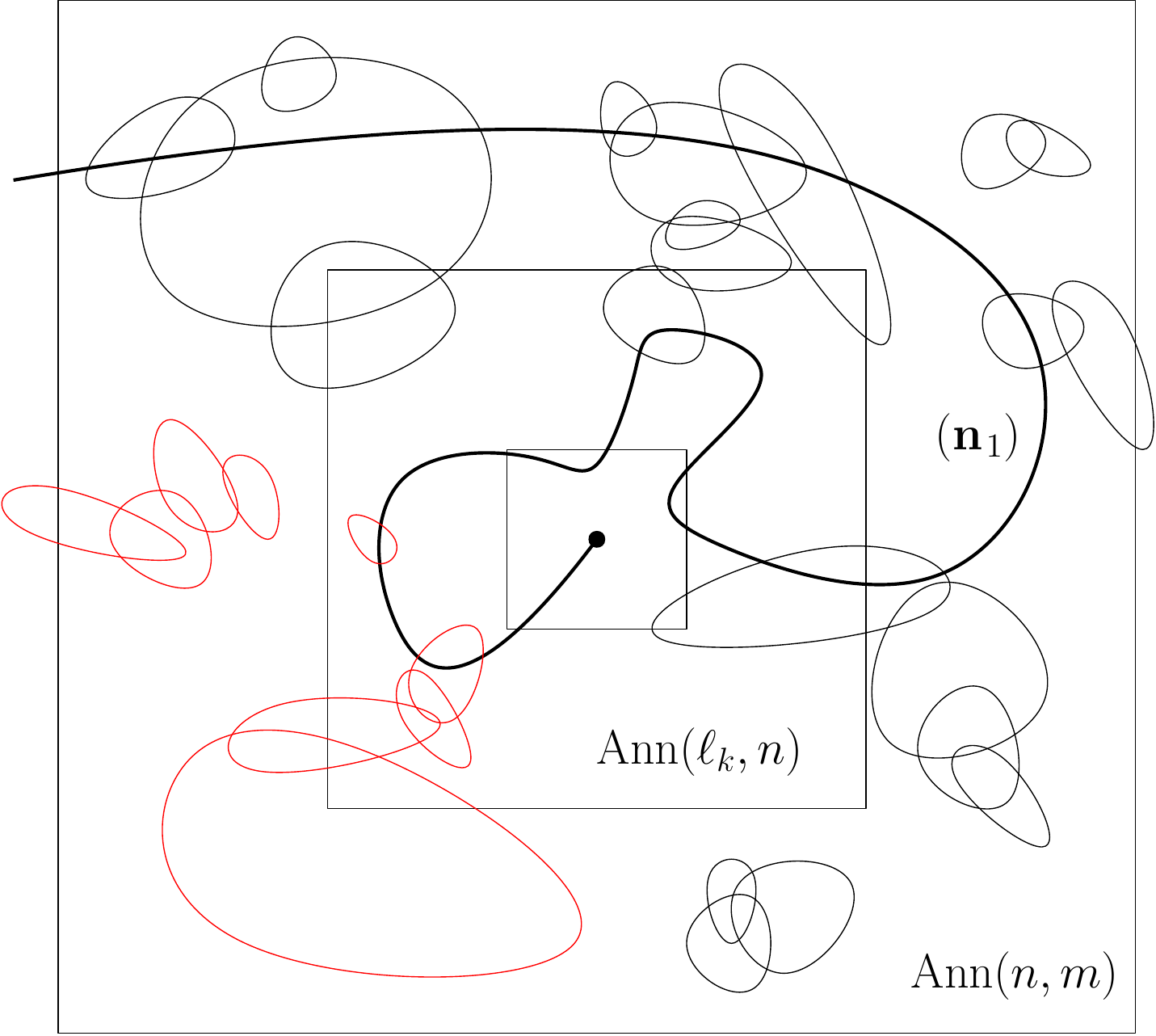}
\caption{In this picture, $\Gamma(\n_1)$ does only one crossing of ${\rm Ann}(\ell_{k},n)$, and $\n_1+\n_2-\Gamma(\n_1)$ does not cross ${\rm Ann}(n,m)\setminus\Gamma(\n_1)$. This prevents the fact that the cluster in red, made of loops in $\n_1+\n_2-\Gamma(\n_1)$ would connect an excursion of $\Gamma(\n_1)$ outside of $\Lambda_{\ell_{k}}$ but not reaching $\partial\Lambda_n$ to  $\partial\Lambda_m$ (which would potentially create an additional cluster crossing ${\rm Ann}(\ell_{k},m)$).}
\label{fig:3}
\end{center}
\end{figure}

\begin{remark}
The condition $\alpha>3^4$ is used in the second part of the proof,
 where we need the exponent connecting the inner and outer radii of annuli to be strictly larger than 3.  We did not try to improve on this exponent.\end{remark}

The second of the above described statements is 
one of the main innovations of this paper.  It concerns a mixing property, which in Section~\ref{sec:4.1} will be 
 stated under a more general form and derived unconditionally  
for every $d\ge4$.
 
\begin{proposition}[Conditional mixing property]\label{prop:mixing multi simplify}
Assume that the complementary pair of  power law bounds  \eqref{eq:S} holds for the Ising model on $\Z^4$ with $\eta =0$, and fix $\alpha>3^8$.  Then there exists $C>0$ such that for every $ n^\alpha\le N$, every $x\notin \Lambda_N$, 
and every pair of events $E$ and $F$ depending on the restriction of $\n$ to edges within $\Lambda_n$ and outside of $\Lambda_N$ respectively, 
\begin{align}|{\bf P}^{0x}_{\beta_c}[E\cap F]-{\bf P}^{0x}_{\beta_c}[E]{\bf P}^{0x}_{\beta_c}[F]|&\le \frac{C}{\sqrt{\log (N/n)}}.
\end{align}
\end{proposition}

The heart of the proof will be the use of a (random) resolution of identity ${\bf N}$, meaning a random variable which is concentrated around 1, given by a weighted sum of indicator functions $\mathbb I[y\stackrel{\n_1+\n_2}\longleftrightarrow 0]$ with $y\in \mathbb Z^d$, where $\partial\n_1=\{0,x\}$ and $\partial\n_2=\emptyset$, which will enable us to write
\be
{\bf P}^{0x}[E\cap F]\approx{\bf E}^{0x,\emptyset}[{\bf N}\mathbb I(\n_1\in E\cap F)].
\ee
 Since ${\bf N}$ will be a certain convex combination of the random variables $\mathbb I[y\stackrel{\n_1+\n_2}{\longleftrightarrow}0]/\langle\sigma_0\sigma_y\rangle$, the term on the right will be a convex sum of 
${\bf P}^{0x,\emptyset}$-probabilities of the events $\{y\stackrel{\n_1+\n_2}{\longleftrightarrow}0,\n_1\in E\cap F\}$. For each fixed $y$, we will use the switching principle to transform the sources $\{0,x\}$ and $\emptyset$ of $\n_1$ and $\n_2$ into $\{0,y\}$ and $\{y,x\}$, exchanging at the same time the roles of $\n_1$ and $\n_2$ inside $\Lambda_n$ without changing anything outside $\Lambda_N$. This useful operation has a nice byproduct: the event $\n_1\in F$ becomes $\n_2\in F$ which is independent of $\n_1\in E$. Deducing the mixing from there will be a matter of elementary algebraic manipulations. 

The error term will be (almost entirely) due to how  concentrated around $1$ ${\bf N}$ is. In order to prove this fact, we will implement a refined second moment method in which we estimate the expectation  and the second moment of ${\bf N}$ sharply. The proof will require some regularity assumptions on the gradient of the two-point function: for every $x\in\Z^d$, 
\begin{align}
\label{eq:nabla S}|\nabla_x\langle\sigma_0\sigma_x\rangle|\le \frac{C}{|x|}\langle\sigma_0\sigma_x\rangle,
\end{align}
which  follows from \eqref{eq:S} by an argument that we choose to postpone to  Section~\ref{sec:3.3} (after the required technology has been introduced).

\begin{proof}
Let us recall that we are discussing here $\beta=\beta_c$, omitting  the symbol from the notation.
Fix $\alpha>3^4$ (the power $4$ instead of $8$ suffices at this stage) and choose $\ep>0$ so that $\alpha>(1+\ep)(9+\ep)^2$. Below, the constants $C_i$ are independent of $\beta$ and $n^\alpha= N\le \xi(\beta)$ (we may assume equality between $N$ and $n^\alpha$ without loss of generality). Introduce two intermediary integers $m\le M$ satisfying that 
\be
m\ge n^{9+\ep}\ ,\  M\ge m^{1+\ep}\ ,\ N\ge M^{9+\ep}
\ee
as well as the notation $n_k=2^km$ for $k\ge1$. Set $K$ such that $n_{K+1}\le M<n_{K+2}$.
The key to our proof will be the random variable 
\be
{\bf N}:=\frac1{K}\sum_{k=1}^{K}\frac{1}{\alpha_{k}}\sum_{y\in{\rm Ann}(n_k,n_{k+1})}\mathbb I[y\stackrel{\n_1+\n_2}\longleftrightarrow0]\quad\text{where}\quad\alpha_k:=\sum_{y\in {\rm Ann}(n_k,n_{k+1})}\langle\sigma_0\sigma_y\rangle.
\ee
Combining the regularity assumptions \eqref{eq:nabla S} and \eqref{eq:S} with Proposition~\ref{prop:3} (the precise computation is presented in Section~\ref{sec:4.2}), we find 
\begin{align}{\bf E}^{0x,\emptyset}[{\bf N}]&=\frac1{K}\sum_{k=1}^{K}\frac{1}{\alpha_{k}}\sum_{y\in{\rm Ann}(n_k,n_{k+1})}\frac{\langle\sigma_0\sigma_y\rangle\langle\sigma_y\sigma_x\rangle}{\langle \sigma_0\sigma_x\rangle}\ge 1-\frac{C_1}{K},\\
{\bf E}^{0x,\emptyset}[{\bf N}^2]&\le\frac1{K^2}\sum_{k,\ell=1}^{K}\frac{1}{\alpha_{k}\alpha_\ell}\sum_{\substack{y\in{\rm Ann}(n_k,n_{k+1})\\ z\in {\rm Ann}(n_\ell,n_{\ell+1})}}\frac{\langle\sigma_0\sigma_y\rangle\langle\sigma_y\sigma_z\rangle\langle\sigma_z\sigma_x\rangle+\langle\sigma_0\sigma_z\rangle\langle\sigma_z\sigma_y\rangle\langle\sigma_y\sigma_x\rangle}{\langle \sigma_0\sigma_x\rangle}\le 1+\frac{C_2}{K}.\end{align}
The Cauchy-Schwarz inequality and the fact that ${\bf P}^{0x}[E\cap F]={\bf P}^{0x,\emptyset}[\n_1\in E\cap F]$ thus imply that
\begin{align}\label{eq:hgf}|{\bf P}^{0x,\emptyset}[\n_1\in E\cap F]-{\bf E}^{0x,\emptyset}[{\bf N}\mathbb I_{\n_1\in E\cap F}]|&\le \sqrt{{\bf E}^{0x,\emptyset}[({\bf N}-1)^2]}
\le \frac{C_3}{\sqrt K}.
\end{align}
Now, fix $y\in {\rm Ann}(m,M)$ and let $G(y)$ be the event (depending on $\n_1+\n_2$ only) that there exists ${\bf k}\le \n_1+\n_2$ such that ${\bf k}=0$ on $\Lambda_n$, ${\bf k}=\n_1+\n_2$ outside $\Lambda_N$, and $\partial{\bf k}=\{x,y\}$.
We find that
\begin{align}\label{eq:ggf1}{\bf P}^{0x,\emptyset}[\n_1\in E\cap F,y\stackrel{\n_1+\n_2}\longleftrightarrow 0,G(y)]=\frac{\langle\sigma_0\sigma_y\rangle\langle\sigma_y\sigma_x\rangle}{\langle\sigma_0\sigma_x\rangle}{\bf P}^{0y,yx}[\n_1\in E,\n_2\in F,G(y)],\end{align}
where we use the following reasoning:
for  $\m\in G(y)$, consider the multi-graph $\calM$ obtained by duplicating every edge of the graph into $\m(x,y)$ edges. If $G(y)$ occurs, the existence of $\mathbf k$ guarantees the existence of a subgraph $\calK\subset\calM$ with $\partial\calK=\{x,y\}$ containing no edge with endpoints in $\Lambda_n$ and all those of $\calM$ with endpoints outside $\Lambda_N$, so that the generalized switching principle formulated in \cite[Lemma~2.1]{AizDumTasWar18} implies that 
\begin{align}\sum_{\calT\le\calM:\partial\calT=\{0,x\}}\mathbb I[\calT\in E\cap F]&=\sum_{\calT\le\calM:\partial\calT\Delta\calK=\{0,x\}}\mathbb I[\calT\Delta\calK\in E\cap F]\nonumber\\
&=\sum_{\calT\le\calM:\partial\calT=\{0,y\}}\mathbb I[\calT\in E,\calM\setminus\calT\in F],\label{eq:swit}\end{align}
where we allow ourselves the latitude of calling $E$ and $F$ the events defined for multi-graphs corresponding to the events $E$ and $F$ for currents. One gets \eqref{eq:ggf1} when rephrasing this equality in terms of weighted currents (exactly like in standard proofs of the switching principle, see e.g.~\cite{Aiz82} or \cite{AizDumTasWar18} for a closely related reasoning).
 
Observe now that forgetting about $G(y)$ on the right-hand side of \eqref{eq:ggf1} gives 
\begin{equation}{\bf P}^{0y,yx}[\n_1\in E,\n_2\in F]={\bf P}^{0y}[E]{\bf P}^{yx}[F].\end{equation}
Furthermore, since $x\notin \Lambda_N$ and $y\in\Lambda_{m}$, \eqref{eq:nabla S} implies that 
\begin{equation}
\Big|\frac{\langle\sigma_0\sigma_y\rangle\langle\sigma_y\sigma_x\rangle}{\langle\sigma_0\sigma_x\rangle}-\langle\sigma_0\sigma_y\rangle\Big|\le \frac{C_4m}{N}.
\end{equation} 
Last but not least, we can bound (from below) ${\bf P}^{0x,\emptyset}[G(y)]$ and ${\bf P}^{0y,yx}[G(y)]$ as follows. We only briefly describe the argument since we will present it in full details in Section~\ref{sec:4.2}. The event $G(y)$ clearly contains the event that ${\rm Ann}(M,N)$ is not crossed by a cluster in $\n_1$, and ${\rm Ann}(n,m)$ is not crossed by a cluster in $\n_2$, since in such case ${\bf k}$ can be defined as the sum of $\n_1$ restricted to the clusters intersecting $\Lambda_N^c$ (this current has no sources) and $\n_2$ restricted to the clusters intersecting $\Lambda_m^c$ (this current has sources $x$ and $y$).
Now, we can bound the probability of $\n_1$ crossing ${\rm Ann}(M,N)$ in the same spirit as  we bounded the probabilities for $F_1$ and $F_3$ in the previous proof by splitting ${\rm Ann}(M,N)$ in two annuli ${\rm Ann}(\sqrt{MN},N)$ and ${\rm Ann}(M,\sqrt{MN})$, then estimating the probability that the backbone of $\n_1$ crosses the inner annulus more than once, and then  the probability that the remaining current (which is sourceless) crosses the outer annulus. Doing the same for the probability that a cluster of $\n_2$ crosses ${\rm Ann}(n,m)$, we find that
\begin{equation}
\label{eq:gggff}\frac{\langle\sigma_0\sigma_x\rangle}{\langle\sigma_0\sigma_y\rangle\langle\sigma_y\sigma_x\rangle}{\bf P}^{0x,\emptyset}[G(y),y\stackrel{\n_1+\n_2}{\longleftrightarrow}0]={\bf P}^{0y,yx}[G(y)]\ge 1-\frac{C_5}{n^\ep}\ge 1-C_6\big(\frac nN\big)^{\ep/(\alpha-1)}.
\end{equation}
Note that we use that $M\ge N^{9+\ep}$ in this part of the proof. 

 Overall, the value of $K$ and \eqref{eq:hgf}--\eqref{eq:gggff} put together imply
\begin{align}\label{eq:gf simple}|{\bf P}^{0x}[E\cap F]-\sum_{y\in {\rm Ann}(m,M)}\delta(y){\bf P}^{0y,\emptyset}[E]{\bf P}^{yx,\emptyset}[F]|&
\le \frac{C_7}{\sqrt{\log(N/n)}},
\end{align}
with $\delta(y)=\langle\sigma_0\sigma_y\rangle/(K\alpha_{k(y)})$ where $k(y)$ is such that $y\in{\rm Ann}(n_{k(y)},n_{k(y)+1})$.

The end of the proof is now a matter of elementary algebraic manipulations. Applying this inequality twice (once with $x$ and once with $x'$) for $F$ being the full set, we obtain that for every $x,x'\notin \Lambda_N$ and every event $E$ which is depending on $\Lambda_n$ only,  
\begin{align}\label{eq:roro}
&|{\bf P}^{0x}[E]-{\bf P}^{0x'}[E]|\le \frac{2C_7}{\sqrt {\log(N/n)}}.
\end{align}
Now, assume the stronger assumption that $\alpha>3^8$ and fix $m=\lfloor\sqrt{Nn}\rfloor$. Applying
\begin{itemize}[noitemsep]
\item \eqref{eq:gf simple} for $m$ and $N$, the full event and $F$, 
\item then \eqref{eq:roro} for $n$, $m$ and $E$ (note that $m\ge n^{3^4}$), 
\item and \eqref{eq:gf simple} for $m$ and $N$, $E$ and $F$, 
\end{itemize}
gives that for every $x\notin \Lambda_N$, 
\begin{align}
|{\bf P}^{0x}[E\cap F]-{\bf P}^{0x}[E]{\bf P}^{0x}[F]|&\le |{\bf P}^{0x}[E\cap F]-{\bf P}^{0x}[E]\sum_{y}\delta(y){\bf P}^{yx}[F]|+\frac{C_7}{\sqrt{\log(N/n)}}\nonumber\\
&\le|{\bf P}^{0x}[E\cap F]-\sum_{y}\delta(y){\bf P}^{0y}[E]{\bf P}^{yx}[F]|+\frac{3C_7}{\sqrt{\log(N/n)}}\nonumber\\
&\le \frac{4C_7}{\sqrt {\log(N/n)}}.
\end{align}
\end{proof}

Using Lemma~\ref{lem:intersection} and Proposition~\ref{prop:mixing multi simplify}, we may now establish the clustering of intersections, under  \eqref{eq:S}.

\begin{proof}[Proposition~\ref{prop:non isolated}]  In view of the translation invariance of the claimed statement, we take $u$ to be the origin. Since $x$ and $y$ are at a distance larger than $2\ell_K$ of each other, one of them is at a distance (larger than or equal to) $\ell_K$ of $u$.  Without loss of generality we take that to be $x$, and make a similar assumption about $z$.

Let $\calS_K$ denote the set of subsets of $\{1,\dots,K-2\}$ containing {\em even} integers only and fix  $S\in \calS_K$. Let $A_S$ be the event that no $I_k$ occurs for $k\in S$. If $s$ denotes the maximal element of $S$, the mixing property Proposition~\ref{prop:mixing multi simplify} used with $n=\ell_{s-1}$ and $N=\ell_{s}$ gives
\be
{\bf P}^{0x,0z,\emptyset,\emptyset}[A_S]\le {\bf P}^{0x,0z,\emptyset,\emptyset}[I_s^c]{\bf P}^{0x,0z,\emptyset,\emptyset}[A_{S\setminus\{s\}}]+\frac{C}{\sqrt{\log \ell_{s-1}}}.
\ee
To be precise and honest, we use a multi-current version, with four currents, of the mixing property. We will state and prove this property in Sections~\ref{sec:4.1} and \ref{sec:4.2} and ignore this additional difficulty for now. Also,
it is here that the stronger restriction $\alpha>3^8$ is used, along with the choice of  $\ell_0=\ell_0(\alpha)$,  to enable the mixing.  Note that we used that the event $I_s$ is expressed in terms of just the restriction of the currents $\n_1,\dots,\n_4$ to ${\rm Ann}(\ell_s,\ell_{s+1})$. 

Now,  the intersection property Lemma~\ref{lem:intersection} and an elementary bound on $\ell_{s-1}$ gives the existence of $c_0>0$ small enough that
\begin{align}{\bf P}^{0x,0z,\emptyset,\emptyset}[A_S]
&\le (1-2c_0){\bf P}^{0x,0z,\emptyset,\emptyset}[A_{S\setminus \{s\}}]+c_0(1-c_0)^{|S|-1}.\end{align}
An induction gives immediately that for every $S\in \calS_K$,
\begin{align}\label{eq:bound B}
&{\bf P}^{0x,0z,\emptyset,\emptyset}[A_S]\le (1-c_0)^{|S|}.\end{align}
Let $B_S\subset A_S$ be the event that the clusters of $0$ in $\n_1+\n_3$ and $\n_2+\n_4$ do not intersect in any of the annuli ${\rm Ann}(\ell_{s},\ell_{s+1})$ for $s\in S$.
Thanks to Corollary~\ref{cor:1}, the probability of $B_S$ increases when removing sources, so that
\begin{equation}
{\bf P}^{0x,0z,0y,0t}[B_S]\le {\bf P}^{0x,0z,\emptyset,\emptyset}[B_S]\le{\bf P}^{0x,0z,\emptyset,\emptyset}[A_S]\le (1-c_0)^{|S|}.
\end{equation}
To conclude, observe that if ${\bf M}_0(\calT; \mathcal L_\alpha, K)\le \delta K$, then there must exist a set $S\in \calS_K$ of cardinality at least $(\tfrac12-\delta) K$ such that $B_S$ occurs. We deduce that
\begin{align}{\bf P}^{0x,0z,0y,0t}[{\bf M}_0(\calT; \mathcal L_\alpha, K)< \delta K]
&\le\sum_{S\in\calS_K:
|S|\ge(1/2-\delta)K}{\bf P}^{0x,0z,0y,0t}[B_S]\nonumber \\
&\le \binom{K/2}{\delta K}(1-c_0)^{(1/2-\delta)K},
\end{align}
which implies the claim by appropriately choosing the value of $\delta$.
\end{proof}

\section{Weak regularity of the two-point function} \label{sec:3}

Progressing towards the unconditional proof of Theorem~\ref{thm:improved tree bound simple}  we establish in this section the abundance, below the correlation length, 
of {\em regular scales}  at which the two-point function has properties similar to those it would have under the power-law decay assumption \eqref{eq:S}.   This auxiliary result is stated here as Theorem~\ref{thm:regular scales}. 

Towards this goal we focus here on the two-point function,  and present some old and new observations.  In particular, we discuss the following three properties of the two-point function: 
\begin{enumerate}[label=(\roman*)]
\item   {\em monotonicity} (Section~\ref{sec:3.1})
\item   {\em sliding-scale spatial Infrared Bound} (Section~\ref{sec:3.2}), 
\item   {\em gradient estimate} (Section~\ref{sec:3.3}),
\item {\em a lower bound for the two point function at $\beta_c$}.
\end{enumerate} 
The first three  are based on the reflection-positivity of the n.n.f.~interaction, and apply to systems of real valued variables of arbitrary (but common) distribution, of sub-gaussian growth, i.e.~satisfying 
\eqref{sub_gauss}.   
That includes the Ising and $\varphi^4$ variables which are of particular interest for us.  
The last item is proven for systems with spins in the GS class.   
\\

\noindent {\bf Some unifying notation:} In statements which apply  to both the Ising and $\varphi^4$ systems, we shall  refer to the spin/field variables by the ``neutral'' symbol $\tau$.   Its a-priori distribution is denoted 
$\rho(d\tau)$.   It may be displayed as a subscript, {\em but also will often be omitted}.    

The expectation value functional with respect to the Gibbs measure, or functional integral, for a system in the domain $\Lambda$  is denoted 
$\langle\cdot\rangle_{\Lambda,\rho, \beta}$, with $\langle\cdot\rangle_{\rho, \beta}$ denoting the states' natural infinite volume limit. 

 We also denote by 
$\beta_c(\rho)$  (or just $\beta_c$ where the spins' a-priori distribution $\rho$ is clear from the context) is the  critical inverse temperature and $\xi(\rho,\beta)$ the correlation length.

Throughout this section  $|J|:=\sum_{y}J_{0,y}$ and 
\be
S_{\rho,\beta}(x):=\langle\tau_0\tau_x\rangle_{\rho,\beta}.
\ee 
We refer to points in  $\bbR^d$ as $x=(x_1,\dots,x_d)$
and denote by  ${\bf e}_j$ the unit vector with $x_j=1$.

\subsection{Messager-Miracle-Sole monotonicity for the two-point function}\label{sec:3.1}

The Messager-Miracle-Sol\'e (MMS) inequality \cite{Heg77,MesMir77,Sch77} states that for models with n.n.f.~interactions (and more generally \emph{reflection-positive} interactions)
in a region $\Lambda$ endowed with reflection symmetry,  
the correlation function $\langle \prod_{x\in A}\tau_x\prod_{x\in B}\tau_x\rangle _{\Lambda,\rho,\beta}$ at sets of sites $A$ and $B$ which are on the same side of a reflection plane, can only decrease when $B$ is replaced by its reflected image, $\mathcal{R}(B)$, i.e.
\begin{equation}\label{eq:MMS}
\big\langle \prod_{x\in A}\tau_x\prod_{x\in B}\tau_x\big\rangle_{\Lambda,\rho,\beta}  \ \geq \ \big\langle \prod_{x\in A}\tau_x\prod_{x\in \calR(B)}\tau_x\big\rangle_{\Lambda,\rho,\beta}. 
\end{equation}   In the infinite volume limit on $\Z^d$,  this principle can be invoked for reflections with respect to 
\begin{itemize}[noitemsep,nolistsep]
\item hyperplanes passing through vertices or mid-edges, i.e.~reflections changing only one coordinate $x_i$, which is sent to $L-x_i$ for some fixed $L\in\tfrac12\bbZ$, 
\item ``diagonal'' hyperplanes, i.e.~reflections changing only two coordinates $x_i$ and $x_j$, which are sent to $x_j\pm L$ and $x_i\mp L$ respectively, for some $L\in\bbZ$.
\end{itemize}
In particular, this  implies the following useful comparison principle. 

\begin{proposition}[MMS monotonicity]\label{prop:MMS}  
For the n.n.f.~model 
on $\Z^d$ ($d\ge1$) with real valued spin variables satisfying \eqref{sub_gauss}:\\ 
\indent  i)
along the principal axes 
the two-point function is monotone decreasing in $\|x\|_\infty$ 
\\ 
\indent ii) for any $x=(x_1,\dots,x_d) \in \Z^d$,
\begin{equation}\label{eq:b55}
S_{\rho,\beta}((\|x\|_\infty,0_\perp)) \geq S_{\rho,\beta}(x) \geq  S_{\rho,\beta}((\|x\|_1,0_\perp) ),
\end{equation}
where  $\|x\|_1 := \sum_{j=1}^d |x_j|$, $\|x\|_\infty:= \max_j |x_j|$, and $0_\perp$ is the null vector in  $\Z^{d-1}$.
\end{proposition} 

The above carries the useful implication  that for any $x,y\in \Z^d$ with $\|y\|_\infty \, \ge \,  d\,  \|x\|_\infty $,
\be  \label{eq:MMS3}
 S_{\rho,\beta}(x)\, \ge \, S_{\rho,\beta}(y)
\ee  
since $ \|x\|_1 \, \le \,  d\, \|x\|_\infty   \le   \|y\|_\infty  $, by \eqref{eq:b55} the two quantities are on the correspondingly opposite sides of $S_{\rho,\beta}((\|x\|_1,0_\perp)$).\\ 

We shall  encounter below also monotonicity statements of the Fourier transform.  Both are useful in extracting  point-wise implications from bounds on the corresponding two point function's bulk averages (as in  Corollary~\ref{cor:SL},  below).   

\subsection{The two-point function's Fourier transform}

In view of the model's translation invariance it is natural to consider the system's behavior also through its Fourier \emph{spin-wave modes}.  These are defined as   
\be
\widehat \tau_\beta(p) := \frac{1}{\sqrt{(2L)^d} } \sum_{x\in (-L,L]^d}e^{ip\cdot x} \tau_x \, 
\ee 
with $p$ ranging over $\Lambda_L^*:=[-\pi,\pi)^d\cap\tfrac{\pi}{L}\Z^d$. 

 These variables are especially relevant in case the Hamiltonian is taken with the periodic boundary conditions, under which sites $x,y\in \Lambda_L$ are neighbors if either $\|x-y\|_1=1$ or  $|y_i-x_i |=  2L-1$ for some $i\in \{1,\dots,d\}$.   With these boundary conditions, the model is invariant under cyclic shifts, and its Hamiltonian decomposes into a sum of single-mode contributions: 

\be
H_{\Lambda_R}(\tau) = \sum_{p\in \Lambda_L^*} \mathcal E(p) \, \vert  \widehat \tau_\beta(p)\vert^2,
\ee 
with
\be 
\mathcal E(p) :=  2\sum_{j=1}^d [1- \cos(p_j)] = 4  \sum_{j=1}^d  \sin^2(p_j/2)\, .
\ee   

Among the various  relations in which the Fourier transform plays a useful role,  the following statements  will be of relevance for our discussion.
\begin{itemize} 
 \item[i)]   The  spin-wave modes' second moment coincides with the finite volume Fourier transform of the two-point correlation function ($S^{(L)}(p)$): 
\be
\widehat S^{(L)}_{\rho,\beta}(p)   :=  \sum_{x\in \Lambda_L} e^{ip\cdot x} \langle \tau_0 \tau_x \rangle_{\Lambda_L,\rho,\beta}^{(b.c.)}  = 
  \langle  \vert  \widehat \tau_\beta(p)\vert^2  \rangle_{\Lambda_L,\rho,\beta}^{(b.c.)}   \geq  0\,.
\ee

 \item[ii)]  For the n.n.f. interaction, and more generally reflection-positive interactions, the following gaussian-domination (aka \emph{infrared}) bound holds~\cite{FILS78,FroSimSpe76}
\be \label{eq_IR}
\mathcal E(p) \widehat S^{(L)}_{\rho,\beta}(p) \leq \frac{1}{2|J|\beta}  \,. 
\ee 
The bound appeals to the physicists' intuition, reminding one of  the equipartition law.    Alas, it has so far been proven only for reflection-positive interactions. 
 \item[iii)]   The Parseval-Plancherel identity yields the sum rule 
\be \label{sumrule}
\langle \tau_0^2 \rangle_{\Lambda_L,\rho,\beta}^{(b.c.)}  =   \frac{1}{\vert \Lambda_L\vert} \sum_{p\in \Lambda_L^*} \widehat S^{(L)}_{\rho,\beta}(p)\,. 
\ee 
\end{itemize} 
As was pointed out in~\cite{FroSimSpe76}, the combination of  \eqref{sumrule} with \eqref{eq_IR} yields 
a (then novel) way to prove the occurrence of spontaneous magnetization in dimensions $d>2$, at  high enough $\beta$.

More explicitly, in \eqref{sumrule}  one may note that the Infrared Bound \eqref{eq_IR}  does not provide any direct control on the  $p=0$ term, since $\mathcal E(0)=0$.   And, in fact,  the hallmark of the low temperature phase ($\beta > \beta_c(\rho)$) is that   
this single value of the summand attains  macroscopic size: 
 \be 
 \widehat S^{(L)}_{\rho,\beta}(0) \approx \vert \Lambda_L\vert \, M(\rho,\beta)^2
 \ee 
 with $M(\rho,\beta)$ the model's spontaneous magnetization.  \\ 
 
We shall also use the following statement on the relation between the finite volume and the infinite volume states.  

\begin{proposition}\label{prop:conv} For {a model in the GS class} on $\Z^d$,  $d >2$,  with translation invariant finite range interactions, for any $\beta < \beta_c(\rho)$: 
\begin{enumerate}  
\item the system has only one infinite volume Gibbs equilibrium state.  
\item the correlation functions of that state satisfy, for any finite $A\subset \Z^d$, and any sequence of finite volumes $V_n \subset \Z^d$ which asymptotically cover any finite region,
\be \label{eq:conv}
\big\langle \prod_{x\in A}\tau_x \big\rangle_{\rho,\beta}  \ = \  \lim_{V_n \to \Z^d}  
\big\langle \prod_{x\in A}\tau_x \big\rangle_{V_n,\rho, \beta}^{(b.c.)}
\ee 
with $\langle -\rangle_{V_n,\rho, \beta}^{(b.c.)}$ denoting the correlation function under  boundary conditions which may include either  cross-boundary spin couplings (e.g.~periodic), or arbitrary specified values of $\tau_{|\partial V_n}$.  
\item with the finite volumes taken as the rectangular domains $\Lambda_L$, also the Fourier Transform functions converge, i.e.~for any $ p  \in [-\pi ,\pi]^d$,  and sequence as in \eqref{eq:conv} 
\be \label{eq:FTconv}
\lim_{ n \to \infty}  \sum_{x\in V_n} e^{ip\cdot x} 
\langle \tau_0 \tau_x \rangle_{V_n,\rho, \beta}^{(b.c.)}  \ = \  
\sum_{x\in \Z^d}e^{ip\cdot x}S_{\rho,\beta}(x) \ =: \ 
\widehat S_{\rho,\beta}(p) 
\ee 
\end{enumerate} 
\end{proposition} 

The statement follows by standard arguments that we omit here.  The main  ingredients are the exponential decay of correlations, which at any $\beta<\beta_c(\rho)$ are exponentially bounded, uniformly in the volume,    and the FKG inequality.   The first two points hold also for 
$\beta=\beta_c(\rho)$ \cite{AizDumSid15}.  However not the last, \eqref{eq:FTconv},   since at the critical temperature the correlation function is not summable.   

We shall employ the freedom which Proposition~\ref{prop:conv} provides in establishing the different monotonicity properties of $S(p)$ in $p$.  

 Furthermore, for the {\em  disordered regime}, where $M(\rho,\beta)=0$, the sum rule combined with the Infrared Bound  implies that for every $\beta<\beta_c(\rho)$,
\be\label{eq:bound tau square}
  \qquad \langle\tau_0^2\rangle_\beta=\int_{[-\pi,\pi]^d} \widehat S_{\rho,\beta}(p)dp\ \leq\  \int_{[-\pi,\pi]^d} \frac{dp}{2|J|\beta \,\calE(p)} \ .
\ee
Since $\calE(p)$ vanishes only at $p=0$ and there at the rate $\calE(p)\sim |p|^2$,  the integral is convergent for $d>2$ and one gets  
 \be \label{sigma^2bound}
\langle\tau_0^2\rangle_{\rho,\beta_c(\rho)} \  \leq\    \frac{C_d}{2 |J|\beta_c(\rho)} 
\ee
with $C_d<\infty$ for $d>2$.   
This bound will be used in  Section~\ref{sec:6}.\\

\subsection{The spectral representation and a sliding-scale Infrared Bound}\label{sec:3.2}

We  next present a Fourier transform counterpart (though one derived by different means) of the Messager-Miracle-Sole monotonicity stated in Section~\ref{sec:3.1}, and use it for a sliding-scale extension of the  Infrared Bound \eqref{eq_IR}.  The results, which include both old~\cite{GliJaf73, Sok82} and new observations, are based on the relation of the two-point function with the transfer matrix, and the positivity of the latter.

The transfer matrix has been the source of many insights on the structure of statistical mechanical systems with finite range interactions.  Its appearance can be seen in  Ising's study of one dimensional systems, for which it permits a simple proof of the absence of phase transition.   Also, in higher dimensions it has played an essential role in many important developments \cite{FroSimSpe76,GliJaf73,SchMatLie64}, some of which rely on positivity properties.   Here we  shall use the following consequences of its spectral representation for the two-point function.

\begin{proposition}[Spectral Representation]\label{prop:spec_rep}  In the n.n.f.~model on $\Z^d$ ($d\ge1$),  at $\beta<\beta_c(\rho)$, for every square summable function $v:\Z^{d-1}\rightarrow\mathbb C$, there exists a positive measure $\mu_{v,\beta}$ of  finite mass
  \be
 \int_{1/\xi(\rho,\beta)}^\infty  d\mu_{v,\beta} ( a) =  
  \sum_{x_\perp,y_\perp\in \Z^{d-1}}v_{x_\perp}\overline{v_{y_\perp}}S_{\rho,\beta}((0,y_\perp-x_\perp))
 \ee  
such that for every $n\in \Z$ 
 \begin{equation} \label{eq:spec_rep}
\sum_{x_\perp,y_\perp\in \Z^{d-1}}v_{x_\perp}\overline{v_{y_\perp}}S_{\rho,\beta}((n,x_\perp-y_\perp))=   
 \int_{1/\xi(\rho,\beta)}^\infty e^{-a|n|} d\mu_{v,\beta} ( a) \, .
 \end{equation}
And for every $p_\perp \in[-\pi,\pi]^{d-1}$  there exists a positive measure  $\mu_{p_\perp,\beta}$ of  finite mass such that for every $p_1 \in[-\pi,\pi]$
          \be \label{eq:repr}
        \widehat S_{\rho,\beta}(p) =\int_0^\infty \frac{ e^a - e^{-a}}
         {\mathcal E_1(p_1)  + ( e^{a/2} - e^{-a/2})^2}d\mu_{p_\perp,\beta}(a),      \ee
with $\mathcal E_1(k) := 2[1- \cos(k)]= 4   \sin^2(k/2).$
\end{proposition}

Although the spectral representation is quite well-known (cf.~\cite{GliJaf73} and references therein)  for completeness of the presentation we include the derivation of \eqref{eq:spec_rep} in the Appendix.      
Equation \eqref{eq:repr} then follows by applying \eqref{eq:spec_rep} to   the function 
$$ v_{p_\perp}(x_\perp):=  
\frac{1}{\sqrt{|\Lambda^{(d-1)}_\ell|}} e^{ip_\perp \cdot x_\perp} \mathbb I\left[ x\in \Lambda_\ell^{(d-1)}\right]$$
and taking the limit $\ell\to \infty$.  Here  $\Lambda^{(d-1)}_\ell$ is the $d-1$ dimensional version of the box $\Lambda_L$ and 
$\mathbb I [\cdot]$ is the indicator function.   The convergence is facilitated by the exponential decay of correlations at  $\beta <\beta_c$. 

Of particular interest for us are the following implications of \eqref{eq:repr} (the  first  was noted and applied  in \cite{GliJaf73}). 

 \begin{proposition}\label{cor:mono}
For {a n.n.f.~model} on $\Z^d$ ($d\ge1$), at any $\beta < \beta_c(\rho)$: 
 \begin{enumerate}[noitemsep]
\item  $\widehat S_{\rho,\beta}(p_1,p_2,\dots,p_d)$ is monotone decreasing in each $|p_j|$, over $[-\pi,\pi]$,  
\item $\calE_1(p_1)\widehat S_{\rho,\beta}(p)$ and $|p_1|^2\widehat S_{\rho,\beta}(p_1)$ 
are monotone increasing in  $|p_1|$, 
\item   the function 
\be \label{eq:evenFT}  
\widehat S^{(mod)}_{\rho,\beta}(p) := \widehat S_{\rho,\beta}(p)+\widehat S_{\rho,\beta}(p+\pi(1,1,0,\dots,0))
\ee  is monotone decreasing in $\delta$ along the line of constant $ \{p_3,\dots,p_d\} $  and  
 \be   
 (p_1,p_2) \ =\    (|p_1-p_2|+\delta,|p_1-p_2|-\delta)  \,,  \quad \delta \in [0, |{p_1-p_2}|] \,. 
  \ee    
\end{enumerate}
and the above remains  true under any permutation of the indices. 
 \end{proposition}
 
 The correction in \eqref{eq:evenFT} is insignificant in the regime where  $\widehat S_{\rho,\beta}(p) $ is large.  That is so since  $|\widehat S_{\rho,\beta}(p+\pi(1,1,0,\dots,0))| \leq  C/\beta$  
uniformly for  $|p| \leq \pi/2$. (The main term diverges in the limit  $\beta \nearrow \beta_c(\rho)$ and $p\to 0$.)
 
  \begin{proof}
The first two statements are implied by the combination of   \eqref{eq:repr} with the observation that each of the following functions is monotone in  $k\in[0,\pi]$:
$k\mapsto \calE_1(k)$, 
$k\mapsto k^2/\calE_1(k)$, and for each $a\ge0$, $k\mapsto \calE_1(k)/(\mathcal E_1(k)  +   ( e^{a/2} - e^{-a/2})^2)$.

The third statement is based on the application of the transfer matrix in the diagonal direction  (cf. Fig.~\ref{fig:rotation}). More explicitly, to produce the spectral representation one may start by considering a partially rotated rectangular region, whose main axes are associated with the coordinate system  $(x_1+x_2, x_1-x_2, x_3,\dots, x_d)$.   The finite-volume Hamiltonian is taken with the correspondingly modified periodic boundary conditions which produce 
cyclicity in these  directions.   As stated in \eqref{eq:FTconv}, for $\beta <\beta_c(\rho)$ the change does not affect the two-point function's infinite volume limit.

In this case, there are two transfer matrices $T$ and $T^*$ corresponding to adding one layer of even (resp.~odd) vertices, i.e.~vertices with $x_1+x_2$ even (resp.~odd). 
The argument by which monotonicity was proven above for the Cartesian directions
 applies to the two-point function's restriction to the sub-lattice of even vertices since the proof would involve the  matrix  $T T^*$, which is positive.  
 
 Then, if  $ \widehat S^{(mod)}_{\rho,\beta}$ is given by \eqref{eq:evenFT}, one finds     
\begin{eqnarray} 
\widehat S^{(mod)}_{\rho,\beta}(p)  &=&  \sum_{x\in \Lambda_L} e^{i p\cdot x}   S_{\rho,\beta}(0,x) \sum_{k=0,1} e^{i\pi(x_1+x_2) k}= 2 \sum_{\substack {x\in \Lambda_L\\    x_1+x_2   \text{ even}}}e^{ip\cdot x}   S_{\rho,\beta}(0,x)  .
\end{eqnarray} 
Thus, the third monotonicity statement follows by a direct adaptation of the proof of the first one. 
\end{proof}
\begin{figure}
\begin{center}
\includegraphics[width = 0.40\textwidth]{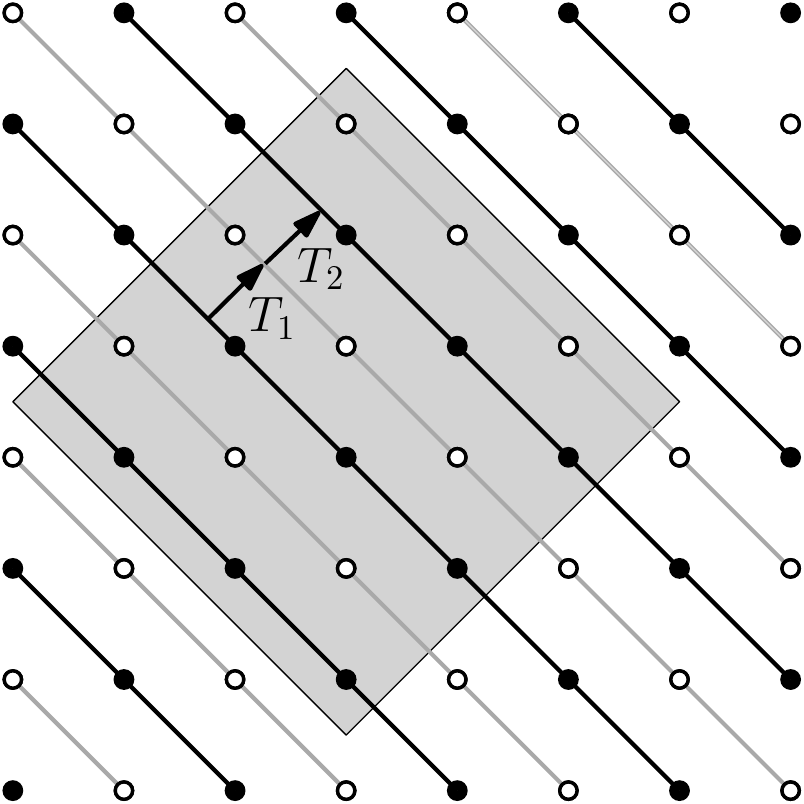}
\caption{The split of $Z^2$  into even and odd sub-lattices and their stratification into intertwined diagonal hyperplanes.  The partition function of $(-L,L]^2$ with rotated-periodic boundary condition can be written as $Z_{2L}= {\rm tr} (T_2 T_1)^L$, with 
$T_j$ a pair of  conjugate mapping between the Hilbert spaces of the even and the odd hyperplanes.   The product  $T_2 T_1=T^*_1 T_1$  provides the even subgraph's transfer matrix  in one of this graph's principal directions. }
\label{fig:rotation}
\end{center}
\end{figure}

\begin{corollary} \label{cor:monotonicity FT}
For {a n.n.f.~model} on $\Z^d$ ($d\ge1$) at any $\beta<\beta_c(\rho)$, the two-point function satisfies, for all $p \in [-\pi/2, \pi/2]^d$,  
\be \label{eq:mono1}
\widehat S_{\rho,\beta}(   \|p\|_\infty, 0_\perp)    \  \geq  \ \widehat S_{\rho,\beta}(p)\  \geq \ \widehat S_{\rho,\beta}( \|p\|_1, 0_\perp)   - \frac{C}{\beta},
\ee
with $C$ depending on the dimension only.
\end{corollary} 

The restriction to $p\in[-\pi/2,\pi/2]^d$ guarantees that the second term of \eqref{eq:evenFT} can be bounded by $C/\beta$ (this bounds corresponds to the $-C/\beta$ term on the right-hand side of \eqref{eq:mono1}), as explained below Proposition~\ref{cor:mono}.
\begin{proof} 
The inequality follows from  Proposition~\ref{cor:mono} through the monotonicity lines used for \eqref{eq:b55}.  
\end{proof} 
The previous bound combined with the second statement in  
Proposition~\ref{cor:mono} yield an interesting consequence for  the  behaviour of the {\em susceptibility} truncated at a distance $L$, which we define as 
 \be
 \chi_L(\rho,\beta):=\sum_{x\in\Lambda_L}S_{\rho,\beta}(x).
 \ee
\begin{theorem}[Sliding-scale Infrared Bound]\label{prop_mon_chi}
There exists a constant $C=C(d)>0$ such that for {every n.n.f.~model} on $\Z^d$ ($d>2$), every $\beta\le \beta_c(\rho)$ and $L\ge \ell\ge 1$,
\be\label{eq:multi chi}
\frac{\chi_L(\rho,\beta)}{L^2}\le \frac{C}{\beta}\frac{\chi_\ell(\rho,\beta)}{\ell^2}.
\ee
\end{theorem}

The case $\ell=1$ is in essence similar to the Infrared Bound~\eqref{eq_IR}, as is explained below, so that \eqref{eq:multi chi} may be viewed as a \emph{sliding-scale}  version of this inequality.   One may also note that \eqref{eq:multi chi}  is a sharp improvement (replacing the exponent $d$ by $2$) on the more naive application of the Messager-Miracle-Sole inequality  giving that for every $L\ge \ell\ge1$,
\be\label{eq:comparison chi}
\frac{\chi_L(\rho,\beta)}{L^d}\le \frac{\chi_\ell(\rho,\beta)}{\ell^d}.
\ee

\begin{proof}   Let us first note that it suffices to prove the claim for all $\beta<\beta_c(\rho)$, with a uniform constant $C$.   Its extension to the critical point can be deduced from the continuity 
\be
S_{\rho,\beta}(x) \ = \ \lim_{\beta\to \beta_c(\rho)} S_{\beta_c(\rho)} (x) 
\ee 
 (which follows from the main result of \cite{AizDumSid15}).    This observation allows us to apply the monotonicity results discussed above. 
  
Below, the constants $C_i$ are to be understood as dependent on $d$ only. Consider the smeared version of $\chi_L(\rho,\beta)$ defined by
\begin{equation}
\widetilde\chi_L(\rho,\beta):=\sum_{x\in \bbZ^d}e^{-(\|x\|_2/L)^2}S_{\rho,\beta}(x).
\end{equation}
with $\|p\|_2^2:=\sum_{i=1}^dp_i^2$.
The MMS monotonicity statement~\eqref{eq:MMS3}  implies that 
\be
e^{-d}\chi_L(\rho,\beta)\le \widetilde\chi_L(\rho,\beta)\le C_1\chi_L(\rho,\beta)
\ee for every $L$, 
so that it suffices to prove that for every $L\ge \ell\ge1$,
\be\label{eq:okk}
\frac{\widetilde\chi_L(\rho,\beta)}{L^2}\le C_2\frac{\widetilde\chi_\ell(\rho,\beta)}{\ell^2}.
\ee
We will work in Fourier space, and use the identity
\be
\widetilde\chi_L(\rho,\beta)~{\asymp}~  
L^d \int_{[-\pi,\pi]^d} e^{-\|p\|_2^2L^2} \, \widehat S_{\rho,\beta}(p) dp \,,
\ee
where $f\asymp g$ means $cg\le f\le Cg$ with $c,C$ independent of everything else (we use that the Fourier transform of the Gaussian on the lattice is a Jacobi theta-function within multiplicative constants of $e^{-\|p\|_2^2L^2}$ on $[-\pi,\pi]^d$).

Now, let 
\be
A:=\{p\in [-\tfrac{\pi\ell}L,\tfrac{\pi\ell}L]^d:\,|p_1|=\|p\|_\infty\}.
\ee Using the symmetries of $\widehat S_{\rho,\beta}$ and the decay of Corollary~\ref{cor:monotonicity FT}, we find that 
\be \label{eq:okkk}
\int_{[-\pi,\pi]^d} e^{-\|p\|_2^2L^2} \,\widehat S_{\rho,\beta}(p) dp \le (d+C_3e^{-\ell^2})\,\int_{A} e^{-\|p\|_2^2L^2} \, \widehat S_{\rho,\beta}(p) dp.
\ee
Since $|p_1|=\|p\|_\infty$ for $p\in A$ and $\|p\|_\infty\ge \|p\|_1/d$, the second property of Proposition~\ref{cor:mono} and Corollary~\ref{cor:monotonicity FT} give that 
\be
\widehat S_{\rho,\beta}(p)\le \widehat S_{\rho,\beta}( \|p\|_\infty, 0_\perp)\le (d\tfrac {L}\ell)^2\widehat S_{\rho,\beta}(\tfrac L\ell \|p\|_1, 0_\perp)\le  (d\tfrac {L}\ell)^2(\widehat S_{\rho,\beta}(\tfrac L\ell p)+C/\beta).
\ee
Using this inequality and making the change of variable $p\mapsto q=\tfrac L\ell p$ gives
\begin{align}
\int_{A} \exp[-\|p\|_2^2L^2] \widehat S_{\rho,\beta}(p) dp&\le C_4\,
(\tfrac \ell L)^{d-2} \Big(\int_{[-\pi,\pi]^d} \exp[-\|q\|_2^2\ell^2]\widehat S_{\rho,\beta}(q) dq+C_5/\beta\Big),
 \end{align}
which after plugging in \eqref{eq:okkk} and taking the Fourier transform implies that 
\be
\widetilde\chi_L(\rho,\beta)\le C_6(\tfrac \ell L)^{d-2}(\widetilde \chi_\ell(\rho,\beta)+C_5/\beta).
\ee
The inequality \eqref{eq:okk} follows from the fact that $\widetilde\chi_\ell(\rho,\beta)\ge 1$, so that the constant $C_5/\beta$ can be removed by changing $C_6$ into a larger constant $C_7/\beta$. 
\end{proof}

Inequality~\eqref{eq:MMS3}  and then the sliding-scale Infrared Bound with $L=|x|$ and $\ell=1$ \eqref{eq:multi chi} implies that for every $x\in\Z^d$,
\begin{equation}\label{eq:IB GS}
S_{\rho,\beta}(x)\ \le \  \frac{C_1}{|x|^{d}}\sum_{y\in {\rm Ann}(d|x|,2d|x|)}S_{\rho,\beta}(y)\ \le\  \frac{C_1}{|x|^{d}}\chi_{2d|x|}(\rho,\beta)\ \le\   \frac{C_2\langle\tau_0^2\rangle_{\rho,\beta} }{|x|^{d-2}}.
\end{equation}
The factor $\langle\tau_0^2\rangle_\beta$ in the upper bound may seem pointless for the Ising model where it is simply equal to 1, but it  becomes very important when studying unbounded spins, as  in Section~\ref{sec:6}, where it is essential for a dimensionless improved tree diagram bound.

It may be noted that the combination of   \eqref{eq:IB GS} with \eqref{eq:bound tau square}  leads to the more standard formulation \cite{FILS78,FroSimSpe76} of the Infrared Bound in $x$-space:
\be \label{eq:IB} 
S_{\rho,\beta}(x) \le
\frac{C}{\beta |J|\,|x|^{d-2}}.
\ee   


\subsection{A lower bound} 

The above upper bound will next be supplemented by a  power-law lower bound on the two point function at $\beta_c$.   Conceptually, it originates in the observation that it the correlations  drop on some scale by a fast enough power law then on larger scales they decay exponentially fast.    An  early version of this principle can be found in Hammersley's analysis of percolation~\cite{Ham57}.   A general statement was presented in Dobrushin's analysis of the constructive criteria for the high temperature phase.   For Ising systems a simple version of such  statement can be deduced from the following observation.

\begin{lemma} \label{lem_SL}
For every ferromagnetic {model in the GS class} on $\Z^d$ ($d\ge1$) with coupling constants that are invariant under translations, every finite $0\in\Lambda \subset \Z^d$ and every $y\notin \Lambda$,
 \be\label{eq:genSL}
\ \qquad S_{\rho,\beta}(y)\le \sum_{\substack{u\in\Lambda\\ v\notin\Lambda}}S_{\rho,\beta}(u)\,\beta J_{u,v}\,S_{\rho,\beta}(y-v).
\ee
\end{lemma} 

This statement is a mild extension of Simon's inequality which was originally formulated for the n.n.f.~Ising models~\cite{Sim80}.  Being spin-dimension balanced,  it is valid also for the Griffiths-Simon class of variables  and more general pair interactions\footnote{The factor $S_{\rho,\beta}(u)$ in \eqref{eq:genSL} can also be replaced  by the finite volume expectation $\langle \tau_0 \tau_u\rangle_\Lambda$, as in Lieb's improvement of Simon's inequality~\cite{Lie80}.     
Both versions have an easy proof through a simple application of the switching lemma, in its mildly improved form.}.

The MMS monotonicity allows us to extract the following point-wise implication, which will be used below 
  
\begin{corollary}[Lower bound on $S_{\rho,\beta}$]\label{cor:SL}
For {a n.n.f.~model in the GS class} on $\Z^d$ ($d\ge1$), there exists $c=c(d)>0$ such that for every $\beta\le\beta_c(\rho)$ and  $x\in\Z^d$,
\be\label{eq:Simon-Lieb}
S_{\rho,\beta}(x)\ge \frac{c}{\beta |J|\,\|x\|_\infty^{d-1}}\exp\big(-\frac{d\|x\|_\infty+1}{\xi(\rho,\beta)}\big).
\ee
\end{corollary}
\begin{proof}

Let us introduce
\be
Y_{\rho,\beta}(\Lambda):=\sum_{\substack{u\in \Lambda\\ v\notin \Lambda}}S_{\rho,\beta}(u)\,\beta J_{u,v}.
\ee
Set $L:=d\|x\|_\infty$. Applying  \eqref{eq:genSL} with $\Lambda_L$ and $y=n{\bf e}_1$, and iterating it $\lfloor \tfrac{n}{L+1}\rfloor$ times (i.e.~as many as possible without reducing the last factor to a distance shorter than $L$),   we get
\be
\beta |J| S_{\rho,\beta}(n{\bf e}_1)\le Y_{\rho,\beta}(\Lambda_{L})^{\lfloor \tfrac{n}{L+1}\rfloor}.
\ee
Since $ \lim_n S_{\rho,\beta}(n{\bf e}_1)^{1/n} = e^{-1/\xi(\rho,\beta)}$, we deduce that  
\be \label{eq_Y}
e^{-1/\xi(\rho,\beta)}   \leq   Y_{\rho,\beta}(\Lambda_{L})^{ \tfrac{1}{L+1} }\,. 
\ee 
On the other hand, by \eqref{eq:MMS3}, for each $x\in \Z^d$,   $ S_{\rho,\beta}(u) \leq S_{\rho,\beta}(x) $  for all $u\in \partial\Lambda_L$, and hence
\begin{equation}\label{eq:SLa} \frac{Y_{\rho,\beta}(\Lambda_{L})}{|\partial\Lambda_{L}|} \leq \beta  |J|S_{\rho,\beta}(x)  \, .\end{equation}
The substitution of  \eqref{eq_Y} in \eqref{eq:SLa} yields the claimed lower bound \eqref{eq:Simon-Lieb}.
\end{proof}


\subsection{Regularity of the two-point function's gradient}\label{sec:3.3}

\begin{proposition}[gradient estimate]\label{prop:regularity gradient}
There exists $C=C(d)>0$ such that for {every n.n.f.~model in the GS class}, every $\beta\le\beta_c(\rho)$, every $x\in \Z^d$  and every $1\le i\le d$,
\be \label{eq_reg}
|S_{\rho,\beta}(x\pm{\bf e}_i)-S_{\rho,\beta}(x)|\le \frac{F(|x|)}{|x|}S_{\rho,\beta}(x),
\ee
where 
\be
F(n):=C\frac{S_{\rho,\beta}(dn{\bf e}_1)}{S_{\rho,\beta}(n{\bf e}_1)}\log\Big(\frac{2S_{\rho,\beta}(\tfrac{n}2{\bf e}_1)}{S_{\rho,\beta}(n{\bf e}_1)}\Big).
\ee
\end{proposition}

The previous proposition  is particularly interesting when
$
S_{\rho,\beta}(2dn{\bf e}_1)\ge c_0 S_{\rho,\beta}(\tfrac n2{\bf e}_1)
$, in which case we obtain the existence of a constant $C_0=C_0(c_0,d)>0$ such that for every $x\in\partial\Lambda_n$ and $1\le i\le d$,
\be\label{eq:gradient estimate}
|S_{\rho,\beta}(x\pm{\bf e}_i)-S_{\rho,\beta}(x)|\le \frac{C_0}{|x|}S_{\rho,\beta}(x).
\ee

\begin{proof}
Without loss of generality, we may assume that $x=(|x|,x_\perp)$. We first assume that $i=1$. 
Introduce the three sequences $u_n:=S_{\rho,\beta}(n{\bf e}_1)$, $v_n:=S_{\rho,\beta}((n,x_{\perp}))$ and $w_n:=u_n+v_n$. The spectral representation applied to the function $v$ being the sum of the Dirac functions at $0_\perp$ and $x_\perp$ implies  the existence of a finite measure $\mu_{x_\perp,\beta}$ such that 
\be
w_n=\int_0^\infty e^{-na}d\mu_{x_\perp,\beta}(a).
\ee
Cauchy-Schwarz gives $w_n^2\le w_{n-1}w_{n+1}$, which when 
iterated between $n$ and $n/2$ (assume $n$ is even, the odd case is similar) leads to
\be
\frac{w_{n+1}}{w_n}\ge \big(\frac{w_n}{w_{n/2}}\big)^{2/n}\ge 1-\tfrac2n\log\big(\frac{w_{n/2}}{w_n}\big).
\ee
We now use that $u_{n/2}\ge v_{n/2}$, $u_{n}\ge v_{n}$, and $u_n\ge u_{n+1}$ which are all consequences of the Messager-Miracle-Sole inequality. Together with trivial algebraic manipulations, we get
\be
v_{n+1}\ge v_n-\frac{4\log(2u_{n/2}/u_n)}{n} u_n.
\ee
The bound we are seeking corresponds to $n=|x|$.

To get the result for $i\ne 1$, use the Messager-Miracle-Sole inequality applied twice to get that
\be
|S_{\rho,\beta}(x\pm {\bf e}_i)-S_{\rho,\beta}(x)|\le S_{\rho,\beta}(x-d{\bf e_1})-S_{\rho,\beta}(x+d{\bf e}_1),
\ee
and then refer to the previous case to conclude (one obtains the result for $n=|x|-d$, but the proof can be easily adapted to get the result for $n=|x|$).
\end{proof}

\begin{remark}When $x=n{\bf e}_1$ and $i=1$, running through the lines of the previous proof shows that one can take $F(n)=2\log(S_{\rho,\beta}(n{\bf e}_1)/S_{\rho,\beta}(\tfrac{n}2{\bf e}_1))$ which is bounded by $(2+o(1))\log n$ thanks to the lower bound \eqref{eq:Simon-Lieb} and the Infrared Bound \eqref{eq:IB}. We therefore get that for every $n\le\xi(\rho,\beta)$,
\be
S_{\rho,\beta}(n{\bf e}_1)-S_{\rho,\beta}((n+1){\bf e}_1)\le (2+o(1))\frac{\log n}{n}S_{\rho,\beta}(n{\bf e}_1).
\ee
It would be of interest to remove the $\log n$ factor, as this would enable a proof that $S_{\rho,\beta}(n{\bf e}_1)$ does not drop too fast between different scales.
\end{remark}

\subsection{Regular scales}\label{sec:3.4}

Using the dyadic distance scales, we shall now introduce the notion of \emph{regular scales},  which in essence means  that on the given scale the two-point function has the properties which in the conditional proof of Section~\ref{sec:2}, were available under the assumption 
\eqref{eq:S}.

\begin{definition}
Fix $c,C>0$.  
An annular region 
$ {\rm Ann}(n/2,4n)$ is said to be {\em regular} if the following four properties are satisfied:
\begin{itemize}
\item[P1] for every $x,y\in {\rm Ann}(n/2,4n)$, $S_{\rho,\beta}(y)\le CS_{\rho,\beta}(x)$;
\item[P2] for every $x,y\in {\rm Ann}(n/2,4n)$, $|S_{\rho,\beta}(x)-S_{\rho,\beta}(y)|\le \tfrac{C|x-y|}{|x|}S_{\rho,\beta}(x)$;
\item[P3] for every $x\in \Lambda_n$ and $y\notin \Lambda_{Cn}$, $S_{\rho,\beta}(y)\le \tfrac12S_{\rho,\beta}(x)$;
\item[P4] $\chi_{2n}(\rho,\beta)\ge (1+c)\chi_n(\rho,\beta)$.
\end{itemize}
A scale $k$ is said to be {\em regular} if the above holds for $n=2^k$, and a vertex $x\in\Z^d$ will be said to be in a regular scale if it  belongs to an annulus with the above properties. 
\end{definition}

One may note that P1 follows trivially from P2 but we still choose to state the two properties independently (the proof would work with weaker versions of P2 so one can imagine cases where the notion of regular scale could be used with a different version of P2 not implying P1).

Under the power-law assumption \eqref{eq:S} of Section~\ref{sec:2} every scale is regular at criticality. However, for now we do not have an unconditional proof of that.   For an unconditional proof of our main results, this gap will be addressed through the following statement, which is the main result of this section. 

 \begin{theorem}[Abundance of regular scales]\label{thm:regular scales}
Fix $d>2$ and $\alpha>2$. There exist $c=c(d)>0$ and $C=C(d)>0$ such that for {every n.n.f.~model in the GS class} and every $n^\alpha\le N\le\xi(\rho,\beta)$, there are at least 
$c\log_2(N/n)$ regular scales $k$ with $n\le 2^k\le N$.
 \end{theorem}

\begin{proof}
The lower bound~\eqref{cor:SL} for $S_{\rho,\beta}$ and the Infrared Bound \eqref{eq:IB} imply that 
\be
\chi_N(\rho,\beta)\ge c_0N\ge c_0(N/n)^{(\alpha-2)/(\alpha-1)}n^2\ge c_1(N/n)^{(\alpha-2)/(\alpha-1)}\chi_n(\rho,\beta).
\ee
Using the sliding-scale Infrared Bound \eqref{eq:comparison chi}, there exist $r,c_2>0$ (independent of $n,N$) such that there are at least $c_2\log_2(N/n)$ scales $m=2^k$ between $n$ and $N$ such that 
\be\label{eq:kkl}
\chi_{rm}(\rho,\beta)\ge \chi_{4dm}(\rho,\beta)+\chi_m(\rho,\beta).
\ee
 Let us verify that the different properties of regular scales are satisfied for such an $m$. 
Applying \eqref{eq:MMS3} in the first inequality,   
the assumption \eqref{eq:kkl} in the second, and \eqref{eq:MMS3} in the third,
one has 
\begin{align}|{\rm Ann}(4dm,rm)|S_{\rho,\beta}(4dm{\bf e}_1)
&\ge \chi_{rm}(\rho,\beta)-\chi_{4dm}(\rho,\beta)\ge \chi_m(\rho,\beta)\ge |\Lambda_{m/(4d)}|S_{\rho,\beta}(\tfrac14m{\bf e}_1).\end{align}
This implies that $S_{\rho,\beta}(4dm{\bf e}_1)\ge c_0S_{\rho,\beta}(\tfrac14m{\bf e}_1)$, which immediately gives P1 by \eqref{eq:MMS3}   for $S_{\rho,\beta}$ and P2 by the gradient estimate given by Proposition~\ref{prop:regularity gradient}.  Furthermore, the fact that $S_{\rho,\beta}(x)\ge S_{\rho,\beta}(4dm{\bf e}_1)\ge \frac{c_3}{m^d}\chi_m(\rho,\beta)$ for every $x\in {\rm Ann}(m,2m)$ implies P4. To prove P3, observe that for every $R$, the previous displayed inequality together with the sliding-scale Infrared Bound \eqref{eq:multi chi} give that for every $y\notin \Lambda_{dRm}$ and $x\in \Lambda_m$,
\begin{align}
|\Lambda_{Rm}|S_{\rho,\beta}(y)\le \chi_{Rm}(\rho,\beta)\le C_4R^{2}\chi_{m}(\rho,\beta)\le C_5R^{2}m^dS_{\rho,\beta}(x), \end{align}
which implies the claim for $C$ and $c$ respectively large and small enough using here the assumption that $d>2$.
\end{proof}

\section{Unconditional proofs of the Ising's results}\label{sec:4}

In this section, we prove our results for every $\beta\le\beta_c$ without making the power-law assumption of Section~\ref{sec:2}.   We emphasize that unlike the introductory discussion of that  section, the proofs given below are unconditional.  The discussion is also not restricted to the critical point itself and covers more general approaches of the scaling limits, from the side $\beta\leq \beta_c$ (hence the correlation length will be mentioned in several places).  However, at this stage the discussion is still restricted to the n.n.f.~Ising model.

\subsection{Unconditional proofs of the intersection-clustering bound and Theorem~\ref{thm:improved tree bound simple} for the Ising model}\label{sec:4.1}

The notation remains  as in Section~\ref{sec:2}. The endgame in this section will be the unconditional proof of the intersection-clustering bound  that we restate below in the right level of generality. The main modification is that the sequence $\calL$ of integers $\ell_k$ will be   chosen dynamically, adjusting it to the behaviour of the two-point function. More precisely, recall the definition of the bubble diagram $B_L(\beta)$ truncated at a distance $L$. Fix $D\gg1$ and define recursively a (possibly finite) sequence $\mathcal L$ of integers $\ell_k=\ell_k(\beta,D)$ by the formula $\ell_0=0$ and
\be\label{eq:def ell}
\ell_{k+1}=\inf\{\ell:B_\ell(\beta)\ge DB_{\ell_k}(\beta)\}.
\ee
By the Infrared Bound \eqref{eq:IB}, $B_L-B_\ell \leq C_0\log(L/\ell)$ (in  dimension $d=4$) from which it is a simple exercise to deduce that under the above definition 
\be 
D^k\le B_{\ell_k}(\beta)\le CD^k 
\ee 
 for every $k$ and some large constant $C$ independent of $k$.

\begin{proposition}[clustering bound]\label{prop:non isolated general}  
 For $d=4$ and $D$ large enough, there exists $\delta=\delta(D)>0$ such that  for every $\beta\le\beta_c$,  every $K>3$ with $\ell_K\le \xi(\beta)$, and every $u,x,y,z,t\in\Z^4$ with mutual distances between $x,y,z,t$ larger than $2\ell_{K}$,
\be \label{eq_cluster}
{\bf P}^{ux,uz,uy,ut}_\beta[{\bf M}_{u}(\calT; \mathcal L, K)< \delta K]\le 2^{-\delta K}.
\ee 
\end{proposition}

Before proving this proposition, let us explain how it implies the improved tree diagram bound.
\begin{proof}[Theorem~\ref{thm:improved tree bound simple}]
Choose $D$ large enough that the previous proposition holds true. We follow the same lines as in Section~\ref{sec:2.1}, simply noting that since $B_{\ell_k}(\beta) \le  CD^k$, we may choose $K\ge c\log B_L(\beta)$ with $2\ell_{K}\le L$, where $c$ is independent of $L$ and $\beta$, so that \eqref{eq:p1} implies the improved tree diagram bound inequality. 
\end{proof}

The main modification we need for an unconditional proof of the intersection-clustering bound 
lies in  the derivation of the intersection and mixing properties. 
The former is similar to Lemma~\ref{lem:intersection}, but restricted to sources that lie in regular scales.  We restate it here in a slightly modified form.

Recall that $I_k$ 
is the event  that there exist unique clusters of ${\rm Ann}(\ell_{k},\ell_{k+1})$ in $\n_1+\n_3$ and $\n_2+\n_4$ crossing the annulus from the inner boundary to the outer boundary and that these two clusters are intersecting.
\begin{lemma}[Intersection property]\label{lem:intersection general}
Fix $d=4$. 
There exists $c>0$ such that for every $\beta\le\beta_c$, every $k$, and every $y\notin \Lambda_{2\ell_{k+1}}$ in a regular scale,
\be
{\bf P}^{0y,0y,\emptyset,\emptyset}_\beta[I_k]\ge c.
\ee
\end{lemma}

\begin{proof}
 Restricting our attention to the case of $y$ belonging to a regular scale enables us to use properties P1 and P2 of the regularity assumption on the scale.
With this additional assumption, we follow the same proof as the one of the conditional version (Lemma~\ref{lem:intersection}).
Introduce the intermediary integers $n\le m\le M\le N$ satisfying 
\be 
\ell_k^4\ge n\ge \ell_{k}^{3+\ep}\ ,\  n^4\ge m\ge n^{3+\ep}\ ,\  M^4\ge N\ge M^{3+\ep}\ ,\ N^4\ge \ell_{k+1}\ge N^{3+\ep}.
\ee
For the second moment method on $\calM$, the first and second moments take the following forms
\begin{align}\label{eq:BB}{\bf E}^{0y,0y,\emptyset,\emptyset}_\beta[|\calM|]&\ge c_1 (B_M(\beta) -B_{m-1}(\beta))\ge c_2B_{\ell_{k+1}}(\beta),\\
{\bf E}^{0y,0y,\emptyset,\emptyset}_\beta[|\calM|^2]&\le c_3(B_M(\beta) -B_{m-1}(\beta))B_{2M}(\beta)\le c_3B_{\ell_{k+1}}(\beta)^2,
\end{align}
where in the second inequality of the first line, we used that $D$ is large enough and  Lemma~\ref{cor:growth} below to get that
\[
B_M(\beta)\ge \frac{B_{\ell_{k+1}}(\beta)}{1+15C}\quad\text{ and }\quad B_{m-1}(\beta)\le (1+15C)B_{\ell_k}(\beta)\le \frac{1+15C}{D}B_{\ell_{k+1}}(\beta).
\]
For the bound on the probabilities of the events $F_1,\dots,F_4$ defined as in Section~\ref{sec:2.2}, recall that the vertices $x$ and $z$ there are in our case both equal to $y$ that belongs to a regular scale.  Using Property 2 of the regularity of scales, the bounds in \eqref{eq:bound} and \eqref{eq:bound sourceless} follow readily from the Infrared Bound \eqref{eq:IB}.
\end{proof}

In the previous proof, we used the following statement.
\begin{lemma}\label{cor:growth}
For $d=4$, there exists $C>0$ such that for every $\beta\le \beta_c$ and every $\ell\le L\le\xi(\beta)$,
\be\label{eq:lemma1} B_{L}(\beta)\le \Big(1+C\frac{\log(L/\ell)}{\log \ell}\Big)\,B_{\ell}(\beta).\ee
\end{lemma}

\begin{proof}
For every $n\le N$ for which $n=2^k$ with $k$ regular, we have that (recall the definition of $\chi_n(\beta)$ from the previous section)
\begin{align}B_{2N}(\beta)-B_N(\beta)
&\le C_0N^{-4}\chi_{N/d}(\beta)^2\nonumber\\
&\le C_1n^{-4}\chi_n(\beta)^2\nonumber\\
&\le C_2n^{-4}(\chi_{2n}(\beta)-\chi_{n}(\beta))^2\nonumber\\
&\le C_3(B_{2n}(\beta)-B_n(\beta)),\end{align}
where in the first inequality we used \eqref{eq:MMS3}, in the second the sliding-scaled Infrared Bound \eqref{eq:multi chi}, in the third Property P4 of the regularity of $n$, and in the last Cauchy-Schwarz.

Now, there are $\log_2(L/\ell)$ scales  between $\ell$ and $L$, and at least $\tfrac1C\log_2 \ell$ regular scales between $1$ and $\ell$ by abundance of regular scales (Theorem~\ref{thm:regular scales}). Since the sums of squared correlations on any of the former contribute less to $B_L(\beta)-B_\ell(\beta)$ than any of the latter to $B_{\ell}(\beta)$, we deduce  that 
\be
B_L(\beta)\le \Big(1+C\frac{\log_2(L/\ell)}{\log_2\ell}\Big)B_\ell(\beta).
\ee
\end{proof}

Next comes the unconditional mixing property.
\begin{theorem}[random currents' mixing property]
\label{thm:mixing multi}
For $d\ge4$, there exist $\alpha,c>0$ such that for every $t\le s$, every $\beta\le \beta_c$, every $n^\alpha\le N\le\xi(\beta)$, every $x_i\in\Lambda_n$ and $y_i\notin\Lambda_N$ $(i\le t)$, and every events $E$ and $F$ depending on the restriction of $(\n_1,\dots,\n_s)$ to edges within $\Lambda_n$ and outside of $\Lambda_N$ respectively,  
\begin{align}\label{eq:mixing}\big|{\bf P}^{x_1y_1,\dots,x_ty_t,\emptyset,\dots,\emptyset}_\beta[E\cap F]-{\bf P}^{x_1y_1,\dots,x_ty_t,\emptyset,\dots,\emptyset}_\beta[E]{\bf P}^{x_1y_1,\dots,x_ty_t,\emptyset,\dots,\emptyset}_\beta[F]\big|&\le s(\log \tfrac Nn)^{-c}.\end{align}
 Furthermore, for every $x'_1,\dots,x'_t\in\Lambda_n$ and $y'_1,\dots,y'_t\notin\Lambda_N$, we have that 
\begin{align}\label{eq:independence}\big|{\bf P}^{x_1y_1,\dots,x_ty_t,\emptyset,\dots,\emptyset}_\beta[E]-{\bf P}^{x_1y_1',\dots,x_ty_t',\emptyset,\dots,\emptyset}_\beta[E]\big|&\le s(\log \tfrac Nn)^{-c},\\
\big|{\bf P}^{x_1y_1,\dots,x_ty_t,\emptyset,\dots,\emptyset}_\beta[F]-{\bf P}^{x_1'y_1,\dots,x_t'y_t,\emptyset,\dots,\emptyset}_\beta[F]\big|&\le s(\log \tfrac Nn)^{-c}\label{eq:independence2}.\end{align}
\end{theorem}

We postpone the proof to Section~\ref{sec:4.2} below.  Before showing how Theorem~\ref{thm:mixing multi} is used in the proof of the improved tree diagram bound, let us make an interlude and comment on this statement.
\bigbreak
\noindent{\bf Discussion} 
 The relation \eqref{eq:mixing} is an assertion of approximate independence between events at far distances, and \eqref{eq:independence}--\eqref{eq:independence2} expresses a degree of   independence of the probability of an event  from the precise placement of  the sources when these are far from the event in question.   This result should be of interest on its own, and possibly have other applications, since mixing properties efficiently replace independence in statistical mechanics. 

The main difficulty of the theorem concerns currents with a source inside $\Lambda_n$ and a source outside $\Lambda_N$ (i.e.~the first $t$ ones). In this case, the currents are constrained to have a path linking the two, and that may be a conduit for information, and correlation,  between $\Lambda_n$ and the exterior of $\Lambda_N$.   
To appreciate the point it may be of help to compare the situation  with Bernoulli percolation: there the mixing property without sources is a triviality (by the variables' independence); while  an analogue of the mixing property with sources $x$ and $y$ would concern  Bernoulli percolation conditioned on having a path from $x$ to $y$.  Proving convergence at criticality, for $x$ set as the origin and $y$ tending to infinity, of these conditioned measures is a notoriously hard problem.  It would in particular imply the existence of the so-called Incipient Infinite Cluster (IIC), and the definition of the IIC was justified in 2D \cite{Kes86} and in high dimension \cite{HofJar04}, but it is still open in dimensions $3\le d\le 10$. When the number  of sources is even inside $\Lambda_n$, things become much simpler and one may in fact prove a quantitative ratio weak mixing using mixing properties for (sub)-critical random-cluster measures with cluster-weight 2 provided by \cite{AizDumSid15}. 

Theorem~\ref{thm:mixing multi}   has an extension to three dimensions using \cite{AizDumSid15}, but there it becomes non-quantitative (the careful reader will notice that the condition $d>3$ is coming from the exponent appearing in the proof of \eqref{eq:aahaha} in Lemma~\ref{lem:G} in the next section). More precisely, one may prove that in dimension $d=3$, for every $n$, $s$ and $\ep$, there exists a constant $N$ sufficiently large that the previous theorem holds with an error $\ep$ instead of $s(\log \tfrac Nn)^{-c}$. This has a particularly interesting application: one may construct the IIC in dimension $d=3$ for this model, since the random-cluster model with cluster weight $q=2$ conditioned on having a path from $x$ to $y$ can be obtained as the random current model with sources $x$ and $y$ together with an additional independent sprinkling (see \cite{AizDumTasWar18}).  This represents a non-trivial result for critical 3D Ising. More generally, we believe that the previous mixing result may be a key tool in the rigorous description of the critical behaviour of the Ising model in three dimensions.

This concludes the interlude, and we  return now to the proof of the intersection-clustering bound. 

\begin{proof}[Proposition~\ref{prop:non isolated general}]
We follow the same argument as in the proof of the conditional version (Proposition~\ref{prop:non isolated}) and borrow the notation from the corresponding proof at the end of Section~\ref{sec:2.2}. We fix $\alpha>2$ large enough that the mixing property Theorem~\ref{thm:mixing multi} holds true. Using Lemma~\ref{cor:growth}, we may choose $D=D(\alpha)$ such that $\ell_{k+1}\ge \ell_k^\alpha$. 

The proof is exactly identical to the proof of Proposition~\ref{prop:non isolated}, with the exception of the bound on ${\bf P}^{0x,0z,\emptyset,\emptyset}[A_S]$ and the fact that we restrict ourselves to subsets $S$ of even integers in $\{1,\dots,K-3\}$. In order to obtain this result, first observe that since we assumed $\ell_K\le \xi(\beta)$, by Theorem~\ref{thm:regular scales} there exists $y\in {\rm Ann}(\ell_{K-1},\ell_{K})$ in a regular scale. Since the event $A_S$ depends on the currents inside $\Lambda_{\ell_{K-2}}$ (since $S$ does not contain integers strictly larger than $K-3$), and that $\ell_{K-1}\ge \ell_{K-2}^\alpha$, the mixing property (Theorem~\ref{thm:mixing multi}) shows that
\be
{\bf P}^{0x,0z,\emptyset,\emptyset}[A_S]\le {\bf P}^{0y,0y,\emptyset,\emptyset}[A_S]+\frac{C}{\sqrt{\log \ell_{K-1}}}\le {\bf P}^{0y,0y,\emptyset,\emptyset}[A_S]+2^{-\delta K}.
\ee
To derive the first bound on the right-hand side, we apply the mixing property repeatedly (Theorem~\ref{thm:mixing multi}) and the intersection property (Lemma~\ref{lem:intersection general}) exactly as in the conditional proof. For the second inequality, we lower bound $\ell_{K-1}$ using $B_{\ell_{K-1}}(\beta)\ge D^{K-1}$ and the Infrared Bound \eqref{eq:IB}.
\end{proof}

\subsection{The mixing property: proof of Theorem~\ref{thm:mixing multi}}\label{sec:4.2}

As we saw, the mixing property is in the core of the proof of our main result.  The strategy of the proof was explained in Section~\ref{sec:2.2} when we proved mixing for one current under the power-law assumption.  In this section we  again define a random variable ${\bf N}$ which is approximately $1$  and is a weighted sum over ($t$-tuple of) vertices connected to the origin. The main difficulty will come from the fact that since we do not fully control the spin-spin correlations, we  will need to define ${\bf N}$ in a smarter fashion. Also, whereas in Section~\ref{sec:2.2} we  treated the case of a single current ($s=1$), here we generalize to multiple currents.

Fix $\beta\le\beta_c$ and drop it from the notation. Also fix $s\ge  t\ge 
1$ and $n^\alpha\le N\le\xi(\beta)$. Below, constants $c_i$ and $C_i$ are independent of the choices of $s,t,
\beta,n,N$ satisfying the properties above. We introduce the integers $m$ and $M$ such that $m/n=(N/n)^{1/3}$ and $N/M=(N/n)^{1/3}$ (we omit the details of the rounding operation).

For ${\bf x}=(x_1,\dots,x_t)$ and ${\bf y}=(y_1,\dots,y_t)$, we will use the following shortcut notation 
\be
{\bf P}^{\bf xy}:={\bf P}^{x_1y_1,\dots,x_ty_t,\emptyset,\dots,\emptyset}\quad\text{and}\quad{\bf P}^{{\bf xy}}\otimes{\bf P}^\emptyset,
\ee
 where the second measure is the law of the random variable $(\n_1,\dots,\n_s,\n'_1,\dots,\n'_s)$, where  $(\n'_1,\dots,\n'_s)$ is an independent family of sourceless currents. 

To define $\mathbf N$, first introduce for every vertex $y\notin \Lambda_{2dm}$,  the set (see Fig.~\ref{fig:4}) 
\begin{align}
\bbA_{y}(m)&:=\big\{u\in {\rm Ann}(m,2m):\forall x\in \Lambda_{m/d}, \langle\sigma_x\sigma_y\rangle \le \big(1+\tfrac{C|x-u|}{|y|}\big)\langle\sigma_u\sigma_y\rangle\big\},
\end{align}
where $C$ is given by the definition of good scales.
\begin{remark}\label{rmk:aa} When $y$ is in a regular scale, then $\bbA_{y}(m)$ is equal to ${\rm Ann}(m,2m)$ by Property P2 of regular scales. The reason why we  consider $\bbA_{y}(m)$ instead of the full annulus ${\rm Ann}(m,2m)$ is technical: since $y$ will not a priori be assumed to belong to a regular scale (in fact $|y|$ may be much larger than $\xi(\beta)$ when $\beta<\beta_c$), we will use (for \eqref{eq:b0} and \eqref{eq:B bound} below) the inequality between $\langle\sigma_x\sigma_y\rangle$ and $\langle\sigma_u\sigma_y\rangle$ in several bounds. Now,  if $y_1=|y|$, then
\be\label{eq:aah1}
\bbA_{y}(m)\supset\{z\in\Z^d: m\le z_1\le 2m\text{ and }0\le z_j\le m/d\text{ for $j>1$}\}\ee
as the Messager-Miracle-Sole inequality implies\footnote{
The claim follows directly from the inequality $S_\beta(y)\le S_\beta(x)$ for every $x,y$ such that $x_1\ge0$ and  $y_1\ge x_1+\sum_{j>1}|y_j-x_j|$. In order to prove this inequality, define, for $0\le i\le d$,
\be
v^{(i)}:=(x_1+\sum_{j=i}^d|y_j-x_j|,x_2,\dots,x_{i},y_{i+1},\dots,y_d).
\ee
Successive applications of the Messager-Miracle-Sole inequality  with respect to the sum or the difference (depending on whether $x_i$ is positive or negative) of the first and $i$-th coordinates implies that 
\be
S_\beta(y)\le S_\beta(v^{(1)})\le S_\beta(v^{(2)})\le\dots\le S_\beta(v^{(d)})=S_\beta(x).
\ee
} that $\langle\sigma_z\sigma_y\rangle\ge \langle\sigma_x\sigma_y\rangle$ for every $x\in \Lambda_{m/d}$.
\end{remark}
From now on, fix a set $\calK$ of regular scales $k$ with $m\le 2^k\le M/2$ satisfying that distinct $k,k'\in \calK$ are differing by a multiplicative factor at least $C$ (where the constant $C$ is given by Theorem~\ref{thm:regular scales}). We further assume that $|\calK|\ge c_1\log(N/n)$, where $c_1$ is sufficiently small. The existence of $\calK$ is guaranteed by the definition of $m$ and $M$ and the abundance of regular scales given by Theorem~\ref{thm:regular scales}.

Define ${\bf N}:=\prod_{i=1}^t{\bf N}_i$, 
where
\begin{equation}{\bf N}_i:=\frac1{|\calK|}\sum_{k\in \calK}\frac{1}{A_{x_i,y_i}(2^k)}\sum_{u\in\bbA_{y_i}(2^k)}\mathbb I[u\stackrel{\n_i+\n'_i}\longleftrightarrow x_i],
\end{equation}
where
$
a_{x,y}(u):=\langle\sigma_x\sigma_u\rangle\langle\sigma_u\sigma_y\rangle/\langle\sigma_x\sigma_y\rangle$ and $ A_{x,y}(m):=\sum_{u\in \bbA_{y}(m)}a_{x,y}(u)$.
The first step of the proof is the following concentration inequality.
\begin{proposition}[Concentration of ${\bf N}$]\label{lem:concentration}
For every $\alpha>2$, there exists $C_0=C_0(\alpha,t)>0$ such that for every $n$ large enough and $n^\alpha\le N\le\xi(\beta)$,
\begin{equation}
{\bf E}^{{\bf xy},\emptyset}[({\bf N}-1)^2]\le \frac{C_0}{\log(N/n)}.
\end{equation}
\end{proposition}

\begin{proof}
We shall apply the telescopic formula
\[
{\bf N}-1\, \equiv \, \prod_{i=1}^t{\bf N}_i-1 \, =\,  \sum_{i=1}^t({\bf N}_i-1)\prod_{j>i}{\bf N}_j 
\]
with the last product  interpreted as $1$ for $i=t$.
Hence, by the Cauchy-Schwarz inequality and the currents' independence, 
\be{\bf E}^{{\bf xy},\emptyset}[({\bf N}-1)^2]\, \le \, t \, \sum_{i=1}^t{\bf E}^{{\bf xy},\emptyset}[({\bf N}_i-1)^2] 
\,  \prod_{j>i}{\bf E}^{{\bf xy},\emptyset}[{\bf N}_j^2]\,.
\ee
It therefore suffices to show that there exists a constant $C_1>0$ such that for every $i\le t$,
\begin{equation}\label{eq:ui}
{\bf E}^{{\bf xy},\emptyset}[({\bf N}_i-1)^2]\le \frac{C_1}{\log (N/n)}.
\end{equation}
To lighten the notation, and since the random variable ${\bf N}_i$ depends only on $\n_i$ and $\n'_i$, we omit the index in $x_i,y_i,\n_i,\n'_i$ and write instead just $x,y,\n,\n'$.  We keep the index in ${\bf N}_i$ to avoid confusion with ${\bf N}$ which is the product of these random variables. 

The proof of \eqref{eq:ui} is also based on a computation of the first and second moments of ${\bf N}_i$.
For the first moment, the switching lemma and the definition of ${\bf N}_i$ imply that
${\bf E}^{xy,\emptyset}[{\bf N}_i]=1$.
From the lower bound on $|\calK|$, to bound the second moment it therefore suffices to show that 
\begin{align}\label{eq:h1}{\bf E}^{xy,\emptyset}[{\bf N}_i^2]\le 1+\frac{C_2}{|\calK|},\end{align}
which follows from the inequality, for every $\ell\ge k$ in $\calK$, 
\begin{equation}\label{eq:h5}
\sum_{\substack{u\in \bbA_{y}(2^k)\\ v\in\bbA_{y}(2^\ell)}}{\bf P}^{xy,\emptyset}[u,v\stackrel{\n+\n'}\longleftrightarrow x]\le A_{x,y}(2^k)A_{x,y}(2^\ell)(1+C_32^{-(\ell-k)}).
\end{equation}

\paragraph{Case $\ell>k$.} We find by \eqref{eq:prop3b} that 
\begin{align}\label{eq:b0}{\bf P}^{xy,\emptyset}[u,v\stackrel{\n+\n'}\longleftrightarrow x]&\le a_{x,y}(u)a_{x,y}(v)\left(\frac{\langle\sigma_x\sigma_y\rangle}{\langle\sigma_u\sigma_y\rangle}\frac{\langle\sigma_u\sigma_v\rangle}{\langle\sigma_x\sigma_v\rangle}+\frac{\langle\sigma_x\sigma_y\rangle}{\langle\sigma_v\sigma_y\rangle}\frac{\langle\sigma_u\sigma_v\rangle}{\langle\sigma_x\sigma_u\rangle}\right).\end{align}
Now, since $u\in \bbA_{y}(2^k)$, $\langle\sigma_x\sigma_y\rangle\le(1+C\tfrac{|u-x|}{|y|})\langle\sigma_u\sigma_y\rangle$.  Furthermore, since $\ell$ is a regular scale, Property P2 of regular scales  implies that $\langle\sigma_u\sigma_v\rangle\le (1+C\tfrac{|u-x|}{|v|})\langle\sigma_x\sigma_v\rangle$. We deduce that 
\begin{align}\label{eq:b1}\frac{\langle\sigma_x\sigma_y\rangle}{\langle\sigma_u\sigma_y\rangle}\frac{\langle\sigma_u\sigma_v\rangle}{\langle\sigma_x\sigma_v\rangle}\le 1+C_02^{-(\ell-k)}.\end{align}
Similarly,  since $v\in \bbA_{y}(2^\ell)$, $\langle\sigma_x\sigma_y\rangle\le(1+C\tfrac{|v-x|}{|y|})\langle\sigma_v\sigma_y\rangle$.  Property P3 for the $\ell-k$ regular scales in $\calK$ between $k$ and $\ell$ implies that 
\begin{equation}\label{eq:b2}\frac{\langle\sigma_x\sigma_y\rangle}{\langle\sigma_v\sigma_y\rangle}\frac{\langle\sigma_u\sigma_v\rangle}{\langle\sigma_x\sigma_u\rangle}\le C_12^{-(\ell-k)}.\end{equation}
Plugging \eqref{eq:b1}--\eqref{eq:b2}  into \eqref{eq:b0} and summing over $u\in\bbA_{y}(2^k)$ and $v\in\bbA_{y}(2^\ell)$ gives \eqref{eq:h5}.

\paragraph{Case $\ell=k$.} Assume that $\langle\sigma_u\sigma_y\rangle\le \langle\sigma_v\sigma_y\rangle$. Use \eqref{eq:prop3b} to write
\begin{align}
\label{eq:dt}{\bf P}^{xy,\emptyset}[u,v\stackrel{\n+\n'}\longleftrightarrow x]&\le\langle\sigma_v\sigma_u\rangle\left(\frac{\langle\sigma_x\sigma_u\rangle}{\langle\sigma_x\sigma_v\rangle}+\frac{\langle\sigma_u\sigma_y\rangle}{\langle\sigma_v\sigma_y\rangle}\right) a_{x,y}(v).
\end{align}
By Property P1 of regular scales, the first term under parenthesis is bounded by a constant. The second one is bounded by 1 by assumption. Now, for each $v\in \bbA_{y}(2^k)$, 
 \begin{align*}
\sum_{u\in \bbA_{y}(2^k):\langle\sigma_u\sigma_y\rangle\le \langle\sigma_v\sigma_y\rangle}\langle\sigma_v\sigma_u\rangle&\le \chi_{2^{k+1}}(\beta)\le C_2(\chi_{2^{k+1}}(\beta)-\chi_{2^k}(\beta))\\
& \le  C_3\sum_{u\in \bbA_{y}(2^k)}\langle\sigma_0\sigma_u\rangle \\
&\le C_4\sum_{u\in \bbA_{y}(2^k)}\frac{\langle\sigma_x\sigma_u\rangle\langle\sigma_u\sigma_y\rangle}{\langle\sigma_x\sigma_y\rangle}=C_4A_{x,y}(2^k),
\end{align*}
where the first inequality is trivial, the second one is true by Property P4, the third by Remark~\ref{rmk:aa} (when $y$ is regular then it is a direct consequence of P4, and when it is not one can use \eqref{eq:aah1} and the Messager-Miracle-Sole inequality), and the fourth inequality follows from Property P1 of regular scales (to replace $\langle\sigma_0\sigma_u\rangle$ by $\langle\sigma_x\sigma_u\rangle$) and the fact that since $u\in \bbA_y(2^k)$,  $\langle\sigma_x\sigma_y\rangle\le (1+C|u-x|/|y|)\langle\sigma_u\sigma_y\rangle\le C_5\langle\sigma_u\sigma_y\rangle$. 

We deduce that 
\begin{align}\label{eq:b3}\sum_{u,v\in\bbA_{y}(2^k)}{\bf P}^{xy,\emptyset}[u,v\stackrel{\n_1+\n_2}\longleftrightarrow x]\le 2\sum_{\substack{u,v\in\bbA_{y}(2^k)\\
\langle\sigma_u\sigma_y\rangle\le \langle\sigma_v\sigma_y\rangle}}{\bf P}^{xy,\emptyset}[u,v\stackrel{\n_1+\n_2}\longleftrightarrow x]&\le  C_6A_{x,y}(2^k)^2.\end{align}
\end{proof}

For a proof of Theorem~\ref{thm:mixing multi} we fix $\alpha>2$ (which will be taken large enough later).   Applying the Cauchy-Schwarz inequality gives
\begin{align}\label{eq:gf}|{\bf P}^{{\bf xy}}[E\cap F]-{\bf E}^{{\bf xy},\emptyset}[{\bf N}\mathbb I_{(\n_1,\dots,\n_s)\in E\cap F}]|&\le \sqrt{{\bf E}^{{\bf xy},\emptyset}[({\bf N}-1)^2]}\le \frac{C_1}{\sqrt{\log(N/n)}}.
\end{align}
Now, for ${\bf u}=(u_1,\dots,u_t)$ with $u_i\in{\rm Ann}(m,M)$ for every $i$, let $G(u_1,\dots,u_t)$ be the event that for every $i\le s$, there exists ${\bf k}_i\le \n_i+\n'_i$ such that ${\bf k}_i=0$ on $\Lambda_n$, ${\bf k}_i=\n_i+\n'_i$ outside $\Lambda_N$, and $\partial{\bf k}_i$ is equal to $\{u_i,y_i\}$ if $i\le t$ and $\emptyset$ if $t<i\le s$.
The switching principle implies as in Section~\ref{sec:3.2} that 
\begin{align}\label{eq:ggf}{\bf P}^{{\bf xy},\emptyset}&[(\n_1,\dots,\n_s)\in E\cap F,u_i\stackrel{\n_i+\n'_i}\longleftrightarrow x_i\text{ for }i\le t,G(u_1,\dots,u_t)]\nonumber\\
&=\Big(\prod_{i=1}^ta_{x_i,y_i}(u_i)\Big)\,{\bf P}^{{\bf xu},{\bf uy}}[(\n_1,\dots,\n_s)\in E,(\n'_1,\dots,\n'_s)\in F,G(u_1,\dots,u_t)].\end{align}
 Also, as before, we have the trivial identity 
\begin{equation}\label{eq:ggff}{\bf P}^{{\bf x u},{\bf uy}}[(\n_1,\dots,\n_s)\in E,(\n'_1,\dots,\n'_s)\in F]={\bf P}^{{\bf xu}}[E]{\bf P}^{\bf uy}[F].\end{equation}

\begin{figure}
\begin{center}
\includegraphics[width = 0.63\textwidth]{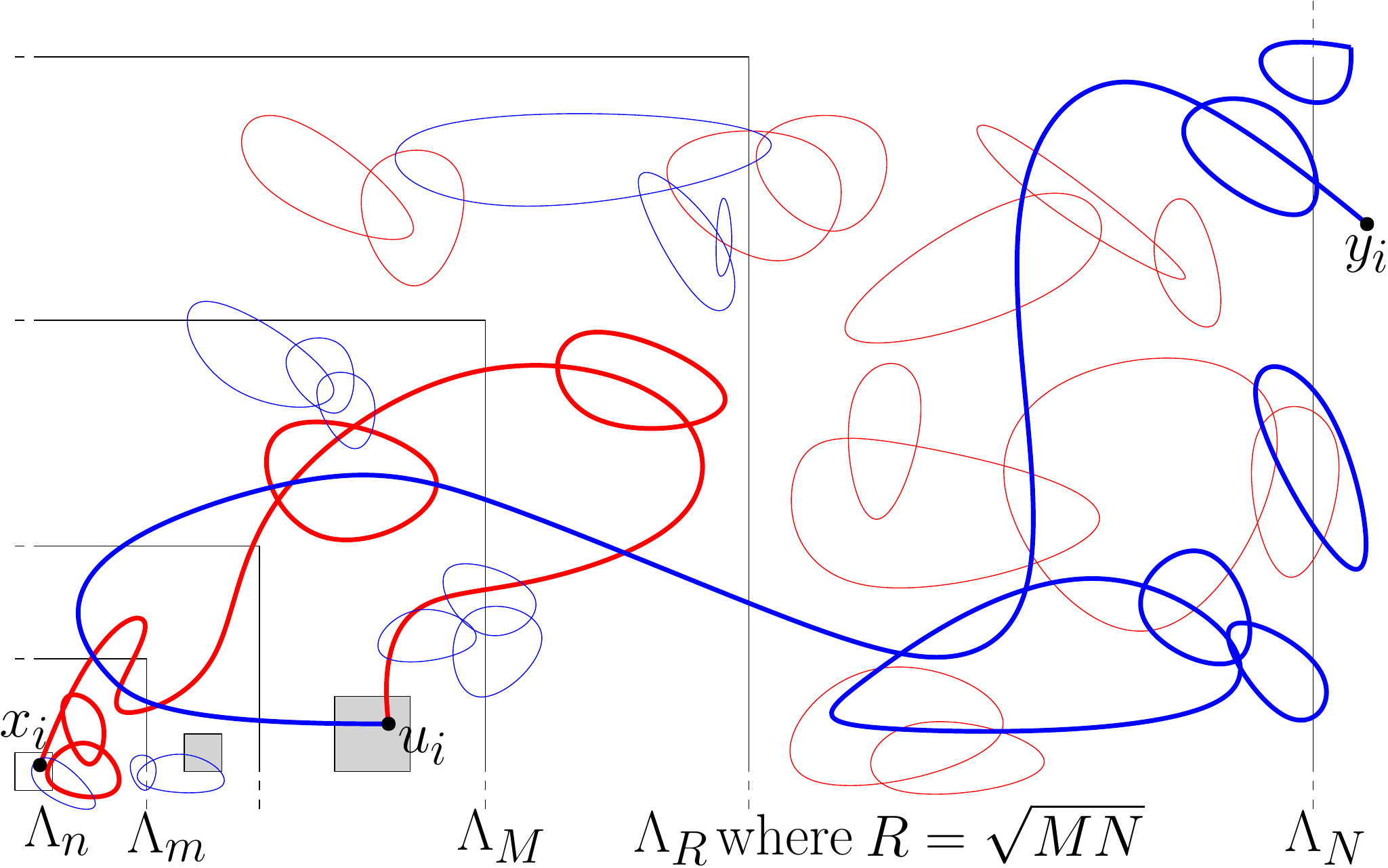}
\caption{The currents $\n_i$ (red) and $\n'_i$ (blue).  Since the sources of $\n_i$, i.e.~$x_i$ and $u_i$,  are both in $\Lambda_M$, a reasoning similar to the proof of uniqueness in the intersection property (first control the backbone, proving that it does not cross the annulus ${\rm Ann}(M,R)$, and then the remaining sourceless current) enables us to conclude that the probability that the current contains a crossing of ${\rm Ann}(M,N)$ is small.   Similarly, since the sources $u_i$ and $y_i$ of $\n'_i$ lie both outside of $\Lambda_m$, we can prove that the probability that $\n'_i$ crosses ${\rm Ann}(n,m)$ is small.  An extra care is needed for establishing  the latter since $y$ is not assumed to be regular.  To circumvent this problem, we consider only intersection sites  $u_i$ in one of the boxes $\bbA_k(y_i)$, which are depicted here in gray.}
\label{fig:4}
\end{center}
\end{figure}

We now pause the argument to establish the following lemma.
\begin{lemma}\label{lem:G}
For $d\ge4$, there exist $\ep>0$ and $\alpha_0=\alpha_0(\ep)>0$ large enough such that for every $n^{\alpha_0}\le N\le \xi(\beta)$ and for every ${\bf u}$ with $u_i\in \bbA_{y_i}(2^{k_i})$ for some $k_i$ with $m\le 2^{k_i}\le M/2$ for every $1\le i\le t$,
\begin{align}
\Big(\prod_{i=1}^ta_{x_i,y_i}(u_i)\Big)^{-1}{\bf P}^{{\bf xy},\emptyset}[u_i\stackrel{\n_i+\n'_i}\longleftrightarrow x_i,\forall i\le t,G(u_1,\dots,u_t)^c]
&={\bf P}^{{\bf xu},{\bf uy}}[G(u_1,\dots,u_t)^c]\nonumber\\
&\le s\big(\tfrac nN\big)^{\ep}.
\end{align}
\end{lemma}

\begin{proof} Fix $\ep>0$ sufficiently small (we will see below how small it should be). The first identity follows from the switching lemma so we focus on the second one. Let $G_i$ be the event that the current ${\bf k}_i$ exists. This event clearly contains (see Fig.~\ref{fig:4}) the event that ${\rm Ann}(M,N)$ is not crossed by a cluster in $\n_i$, and ${\rm Ann}(n,m)$ is not crossed by a cluster in $\n'_i$, since in such case ${\bf k}_i$ can be defined as the sum of $\n_i$ restricted to the clusters intersecting $\Lambda_N^c$ (this current has no sources) and $\n_i'$ restricted to the clusters intersecting $\Lambda_m^c$ (this current has sources $u_i$ and $y_i$). We focus on the probability of this event for $i\le t$, the case $t<i\le s$ being even simpler since there are no sources.

We  bound the probability of $\n_i$ crossing ${\rm Ann}(M,N)$  by splitting ${\rm Ann}(M,N)$ in two annuli ${\rm Ann}(M,R)$ and ${\rm Ann}(R,N)$ with $R=\sqrt{MN}$, then estimating the probability that the backbone of $\n_i$ crosses the inner annulus, and then the probability that the remaining current crosses the outer annulus. More precisely, the chain rule for backbones \cite{AizBarFer87} gives that for $\alpha_0=\alpha_0(\ep)>0$ large enough and $N\ge n^{\alpha_0}$, 
\be\label{eq:aahaha1}
{\bf P}^{\bf xy}[\Gamma(\n_i)\text{ crosses }{\rm Ann}(M,R)]\le \sum_{v\in \partial\Lambda_{R}}\frac{\langle\sigma_{x_i}\sigma_v\rangle\langle\sigma_v\sigma_{u_i}\rangle}{\langle\sigma_{x_i}\sigma_{u_i}\rangle}\le C_2R^3 \frac{M^3}{R^4}\le (n/N)^\ep,
\ee
where  the lower bound \eqref{eq:Simon-Lieb} to bound the denominator and the Infrared Bound~\eqref{eq:IB} for the numerator. 
Then, observe that the remaining current $\n_i\setminus\Gamma(\n_i)$ is sourceless. Adding an additional sourceless current  and using the switching lemma and Griffiths inequality \cite{Gri67} 
(very much like in the bound \eqref{eq:bound sourceless} in the proof of Lemma~\ref{lem:intersection}) gives
\be\label{eq:aahaha}
{\bf P}^{\bf xy}[\n_i\setminus\Gamma(\n_i)\text{ crosses }{\rm Ann}(R,N)]\le\sum_{\substack{v\in \partial\Lambda_R\\ w\in\partial\Lambda_N}}\langle\sigma_v\sigma_w\rangle^2\le C_3R^3N^3(R/N)^4\le (n/N)^\ep,
\ee
where we used the Infrared-Bound~\ref{eq:IB}, and in the last one the definition of $R$ and the fact that $N\ge n^{\alpha_0}$ for $\alpha_0$ large enough.

When dealing with the probability of $\n'_i$ crossing ${\rm Ann}(n,m)$, fix $r=\sqrt{nm}$ and apply the same reasoning with the annuli ${\rm Ann}(n,r)$ and ${\rm Ann}(r,m)$.
The equivalent of \eqref{eq:aahaha} is the same as before, but one must be more careful about the bound on the probability of the event dealing with the backbone:
\be\label{eq:B bound}
{\bf P}^{\bf xy}[\Gamma(\n_i')\text{ crosses }{\rm Ann}(r,m)]\le C_4\sum_{v\in \partial\Lambda_r}\frac{\langle\sigma_{u_i}\sigma_v\rangle\langle\sigma_v\sigma_{y_i}\rangle}{\langle\sigma_{u_i}\sigma_{y_i}\rangle}\le C_5r^3/m^2\le  (n/N)^\ep,
\ee
where we used the Infrared Bound \eqref{eq:IB} and  our assumption that $u_i$ belongs to one of the $\bbA_k(y_i)$ (to show that $\langle\sigma_v\sigma_{y_i}\rangle\le C_4\langle \sigma_{u_i}\sigma_{y_i}\rangle$).
  \end{proof}

Invoking the above lemma we now return  to the proof of Theorem~\ref{thm:mixing multi}. 

 Introduce the coefficients $\delta({\bf u}, {\bf x},{\bf y})$    equal to
\be
\delta({\bf u}, {\bf x},{\bf y}):=\prod_{i=1}^t\frac{a_{x_i,y_i}(u_i)}{|\calK| A_{x_i,y_i}(2^{k_i})}
\ee
for ${\bf u}$ such that for every $i\le t$, $u_i\in \bbA_{y_i}(2^{k_i})$ for some $k_i$, and equal to 0 for other ${\bf u}$.
Gathering \eqref{eq:gf}--\eqref{eq:ggff} as well as Lemma~\ref{lem:G}, and observing that the sum on $(u_1,\dots,u_t)$ of $\delta({\bf u},{\bf x},{\bf y})$ is 1, we obtain that 
 \begin{align}\label{eq:aaag}|{\bf P}^{{\bf xy}}[E\cap F]-\sum_{{\bf u}}\delta({\bf u},{\bf x},{\bf y}){\bf P}^{{\bf xu}}[E]{\bf P}^{\bf uy}[F]|&
\le \frac{C_5s}{\sqrt{\log(N/n)}}+2C_6s(n/N)^\ep\le \frac{C_7s}{\sqrt{\log(N/n)}},
\end{align}
provided that $N\ge n^{\alpha_0}$ where $\alpha_0$ is given by the previous lemma.

To conclude the proof is now a matter of elementary algebraic manipulations. We begin by proving \eqref{eq:independence} when all the $y_i,y'_i$ for $i\le t$ belong to regular scales (not necessarily the same ones). In this case, 
apply twice (once for $\bf y$ and once for $\bf y'$) the previous inequality for our event $E$ and the event on the outside being the full event to find
\begin{align}\label{eq:od}
&|{\bf P}^{{\bf xy}}[E]-{\bf P}^{{\bf xy}'}[E]|\le \big|\sum_{{\bf u}}(\delta({\bf u},{\bf x},{\bf y})-\delta({\bf u},{\bf x},{\bf y}')){\bf P}^{{\bf xu}}[E]\big|+\frac{2C_7s}{\sqrt{\log(N/n)}}.
\end{align}
Since all the $y_i,y'_i$ are in regular scales,   Remark~\ref{rmk:aa} implies that $\bbA_{y_i}(2^{k_i})=\bbA_{y_i'}(2^{k_i})={\rm Ann}(2^{k_i},2^{k_i+1})$. Furthermore, Property 2 of regular scales implies\footnote{Note that in this case $\delta({\bf u},{\bf x},{\bf y})$ and $\delta({\bf u},{\bf x},{\bf y}')$ are both close to 
\[\delta'({\bf u},{\bf x}):=\prod_{i\le t}\frac{\langle\sigma_{x_i}\sigma_{u_i}\rangle}{|\calK|\sum_{v_i\in {\rm Ann}(2^{k_i},2^{k_i+1})}\langle\sigma_{x_i}\sigma_{v_i}\rangle} .\]} that 
\be|\delta({\bf u},{\bf x},{\bf y})-\delta({\bf u},{\bf x},{\bf y}')|\le C_8s\tfrac MN\delta({\bf u},{\bf x},{\bf y})\le C_9s(n/N)^{1/3}\delta({\bf u},{\bf x},{\bf y}).\ee
Therefore, \eqref{eq:independence} follows readily (with a large constant $C_{10}$) in this case. The same argument works for the second identity~\eqref{eq:independence2} for every ${\bf x},{\bf x}'$ and ${\bf y}$, noticing that for every regular ${\bf u}$ for which the coefficients are non-zero,
\be|\delta({\bf u},{\bf x}',{\bf y})-\delta({\bf u},{\bf x},{\bf y})|\le C_{10}s\tfrac mn\delta({\bf u},{\bf x},{\bf y})\le C_{11}s(n/N)^{1/3}\delta({\bf u},{\bf x},{\bf y}).\ee

Now, assume further that $N\ge n^{3\alpha_0}$. Consider ${\bf z}=(z_1,\dots,z_t)$ with $z_i$ in a regular scale for every $i\le t$ and $z_i\in {\rm Ann}(m,M)$. We also pick ${\bf y}$ on which we do not assume anything. The fact that the $\delta({\bf u},{\bf x},{\bf y})$ sum to 1 implies that 
\begin{align}
|{\bf P}^{{\bf xy}}[E]-{\bf P}^{{\bf xz}}[E]|&=\big|{\bf P}^{{\bf xy}}[E]-\sum_{{\bf u}}\delta({\bf u},{\bf x},{\bf y}){\bf P}^{{\bf xz}}[E]\big|\nonumber\\
&\le \big|{\bf P}^{{\bf xy}}[E]-\sum_{{\bf u}}\delta({\bf u},{\bf x},{\bf y}){\bf P}^{{\bf xu}}[E]\big|+\frac{C_6s}{\sqrt{\log(m/n)}}\le \frac{C_7s}{\sqrt{\log(N/n)}},
\end{align}
where in the second line we have used \eqref{eq:od} for ${\bf u}$ and ${\bf z}$ for $n$ and $m$ which is justified since $m/n=(N/n)^{1/3}\ge n^{\alpha_0}$.

Finally, we repeat the reasoning above. For $N\ge n^{9\alpha_0}$, the proof of \eqref{eq:mixing} is obtained by applying \eqref{eq:aaag} for $E$, $F$, $N$ and $n$, the previous inequality for $E$, $m$ and $n$ (which is justified since $m/n=(N/n)^{1/3}\ge n^{3\alpha_0}$), and then again \eqref{eq:aaag} for $F$, $N$ and $n$. This concludes the proof with $\alpha=9\alpha_0$.

\subsection{Proof of Proposition~\ref{prop:gaussian b}}\label{sec:4.3}

As mentioned in the introduction, Newman \cite{New75} showed,  using the model's Lee-Yang property,  that gaussianity of the limit is implied by the asymptotic vanishing of the scaling limit of  $|U_4^\beta|$. A more quantitative proof is available  through the  inequality, valid for every $n\ge2$ and derived using the switching lemma in \cite[Proposition 12.1]{Aiz82},
\begin{align} \label{nGaussian}
0 &\leq \mathcal {G}_n[S_\beta] (x_1,\dots,x_{2n}) -S_\beta(x_1,\dots, x_{2n})\nonumber\\    
&\le - \frac{3}{2} \sum_{1\leq i<j<k<l \leq 2n}  S_{\beta}(x_1,\dots,\cancel {x_i},\dots, \cancel {x_j},\dots, \cancel {x_k},\dots, \cancel {x_l},\dots, x_{2n})  
\,U_4^\beta(x_i, x_j, x_k,x_l). 
\end{align}
Combined with the improved tree diagram bound \eqref{eq:improved tree bound}, this inequality has the following consequence: for a continuous function $f$ which vanishes outside $[-r,r]^4$, and for $\beta\le\beta_c$ and $L\le \xi(\beta)$,  
\begin{equation}\label{eq:moment}
\big|\langle T_{f,L}(\sigma)^{2n}\rangle_\beta-\tfrac{(2n)!}{2^nn!}\langle T_{f,L}(\sigma)^2\rangle_\beta^n\big|~\le~ \tfrac32 (2n)^4\langle T_{f,L}(\sigma)^{2n-4}\rangle_\beta \, \|f\|_\infty^4 \,S(L,r,\beta),
\end{equation}
where 
\be\label{eq:bubble log bound}
S(L,r,\beta)~:=\sum_{\substack{x\in\Z^d\\ x_1,\dots,x_4\in\Lambda_{rL}}}2\frac{\langle\sigma_x\sigma_{x_1}\rangle_\beta\langle\sigma_x\sigma_{x_2}\rangle_\beta\langle\sigma_x\sigma_{x_3}\rangle_\beta\langle\sigma_x\sigma_{x_4}\rangle_\beta}{\Sigma_L(\beta)^2\cdot B_{L(x_1,\dots,x_4)}(\beta)^{c}},
\ee
where $L(x_1,\dots,x_4)$ is the minimal distance between the $x_i$. Also note that by flip symmetry we find that $\langle T_{f,L}(\sigma)^{2n+1}\rangle_\beta=0$.
Multiplying \eqref{eq:moment} by $\tfrac{z^{2n}}{(2n)!}$ and summing over $n$, we obtain
\begin{align}\label{S_L}
 \Big|\langle \exp[z\,T_{f,L}(\sigma)]\rangle_\beta- \exp[\tfrac{z^2}2\langle T_{f,L}(\sigma)^2\rangle_\beta]\Big|~\le~ \exp[\tfrac{z^2}2\langle T_{f,L}(\sigma)^2\rangle_\beta] C_1z^4\|f\|_\infty^4 S(L,r,\beta).\end{align}
Proposition~\ref{prop:gaussian b} therefore
follows from the bound
\be\label{eq:bound S}
S(L,r,\beta)\le C_2r^{12}\Big(\frac{\log \log L}{\log L}\Big)^c
\ee
that we prove next. 
Note that \eqref{S_L} is easy to obtain under the power-law assumption \eqref{eq:S}. Without this assumption, the derivation requires  a few lines of (not particularly interesting) computations to address  the fact that we do not a priori know that $B_L(\beta)$ grows logarithmically. Roughly speaking, the argument below uses the idea that in the case in which $B_L(\beta)$ is small, then $\Sigma_L(\beta)$ is much smaller than $L^2$, so that both combine in such a way that the previous bound is always valid for $c>0$ sufficiently small. 
We now provide the proof of \eqref{eq:bound S}.  Fix $0<a<1$ (any choice would do) and split the sum into four sums
\begin{equation}
S(L,r,\beta)=\underbrace{\sum_{\substack{x\in \Lambda_{drL}\\ x_1,\dots,x_4\in\Lambda_{rL}\\ L(x_1,\dots,x_4)\ge L^a}}(\cdots)}_{(1)}+\underbrace{\sum_{\substack{x\notin \Lambda_{drL}\\ x_1,\dots,x_4\in\Lambda_{rL}\\ L(x_1,\dots,x_4)\ge L^a}}(\cdots)}_{(2)}+\underbrace{\sum_{\substack{x\in \Lambda_{drL}\\ x_1,\dots,x_4\in\Lambda_{rL}\\ L(x_1,\dots,x_4)< L^a}}(\cdots)}_{(3)}+\underbrace{\sum_{\substack{x\notin \Lambda_{drL}\\ x_1,\dots,x_4\in\Lambda_{rL}\\ L(x_1,\dots,x_4)< L^a}}(\cdots)}_{(4)}.
\end{equation}
\begin{description}
\item[Bound on (1)] We focus on this term and give more details since it is in fact the main contributor. By Lemma~\ref{cor:growth}, 
\be\label{eq:nt}
B_{L(x_1,\dots,x_4)}(\beta)\ge \tfrac1{C_3}B_L(\beta).
\ee
 Summing over the sites in $\Lambda_{rL}$, we get that
 \be\label{eq:oki}
 (1)\le C_3 \frac{|\Lambda_{rL}|\chi_{2rL}(\beta)^4}{\Sigma_L(\beta)^2B_L(\beta)^c}\le C_4r^{12}\Big(\frac{\chi_L(\beta)^2}{L^4B_L(\beta)}\Big)^c,
 \ee
 where in the second inequality we used the sliding-scale Infrared Bound \eqref{eq:multi chi} to bound $\chi_{2drL}(\beta)$ in terms of $\chi_L(\beta)\le C_5L^{-4}\Sigma_L(\beta)$ and the Infrared Bound \eqref{eq:IB} to   
  write
  \be\label{eq:ahaha}
\chi_L(\beta)\le C_6L^2.
\ee
Applying Cauchy-Schwarz for the first inequality below,  then bounding the two terms in the middle by \eqref{eq:ahaha} and Lemma~\ref{cor:growth} correspondingly, 
we find that 
\begin{align}
\frac{\chi_L(\beta)^2}{L^4B_L(\beta)}
&\le 2\frac{\chi_{L/\log L}(\beta)^2}{L^4}+C_7\frac{B_L(\beta)-B_{L/\log L}(\beta)}{B_L(\beta)}\le C_8\frac{\log \log L}{\log L}.
\end{align}
Plugging this estimate in \eqref{eq:oki} gives
\be
(1)\le C_9r^{12}\Big(\frac{\log \log L}{\log L}\Big)^c.
\ee

\item[Bound on (2)] Combine \eqref{eq:MMS3} and  the sliding-scale Infrared Bound \eqref{eq:multi chi} to get that for $i=1,\dots,4$, 
\be\label{eq:aha}
\langle\sigma_x\sigma_{x_i}\rangle_\beta\le \frac{C_{10}}{|x|^4}\chi_{|x|/d}(\beta)\le \frac{C_{11}}{L^2|x|^2}\chi_L(\beta).
\ee
Summing over the sites gives the same bound as in \eqref{eq:oki} so that the reasoning in (1) gives
\be
(2)\le C_{12}r^{12}\Big(\frac{\log \log L}{\log L}\Big)^c.
\ee
\item[Bound on (3)] This term is much smaller than the previous two due to the constraint that two sites must be close to each other. In fact, we will not even need the improved part of the tree diagram bound and  will simply use that $B_{L(x_1,\dots,x_4)}(\beta)\ge1$. Then, we use the Infrared Bound \eqref{eq:IB} to bound the terms $\langle\sigma_x\sigma_{x_i}\rangle_\beta$ and $\langle\sigma_x\sigma_{x_j}\rangle_\beta$ for which $x_i$ and $x_j$ are at a distance exactly $L(x_1,\dots,x_4)$. Summing over the other two sites $x_k$ and $x_l$ gives a contribution bounded by $\chi_{2rL}(\beta)^2\le C_{13}L^{-8}r^4\Sigma_L(\beta)^2$ by the sliding-scale Infrared Bound~\eqref{eq:multi chi}. Summing over $x$ and then $x_i$ and $x_j$ gives that 
\be
(3)\le \frac{C_{14}r^4\log(Lr)}{L^{4-4a}}.
\ee
\item[Bound on (4)] This sum is even simpler to bound than (3). Again, we simply use $B_{L(x_1,\dots,x_4)}(\beta)\ge1$, bound two of the terms $\langle\sigma_x\sigma_{x_i}\rangle_\beta$ using \eqref{eq:aha}, and the other two using the Infrared Bound \eqref{eq:IB}. Summing over the vertices and using the constraint that two of the sites must be close to each other gives the bound 
\be
(4)\le \frac{C_{15}r^8}{L^{4-4a}}.
\ee
\end{description}
In conclusion, all the sums (1)--(4) are sufficiently small (recall that by definition $r\ge1$) and the claim is derived. 
\begin{remark}
For $\beta<\beta_c$, applying \eqref{eq:bubble log bound} with $r=1$ and $L=\xi(\beta)$  gives the following bound on the {\em renormalized coupling constant} $g(\beta)$:
\begin{equation}\label{eq:RCC}
g(\beta):=\frac1{\xi(\beta)^4\chi(\beta)^2}\sum_{x,y,z\in \bbZ^d}|U_4^\beta(0,x,y,z)|\le \log(\tfrac1{|\beta-\beta_c|})^{-c},\end{equation}
where we used that 
\[\xi(\beta)^2\ge c_0\chi_{\xi(\beta)}(\beta)\ge c_1\chi(\beta)\ge c_2/(\beta_c-\beta) \,. \]
(The first inequality follows from the infrared bound, the second is a classical bound obtained first by Sokal \cite{Sok82}, and the third is a mean-field 
lower bound on $\chi(\beta)$ \cite{Aiz82}).
In field theory this quantity  is often referred to as the (dimensionless)  renormalized coupling constant.  In \cite{HarTas87} it was proved that for lattice $\phi^4_4$ measures of small enough $\lambda$ it  converges to $0$ at the rate $1/\log(\tfrac1{|\beta-\beta_c|})$.  Such behaviour is  expected to be true, in  dimension $d=4$, also for the n.n.f.~Ising model.
\end{remark}

\section{Generalization to models in the Griffiths-Simon class}\label{sec:6}

In this section we extend the results to nearest-neighbor ferromagnetic models in the GS class. 
An important observation is that  the results from the previous section also extend. 
Note, however, that $\rho$ can have unbounded support, so that to be of relevance the  relations of interest need to be expressed in spin-dimension balanced forms. 
Once this is done, many of the basic diagrammatic bounds which are available for the Ising model extend to the GS class essentially by linearity, and then to the GS class by continuity.   
Below, we carefully present the generalizations.

In the whole section, $U_4^{\rho,\beta}$ denotes the 4-point Ursell function of the $\tau$ variables, and $B_L(\rho,\beta)$ the bubble diagram truncated at a distance $L$. We also reuse the notation $\xi(\rho,\beta)$ and $\beta_c(\rho)$ introduced in Section~\ref{sec:3}.

\subsection{An improved tree diagram bound for models in the GS class}
 
For bounds which are not homogeneous in the spin dimension, one needs to pay attention to the fact that $\tau$ is neither dimensionless nor bounded, and prepare the extension by reformulating the Ising relations in a spin-dimensionless form.   
 
For example,  the basic tree diagram bound \eqref{tree} has four Ising spins on the left side and four pairs on the right.  An extension of the inequality to GS models can be reached by site-splitting the terms in which  an Ising spin is repeated, using the inequality \eqref{eq:genSL}
 (which has a simple proof by means of the switching lemma)
 \footnote{An alternative method for reducing a diagrammatic expression's spin-dimension is to divide by $\langle \sigma_u^2\rangle_0$.  Both methods are of use, and may be compared through   
 \eqref{sigma^2bound}.}.
   The resulting diagrammatic bounds  may at first glance appear as slightly more complicated than the one for the Ising case, but it has the advantage of being  dimensionally balanced.  That is a required condition for a bound to hold uniformly throughout the GS class of models.
 Additional consideration is needed for the factors by which the tree diagram bound of \cite{Aiz82} is improved here.   Taking care of that we get the following extension of the result, which also covers the $\phi^4$ lattice models.

\begin{theorem}[Improved tree diagram bound for the GS class]\label{thm:improved GS tree bound} There exist $C,c>0$ such that for every  n.n.f.~model in the GS class on $\Z^4$, every $\beta\le \beta_c(\rho)$,  $L\le \xi(\rho,\beta)$ and every $x,y,z,t\in\Z^d$ at distances larger than $L$ of each other,
\be\label{eq:improved tree bound GS}
|U_4^{\rho,\beta}(x,y,z,t)|\le C\Big(
\frac{B_0(\rho,\beta)}{B_L(\rho,\beta)}\Big)^c\sum_{u} \sum_{u',u''} 
\langle\tau_x\tau_u\rangle_{\rho,\beta} \,\beta J_{u,u'}\, \langle\tau_{u'}\tau_y\rangle_{\rho,\beta}\langle\tau_z\tau_u\rangle_{\rho,\beta}\,\beta J_{u,u''}\,  \langle\tau_{u''}\tau_t\rangle_{\rho,\beta}.
\ee
\end{theorem}
Before diving into the proof, note that the improved tree diagram bound implies, as it did for the Ising model, 
the following quantitative bound on the convergence to gaussian of the scaling limit of the $\tau$ field
in four dimensional models with variables  in the GS class. 
\begin{proposition}\label{prop:gaussian b_GS}
There exist two constants $c,C>0$ such that for every n.n.f.~model in the GS class on $\Z^4$, every $\beta\le \beta_c(\rho)$,  $L\le \xi(\rho,\beta)$, every continuous function $f:\R^4\rightarrow\R$ with bounded support and every $z\in\mathbb R$,
\be \label{eq:gauss_GS}
 \Big|\,\big\langle \exp[z\,T_{f,L}(\tau)-\tfrac{z^2}2\langle T_{f,L}(\tau)^2\rangle_{\rho,\beta} ]\big\rangle_{\rho,\beta}\,-\,1\,\Big|~\le~ \frac{C\|f\|_\infty^4r_f^{12}}{(\log L)^c}\,z^4.
\ee
\end{proposition}

We now return to the proof of the improved tree diagram bound, following the path outlined above.  
 The GS class of variables is naturally divided into two kinds. The core consists of those that directly fall under the Definition~\ref{def:rho}.  The rest can be obtained as weak limits of the former. Since the constants in~\eqref{eq:improved tree bound GS} are uniform, it suffices to prove the result for the former to get it for the latter. We therefore focus on site-measures $\rho$ satisfying Definition~\ref{def:rho}, which can directly be represented as Ising measures on a graph where every vertex is replaced by blocks, as explained in the previous section. In this case, we identify $\langle\cdot\rangle_{\rho,\beta}$ with the Ising measure, and $\tau_x$ with the proper average of  Ising's variables. With this identification, we can harvest all the nice inequalities that are given by Ising's theory. In particular, we can use the random current representation.

More explicitly,  to generalize the argument used in the Ising's proof, we introduce the measure ${\bf P}^{xy}$ defined on  the graph   $\Z^d\times \{1,\dots,N\}$  in two steps: 
\begin{itemize}
\item first, sample two integers $1\le i,j\le N$ with probability 
$$Q_{i}Q_{j}\langle\sigma_{x,i}\sigma_{y,j}\rangle_{\rho,\beta}/\langle\tau_x\tau_y\rangle_{\rho,\beta}\,, $$ 
\item second, sample a current according to the measure ${\bf P}^{\{(x,i),(y,j)\}}_{\rho,\beta}$ corresponding to the random current representation of the Ising model $\langle\cdot\rangle_{\rho,\beta}$. 
\end{itemize}
The interpretation of this object is that of a random current with two {\em random sources} $(x,i)\in \calB_x$ and $(y,j)\in \calB_y$. Also note that the superscript $xy$ will unequivocally denote this type of measures (we will avoid using measures with deterministic sources in this section to prevent confusion)  and ${\bf P}^\emptyset_{\rho,\beta}$ which have no sources. 

The interest in ${\bf P}^{xy}_{\rho,\beta}$ over measures with deterministic sets of sources comes from the fact that the probability that the cluster of the sources intersects a set of the form $\calB_u$ can be bounded in terms of correlations of the variables $\tau_x, x\in\Z^d$ (see Proposition~\ref{prop:intersection phi4 upper}).

\begin{proof}[Theorem~\ref{thm:improved GS tree bound}] As mentioned above, every  $\rho$ in the GS class is a weak limit of measures satisfying the first condition of Definition~\ref{def:rho}.  We therefore focus on such measures. 

Exactly like in the case of the Ising model, the core of the proof of Theorem~\ref{thm:improved GS tree bound} will be 
the proof of the intersection-clustering property that we now state and whose proof is postponed after the proof of the theorem. Define $\ell_0=0$ and $\ell_k=\ell_k(\rho,\beta)$ using the same definition (using $B_L(\rho,\beta)$ this time) as in \eqref{eq:def ell}.  Let 
$\calT_u$ be the set of vertices $v\in\Z^d$ such that $\calB_v$ is connected in $\n_1+\n_3$ to a box $\calB_{u'}$ and in $\n_2+\n_4$ to a box $\calB_{u''}$, with $u'$ and $u''$ at graph distance at most 2 of $u$. Note that $\calT_u$ is now a function of $u$ and that it is defined in terms of ``coarse intersections'', i.e.~lattice sites  $v$ such that both clusters intersect $\calB_v$ (but do not necessarily intersect each other). 
 
\begin{proposition}[intersection-clustering bound for the GS class]\label{prop:non isolated G-S}  
For $d=4$ and $D$ large enough, there exists $\delta=\delta(D)$ such that  for every model in the GS class, every $\beta\le\beta_c(\rho)$,  every $K$ such that $\ell_{K}\le \xi(\rho,\beta)$ and every $u,u',u'',x,y,z,t\in\Z^d$ with $u'$ and $u''$ neighbors of $u$ and $x,y,z,t$ at mutual distances larger than $2\ell_{K}$,
\be \label{eq_cluster_GS}
{\bf P}^{ux,uz,u'y,u''t}_{\rho,\beta}[{\bf M}_u(\calT_u;\calL,K)\le \delta K]\le 2^{-\delta K}.
\ee 
\end{proposition}

Postponing the proof of this estimate, we proceed  with the proof of the Theorem. 
 Express $U_4^{\rho,\beta}$ in terms of intersection properties of currents by summing \eqref{eq:inter} over vertices of $\calB_x,\dots,\calB_z$:
 \begin{align}\label{eq:inter phi4} |U_4^{\rho,\beta}(x,y,z,t)|   & \leq    2\langle \tau_x\tau_y\rangle_{\rho,\beta}\langle \tau_z\tau_t\rangle_{\rho,\beta}{\bf P}^{xy,zt,\emptyset,\emptyset}_{\rho,\beta}[ {\mathbf C}_{\n_1+\n_3}(\partial\n_1)\cap{\mathbf C}_{\n_2+\n_4}(\partial\n_2)\ne \emptyset],\end{align}
 where ${\mathbf C}_{\n_1+\n_3}(\partial\n_1)$ and ${\mathbf C}_{\n_2+\n_4}(\partial\n_2)$ refer to the clusters in $\n_1+\n_3$ and $\n_2+\n_4$ of the sources in $\partial\n_1$ and $\partial\n_2$ respectively (we introduce this notation since the sources are not deterministic anymore).
 
 Define $K\ge c\log [B_L(\rho,\beta)/B_0(\rho,\beta)]$ as in the Ising case. 
We now implement the same reasoning as for the Ising model, with the twist that we consider coarse intersections. If ${\mathbf C}_{\n_1+\n_3}(\partial\n_1)$ and ${\mathbf C}_{\n_2+\n_4}(\partial\n_2)$ intersect, then
\begin{itemize}
\item either the number of $u\in \Z^d$ such that ${\bf C}_{\n_1+\n_3}(\partial\n_1)$ and ${\bf C}_{\n_2+\n_4}(\partial\n_2)$ intersect $\calB_u$ is larger than or equal to $2^{\delta K/5}$,
\item or there exists $u\in\Z^d$ such that 
${\bf C}_{\n_1+\n_3}(\partial\n_1)$ and ${\bf C}_{\n_2+\n_4}(\partial\n_2)$ intersect $\calB_u$,
and ${\bf M}_u(\calT_u;\calL,K)<\delta K$. 
\end{itemize}
Using the Markov inequality and \eqref{eq:disconnect source phi4} on the first line, and Lemma~\ref{prop:coarse switching} in the second one, we find (drop $\rho$ and $\beta$ from notation)
\begin{align}
|U_4(&x,y,z,t)|\le 2^{-\delta K/5}\sum_{u,u',u''\in\Z^d}\langle\tau_x\tau_u\rangle\,\beta J_{u,u'}\langle\tau_{u'}\tau_y\rangle\langle\tau_z\tau_u\rangle\,\beta J_{u,u''}\,\langle\tau_{u''}\tau_t\rangle\nonumber\\
&+\sum_{u,u',u''\in\Z^d}\langle\tau_x\tau_u\rangle\,\beta J_{u,u'}\,\langle\tau_{u'}\tau_y\rangle\langle\tau_z\tau_u\rangle\,\beta J_{u,u''}\,\langle\tau_{u''}\tau_t\rangle{\bf P}^{xu,zu,u'y,u''t}[{\bf M}_{u}(\calT_u;\calL,K)<\delta K].
\end{align}

Lemma~\ref{prop:coarse switching} was invoked here since in the present context ${\bf M}_{u}(\calT_u;\calL,K)$  is defined in terms of coarse rather than true intersections. 
The intersection-clustering bound (Proposition~\ref{prop:non isolated G-S}) concludes the proof.
\end{proof}

We now need to prove Proposition~\ref{prop:non isolated G-S}. The proof  itself is exactly the same as for Proposition~\ref{prop:non isolated general} (the monotonicity property of \eqref{cor:1} is not impacted), except for the proofs of the mixing and intersection properties (i.e.~statements corresponding to Lemma~\ref{lem:intersection general} and Theorem~\ref{thm:mixing multi} respectively). Below, we briefly detail the statements and proofs of these results. Let $I_k(0)$ be the event that there exists $v\in {\rm Ann}(\ell_k,\ell_{k+1})$ such that $\calB_v$ is connected in $\n_1+\n_3$ and in $\n_2+\n_4$ to the union of the boxes $\calB_{w}$ with $w$ at a distance at most 2 of $0$.

\begin{lemma}[intersection property for the GS class]\label{lem:intersection phi4}
There exists $c>0$ such that for every $\rho$ in the GS class, every $\beta\le\beta_c(\rho)$, every $k$, every neighbour $0'$ of the origin, and every $y\notin \Lambda_{2\ell_{k+1}}$ in a regular scale,
\be
{\bf P}^{0y,0'y,\emptyset,\emptyset}_{\rho,\beta}[I_k(0)]\ge c.
\ee
\end{lemma}
\begin{proof}
Reuse the notions included in the proofs of the intersection property in previous sections. Let 
\be
\calM:=\sum_{v\in {\rm Ann}(m,M)} \sum_{i,i'=1}^n Q_i^2\mathbb I[\partial\n_1\stackrel{\n_1+\n_3}{\longleftrightarrow}(v,i)]\,Q_{i'}^2\mathbb I[\partial\n_1\stackrel{\n_1+\n_3}{\longleftrightarrow}(v,i')].
\ee
A computation similar to before gives
\begin{align}
{\bf E}^{0y,0'y,\emptyset,\emptyset}_{\rho,\beta}[|\calM|]&\ge c_1(B_{M}(\rho,\beta)-B_{m-1}(\rho,\beta))\\
{\bf E}^{0y,0'y,\emptyset,\emptyset}_{\rho,\beta}[|\calM|^2]&\le C_2B_{\ell_{k+1}}(\rho,\beta)^2.
\end{align}
Now, in the first line we use the same reasoning as below \eqref{eq:BB}. We include it for completeness to see where the division by $B_0(\rho,\beta)$ enters into the game (it is the only place it does). The Infrared Bound \eqref{eq:IB GS} (note that $\langle\tau_0^2\rangle_{\rho,\beta}=B_0(\rho,\beta)$) implies that 
\be
B_{M}(\rho,\beta)-B_{m-1}(\rho,\beta)\ge B_{\ell_{k+1}}(\rho,\beta)-B_{\ell_k}(\rho,\beta)-C_3B_0(\rho,\beta)\ge (1-\tfrac{1+C_3}D)B_{\ell_{k+1}}(\rho,\beta).
\ee
Cauchy-Schwarz therefore implies the fact that $\calM\ne \emptyset$ with positive probability, which implies in particular the existence of a vertex $v\in{\rm Ann}(m,M)$ which is connected in $\n_1+\n_3$ to $\calB_0$ and in $\n_2+\n_4$ to $\calB_0'$.

The second part of the proof bounding the probabilities of $F_1,\dots,F_4$ follows by the same proof as for the Ising model. More precisely,
for $F_1$, the chain rule for backbones \cite{AizBarFer87} and a decomposition on the first edge of the backbone with one endpoint in (a block of a vertex in) $\Lambda_{n-1}$ and the other (in a block of a vertex) in $\partial\Lambda_n$, and then the first edge after this between an endpoint (in a block of a vertex) outside $\Lambda_{\ell_k}$ and one in (a block of a vertex in)  $\Lambda_{\ell_{k}}$ implies that
\begin{equation}{\bf P}^{0y,\emptyset}_{\rho,\beta}[F_1]\le \sum_{\substack{v\in \partial\Lambda_{n}\\w\in \partial\Lambda_{\ell_{k}}\\ v',w'\in\Z^d}}\frac{\langle\tau_0\tau_v\rangle_{\rho,\beta}\,\beta J_{v',v}\ \langle\tau_{v}\tau_{w'}\rangle_{\rho,\beta}\,\beta J_{w',w}\,\langle\tau_{w}\tau_y\rangle_{\rho,\beta}}{\langle\tau_0\tau_y\rangle_{\rho,\beta}}\le C_3 n^3\ell_{k}^3 n^{-4}\le C_4\ell_{k}^{-\ep}.\end{equation}
This inequality uses Property P2 of regular scales, the lower bound \eqref{eq:Simon-Lieb} on the two-point function, and the Infrared Bound~\eqref{eq:IB}.
For $F_3$, the same reasoning as for Ising, with Proposition~\ref{prop:intersection phi4 upper} replacing the switching lemma, leads to \begin{align}
{\bf P}^{0x,\emptyset}_{\rho,\beta}[F_3]&\le \sum_{\substack{v\in \partial\Lambda_{n}\\ w\in\partial\Lambda_m}}{\bf P}^{\emptyset,\emptyset}_{\rho,\beta}[\calB_v\stackrel{\n_1+\n_2}{\longleftrightarrow} \calB_w]\le \sum_{\substack{v\in \partial\Lambda_{n}\\ w\in\partial\Lambda_m\\ v',w'\in\Z^d}}\langle\tau_v\tau_w\rangle\,\beta J_{w,w'}\,\langle\tau_{w'}\tau_{v'}\rangle\,\beta J_{v',v}\le C_5\ell_{k}^{-\ep},\end{align}
where in the last line we used again the Infrared Bound \eqref{eq:IB}. \end{proof}

We now turn to the proof of the mixing property for the measures ${\bf P}^{xy}_\beta$, which is the exact replica of the Ising statement.
\begin{theorem}[mixing of random currents for the GS class]\label{thm:mixing multi G-S}
For $d\ge4$, there exist $\alpha,c>0$ such that for every $\rho$ satisfying Definition~\ref{def:rho}, every $t\le s$, every $\beta\le \beta_c(\rho)$, every $n^\alpha\le N\le\xi(\rho,\beta)$, every $x_i\in\Lambda_n$ and $y_i\notin\Lambda_N$ for every $i\le t$, and every events $E$ and $F$ depending on the restriction of $(\n_1,\dots,\n_s)$ to edges within $\Lambda_n$ and outside of $\Lambda_N$ respectively,  
\begin{align}\label{eq:mixing GS}\big|{\bf P}^{x_1y_1,\dots,x_ty_t,\emptyset,\dots,\emptyset}_{\rho,\beta}[E\cap F]-{\bf P}^{x_1y_1,\dots,x_ty_t,\emptyset,\dots,\emptyset}_{\rho,\beta}[E]{\bf P}^{x_1y_1,\dots,x_ty_t,\emptyset,\dots,\emptyset}_{\rho,\beta}[F]\big|&\le s(\log \tfrac Nn)^{-c}.\end{align}
 Furthermore, for every $x'_1,\dots,x'_t\in\Lambda_n$ and $y'_1,\dots,y'_t\notin\Lambda_N$, 
\begin{align}\label{eq:independence GS}\big|{\bf P}^{x_1y_1,\dots,x_ty_t,\emptyset,\dots,\emptyset}_{\rho,\beta}[E]-{\bf P}^{x_1y_1',\dots,x_ty_t',\emptyset,\dots,\emptyset}_{\rho,\beta}[E]\big|&\le s(\log \tfrac Nn)^{-c},\\
\big|{\bf P}^{x_1y_1,\dots,x_ty_t,\emptyset,\dots,\emptyset}_{\rho,\beta}[F]-{\bf P}^{x_1'y_1,\dots,x_t'y_t,\emptyset,\dots,\emptyset}_{\rho,\beta}[F]\big|&\le s(\log \tfrac Nn)^{-c}.\end{align}
\end{theorem}

\begin{proof}
The beginning is the same as for the Ising model, until the definition of the variable ${\bf N}_i$ that now becomes
\begin{equation}{\bf N}_i:=\frac1{|\calK|}\sum_{k\in \calK}\frac{1}{A_{x_i,y_i}(k)}\sum_{u\in\bbA_k(y_i)}\sum_{j=1}^NQ_{j}^2\,\mathbb I[(u,j)\stackrel{\n_i+\n'_i}\longleftrightarrow \partial\n_i],
\end{equation}
where
$
a_{x,y}(u):=\langle\tau_x\tau_u\rangle\langle\tau_u\tau_y\rangle/\langle\tau_x\tau_y\rangle$ and $ A_{x,y}(k):=\sum_{u\in \bbA_k(y_i)}a_{x,y}(u)$.
The proof of the concentration inequality follows the same lines as in the Ising case. Indeed, the choice of the weight $Q_{j}^2$ enables to rewrite the moments of the random variables ${\bf N}_i$ in terms of the correlations of the random variables $(\tau_z:z\in \Z^d)$. 
The rest of the proof is exactly the same, with trivial changes. For instance, in the proof of Lemma~\ref{lem:G}, one must be careful to derive bounds on probabilities involving $\beta |J|$. This is easily doable  using Proposition~\ref{prop:intersection phi4 upper} exactly like in the previous proof. \end{proof}

\appendix
\section{Appendix}

\subsection{Random currents's partial monotonicity statements}

An inconvenient feature of the random  current representation is the  lack of an FKG-type monotonicity, as the one  valid for the Fortuin-Kasteleyn random cluster models (cf. \cite{Gri06}).     The addition of a pair of  sources may  enhance the configuration, e.g.~forcing a long line where such were rare, but in some situations it may facilitate a split in a connecting line, thereby reducing the current's connectivity properties.    
Nevertheless, some monotonicity properties can still be found, and are used in our analysis.

In this section, we set $\sigma_A$ for the product of the spins in $A$ and write ${\bf C}_\n(S)=\cup_{x\in S}{\bf C}_\n(x)$. 

\begin{lemma}\label{lem:a}  
Let $A,B,S$ be subsets of $\Lambda$ and $F$ a non-negative function defined over pairs of currents, which is determined by just the values of $(\n_1,\n_2)$ along the edges touching the connected cluster ${\bf C}_{\n_1+\n_2}(S) $ and such that $F(\n_1,\n_2)=0$ whenever that cluster intersects $B$ and $(\partial\n_1,\partial \n_2)=(A,B)$.  
Then
\begin{eqnarray}\label{eq:mm}
\,{\bf E}^{A,B}_{\Lambda,\beta}[F(\n_1,\n_2)] ={\bf E}^{A,\emptyset}_{\Lambda,\beta}[F(\n_1,\n_2)\frac{\langle\sigma_B\rangle_{\Lambda\setminus{\bf C}_{\n_1+\n_2}(S),\beta}}{\langle\sigma_B\rangle_{\Lambda,\beta}}] \leq  {\bf E}^{A,\emptyset}_{\Lambda,\beta}[F(\n_1,\n_2)].
\end{eqnarray}
\end{lemma}

\begin{proof}The second inequality is a trivial application of Griffiths' inequality \cite{Gri67}. The first one is proven by a fairly straightforward manipulation involving currents that we now present.  We drop $\beta$ from the notation. Fix $T\subset\Lambda$ not intersection $B$ and choose $F$ given by
\be
F(\n_1,\n_2):=\mathbb I[{\bf C}_{\n_1+\n_2}(S)=T]\,\mathbb I[\n_1=\n]\,\mathbb I[\n_2=\m\text{ on }T]\ee for $\n$ and $\m$ currents on $\Lambda$ and $T$ respectively. For such a choice of function, we find that
\begin{align}
\langle\sigma_A\rangle_\Lambda\langle\sigma_B\rangle_\Lambda\,{\bf E}^{A,B}_\Lambda[F(\n_1,\n_2)]&=\frac{4^{|\Lambda|}}{Z(\Lambda,\beta)^2}\sum_{\n_1:\partial\n_1=A}\sum_{\n_2:\partial\n_2=B}F(\n_1,\n_2)w(\n_1)w(\n_2)\nonumber\\
&=\frac{4^{|\Lambda|}w(\n)w(\m)}{Z(\Lambda,\beta)^2}\sum_{\n_2':\partial\n_2'=B}w(\n_2)\nonumber\\
&=\frac{4^{|\Lambda|}w(\n)w(\m)}{Z(\Lambda,\beta)^2}\langle\sigma_B\rangle_{\Lambda\setminus T}\sum_{\n_2':\partial\n'_2=\emptyset}w(\n_2)\nonumber\\
&=\langle\sigma_A\rangle_\Lambda\langle\sigma_B\rangle_{\Lambda\setminus T}\,{\bf E}^{A,\emptyset}_\Lambda[F(\n_1,\n_2)],
\end{align}
where $\n'_2$ is referring to a current on $\Lambda\setminus T$. In the second line, we used that for $F(\n_1,\n_2)$ to be non-zero, $\n_1$ must be equal to $\n$ and $\n_2$ be decomposed into the current $\m$ on $T$ and a current $\n'_2$ outside $T$ (also, $\n_2(x,y)$ is equal to zero for every $x\in T$ and $y\notin T$). In the last line, we skipped the steps corresponding to going backward line to line to end up with ${\bf E}^{A,\emptyset}_\Lambda[F(\n_1,\n_2)]$.

The proof follows readily for every function $F$ satisfying the assumptions of the lemma. Also, we obtain the result on $\bbZ^d$ by letting $\Lambda$ tend to $\Z^d$.\end{proof}

An interesting application of the lemma is the following pair of \emph{disentangling bounds}. The first inequality appeared in \cite[Proposition  5.2]{Aiz82}, the second is new.
\begin{corollary}\label{cor:1} \label{prop:indep}\label{lem:Aizenman}
For every $\beta>0$, every four vertices $x,y,z,t\in\bbZ^d$ and every set $S\subset\bbZ^d$, 
\begin{align}{\bf P}^{xy,zt}_\beta[{\bf C}_{\n_1+\n_2}(x)\cap{\bf C}_{\n_1+\n_2}(z)\ne \emptyset]&\le{\bf P}^{xy,\emptyset,zt}_\beta[{\bf C}_{\n_1+\n_2}(x)\cap{\bf C}_{\n_3}(z)\ne \emptyset],\\
{\bf P}^{0x,0z,\emptyset,\emptyset}_\beta[{\bf C}_{\n_1+\n_3}(0)\cap{\bf C}_{\n_2+\n_4}(0)\cap S\ne \emptyset]&\le{\bf P}^{0x,0z,0y,0t}_\beta[{\bf C}_{\n_1+\n_3}(0)\cap{\bf C}_{\n_2+\n_4}(0)\cap S\ne \emptyset].
\end{align}
\end{corollary}

\begin{proof}
Fix $\beta>0$, $\Lambda$ finite (the claim will then follow by letting $\Lambda$ tend to $\Z^d$) and drop $\beta$ from the notation. For the first identity, introduce the random variable 
\be
{\bf C}={\bf C}(\n_1,\n_2,\n_3):={\bf C}_{\n_3}({\bf C}_{\n_1+\n_2}(x)).
\ee Lemma~\ref{lem:a} applied in the first and third lines, Griffiths' inequality \cite{Gri67}, and the trivial inclusion ${\bf C}_{\n_1+\n_2}(x)\subset{\bf C}$ in the second, give
\begin{align}{\bf P}^{xy,zt}_\Lambda[{\bf C}_{\n_1+\n_2}(x)\cap{\bf C}_{\n_1+\n_2}(z)= \emptyset]&={\bf E}^{xy,\emptyset}_\Lambda\Big[\mathbb I[z,t\notin {\bf C}_{\n_1+\n_2}(x)]\frac{\langle\sigma_z\sigma_t\rangle_{\Lambda\setminus{\bf C}_{\n_1+\n_2}(x)}}{\langle\sigma_z\sigma_t\rangle_\Lambda}\Big]\nonumber\\
&\ge{\bf E}^{xy,\emptyset,\emptyset}_\Lambda\Big[\mathbb I[z,t\notin {\bf C}]\frac{\langle\sigma_z\sigma_t\rangle_{\Lambda\setminus {\bf C}}}{\langle\sigma_z\sigma_t\rangle_\Lambda}\Big]\nonumber\\
&={\bf P}^{xy,\emptyset,zt}_\Lambda[z,t\notin {\bf C}],\end{align}
which gives the first inequality.

The second identity requires two successive applications of Lemma~\ref{lem:a}. First, conditioning on $\n_2+\n_4$, the proposition applied to $\mathbf S:={\bf C}_{\n_2+\n_4}(0)\cap S$ gives
\be
{\bf P}^{0x,0z,\emptyset,0t}_\Lambda[{\bf C}_{\n_1+\n_3}(0)\cap{\bf C}_{\n_2+\n_4}(0)\cap S\ne \emptyset]\le {\bf P}^{0x,0z,0y,0t}_\Lambda[{\bf C}_{\n_1+\n_3}(0)\cap{\bf C}_{\n_2+\n_4}(0)\cap S\ne \emptyset].\ee
Similarly, conditioning on $\n_1+\n_3$, the proposition applied to $\mathbf S':={\bf C}_{\n_1+\n_3}(0)\cap S$ gives
\be{\bf P}^{0x,0z,\emptyset,\emptyset}_\Lambda[{\bf C}_{\n_1+\n_3}(0)\cap{\bf C}_{\n_2+\n_4}(0)\cap S\ne\emptyset]\le {\bf P}^{0x,0z,\emptyset,0t}_\Lambda[{\bf C}_{\n_1+\n_3}(0)\cap{\bf C}_{\n_2+\n_4}(0)\cap S\ne\emptyset],\ee
thus concluding the proof.
\end{proof}

\subsection{Multi-point connectivity probabilities}
The following two relations facilitate the derivation of estimates guided by the random walk analogy.  

\begin{proposition}\label{prop:3}
For every $x,u,v\in \Z^d$, we have that
\begin{align}\label{eq:prop2b}
{\bf P}^{0x,\emptyset}_\beta[u\stackrel{\n_1+\n_2}\longleftrightarrow 0]&=\frac{\langle\sigma_0\sigma_u\rangle_\beta\langle\sigma_u\sigma_x\rangle_\beta}{\langle\sigma_0\sigma_x\rangle_\beta},\\
\label{eq:prop3b}{\bf P}^{0x,\emptyset}_\beta[u,v\stackrel{\n_1+\n_2}\longleftrightarrow0]&\le \frac{\langle\sigma_0\sigma_v\rangle_\beta\langle\sigma_v\sigma_u\rangle_\beta\langle\sigma_u\sigma_x\rangle_\beta}{\langle\sigma_0\sigma_x\rangle_\beta}+\frac{\langle\sigma_0\sigma_u\rangle_\beta\langle\sigma_u\sigma_v\rangle_\beta\langle\sigma_v\sigma_x\rangle_\beta}{\langle\sigma_0\sigma_x\rangle_\beta}.\end{align}
\end{proposition}
The equality  \eqref{eq:prop2b} is a direct consequence of the switching lemma and has been used several times in the past.  The inequality \eqref{eq:prop3b} is an important new addition, which is proven below.  Its structure  suggests a more general $k$-step random walk type bound, but the present proof does not extend to $k>2$. In particular, if a $k$-step bound could be proven for every $k$, it would improve the concentration estimate for ${\bf N}$ in the proof of mixing from an inverse logarithmic bound to a small polynomial one, which would translate into a similar bound for the mixing property which may be very useful for the study of the critical regime. Note that this would not improve the log correction in our result since the intersection property also requires the $\ell_k$ to grow fast.

\begin{proof}
Fix $\beta>0$ and drop it from the notation. We work with finite $\Lambda$ and then take the limit as $\Lambda$ tends to $\Z^d$. In the whole proof, $\longleftrightarrow$ denotes the connection in $\n_1+\n_2$, and $\not\longleftrightarrow$ denotes the absence of connection.
As mentioned above, \eqref{eq:prop2b} follows readily from the switching lemma.
To prove \eqref{eq:prop3b}, use the switching lemma to find
\begin{align}{\bf P}^{0x,\emptyset}_\Lambda[u,v\longleftrightarrow0]&=\frac{\langle\sigma_0\sigma_u\rangle_\Lambda\langle\sigma_u\sigma_x\rangle_\Lambda}{\langle\sigma_0\sigma_x\rangle_\Lambda}{\bf P}^{0u,ux}_\Lambda[v\longleftrightarrow u].\end{align}
Then, our goal is to show that
\begin{align}\label{eq:imp}{\bf P}^{0u,ux}_\Lambda[v\longleftrightarrow u]&\le{\bf P}^{0u,\emptyset}_\Lambda[v\longleftrightarrow u]+{\bf P}^{\emptyset,ux}_\Lambda[v\longleftrightarrow u]-{\bf P}^{\emptyset,\emptyset}_\Lambda[v\longleftrightarrow u]
\end{align}
which implies \eqref{eq:prop3b} readily using \eqref{eq:prop2b}. 
In order to show \eqref{eq:imp}, set ${\bf C}={\bf C}_{\n_1+\n_2}(v)$ and apply Lemma~\ref{lem:a} to $F(\n_1,\n_2):=\mathbb I[u\not\longleftrightarrow v]$ to obtain 
\begin{align}\label{eq:hj}
{\bf P}^{0u,ux}_\Lambda[u\not\longleftrightarrow v]&={\bf E}^{0u,\emptyset}_\Lambda\Big[\mathbb  I[u\not\longleftrightarrow v]\frac{\langle\sigma_0\sigma_y\rangle_{\Lambda\setminus{\bf C}}}{\langle\sigma_0\sigma_y\rangle_\Lambda}\Big].
\end{align}
Next, apply Lemma~\ref{lem:a}   to 
\be
F(\n_1,\n_2):=\mathbb I[u\not\longleftrightarrow v]\Big(1-\frac{\langle\sigma_0\sigma_x\rangle_{\Lambda\setminus {\bf C}}}{\langle \sigma_0\sigma_x\rangle_\Lambda}\Big)\ge0
\ee (the inequality is due to Griffiths' inequality \cite{Gri67}) to obtain \eqref{eq:imp} thanks to the following inequalities
\begin{align}
{\bf P}^{0x,\emptyset}_\Lambda[u\not\longleftrightarrow v]-{\bf P}^{0x,0y}_\Lambda[u\not\longleftrightarrow v]&={\bf E}^{0x,\emptyset}_\Lambda[F(\n_1,\n_2)]={\bf E}^{\emptyset,\emptyset}_\Lambda\Big[F(\n_1,\n_2)\frac{\langle\sigma_0\sigma_x\rangle_{\Lambda\setminus{\bf C}}}{\langle\sigma_0\sigma_x\rangle_\Lambda}\Big]\nonumber\\
&\le {\bf E}^{\emptyset,\emptyset}_\Lambda[F(\n_1,\n_2)]={\bf P}^{\emptyset,\emptyset}_\Lambda[u\not\longleftrightarrow v]-{\bf P}^{\emptyset,0y}_\Lambda[u\not\longleftrightarrow v].
\end{align}
\end{proof}

\begin{remark}
Griffiths' inequality \cite{Gri67} plugged in \eqref{eq:hj} gives 
\be
\label{eq:imp2}{\bf P}^{0u,ux}[v\longleftrightarrow u]\ge {\bf P}^{0u,\emptyset}[v\longleftrightarrow u].
\ee
\end{remark}
\begin{remark}
The inequalities \eqref{eq:imp} and \eqref{eq:imp2} can be extended to every set $S\subset\bbZ^d$ and every two vertices $x,y\in\bbZ^d$:
\be\label{eq:ag}{\bf P}^{0x,\emptyset}_\beta[0\stackrel{\n_1+\n_2}{\longleftrightarrow}S]\le {\bf P}^{0x,0y}_\beta[0\stackrel{\n_1+\n_2}{\longleftrightarrow}S]\le {\bf P}^{0x,\emptyset}_\beta[0\stackrel{\n_1+\n_2}{\longleftrightarrow}S]+{\bf P}^{\emptyset,0y}_\beta[0\stackrel{\n_1+\n_2}{\longleftrightarrow}S]-{\bf P}^{\emptyset,\emptyset}_\beta[0\stackrel{\n_1+\n_2}{\longleftrightarrow}S].\ee
\end{remark}

\subsection{The spectral representation}

In Section \ref{sec:3.2} we make use of a spectral representation of the correlation function  $S(x):=\langle\tau_0 \tau_x\rangle$.
Though the statement is  well known,  cf.~\cite{GliJaf73} and references therein, for completeness of the presentation following is its derivation.  
For the present purpose it is convenient to present the system's Hamiltonian as the semi-definite function 
\be 
H =  - \sum_{x,y} J_{x,y} (\tau_x - \tau_y)^2 \,. 
\ee 
The difference from the expressions used elsewhere in the paper are the diagonal quadratic terms $\tau_x^2$ whose effect on the Gibbs measure can be incorporated by an adjustment in the spins'~a-priori distribution (which is doable as we assumed in the first place that the site distribution was satisfying \eqref{sub_gauss}).

To avoid burdensome notation, when the domain over which the spins are defined is clear from the context of the discussion we shall use the symbol   
$\tau= \{\tau_x\}_x $ to denote the entire collection of spins in that  region,  and by $\rho_0(d\tau)$ the corresponding product measure.  

\begin{proposition}[Spectral Representation]\label{prop:spec_rep}  Let $\rho_0$ be a single variable distribution for which  the Gibbs states on $\Z^d$ with the n.n.f.~Hamiltonian satisfies 
\be 
\langle |\tau_0|^2 \rangle_\beta  < \infty\ , \quad \forall  \beta \geq 0 \,.
\ee 
Then, for every $0 < \beta < \infty$  and  every  square-summable   $v\in \ell^2(\Z^{d-1})$, there exists a positive measure $\mu_{v,\beta}$ with a total mass satisfying 
\be \mu_{v,\beta}([0,\infty)) \ \leq \    \|v\|_2^2 \, \, \langle |\tau_0|^2 \rangle_\beta 
\ee  
such that for every $n\in \Z$,
 \begin{equation} \label{eq:spec_rep_app}
\sum_{x_\perp,y_\perp\in \Z^{d-1}}v_{x_\perp}\overline{v_{y_\perp}}S_\beta((n,x_\perp-y_\perp))=   
 \int_{0}^\infty e^{-a|n|} d\mu_{v,\beta} ( a).
 \end{equation}
For $\beta < \beta_c$ the measure' support is limited to $a\geq 1/\xi(\beta)$ (here $\xi(\beta)$ is the correlation length of the system). \end{proposition}

In particular, with  $v=\delta_\perp$ the Kronecker function (at the origin) on $\mathbb Z^{d-1}$,  this yields the following spectral representation for the correlation function along a principal axis 
\be \label{eq:S_spec}
 S_{\beta}((n,0_\perp))  =   
 \int_{1/\xi(\beta)}^\infty e^{-a|n|} d\mu_{\delta_\perp,\beta} (a),
 \ee
with a measure whose total mass is   $\mu_{\delta_\perp,\beta}([0,\infty))= \langle |\tau_0|^2 \rangle_\beta$. 

\begin{proof}  Throughout the proof $\beta$ is held constant, and  to a large extent will be omitted from the notation. 
It is convenient to first derive the corresponding statements for finite volume versions of the model,  in tubular domains with periodic boundary conditions $\bbT(m,\ell):=(\bbZ/m\bbZ)\times (\bbZ/\ell\bbZ)^{d-1}$ (with the notational convention
$\bbZ/\infty\bbZ=\bbZ$).   The corresponding finite volume correlation function is naturally denoted 
$S_{m,\ell;\beta}(x):=\langle\tau_0 \tau_x\rangle_{\bbT(m,\ell)}$.

Let $\calV_\ell$ be the $\bbC$-vector space   of 
$L^2(\otimes_{x\in (\bbZ/\ell\bbZ)^{d-1}} \rho(d\tau_x))$ 
 of functions supported on the transversal hyperplane $(\bbZ/\ell\bbZ)^{d-1}$, over the product measure $\otimes \rho(d\tau_x)$. 
 On $\calV_\ell$, let $T_\ell$  be the  self adjoint operator whose  kernel is given by
\be 
T_\ell(\tau,{\tau'}):=\exp\Big\{-\frac{\beta J}{4} \sum_{\substack{\{x,y\}\subset(\bbZ/\ell\bbZ)^{d-1}\\ \{x,y\}\text{ edge}}} [ (\tau_x-\tau_y)^2 +(\tau'_x-\tau'_y)^2] \, -\, \frac{\beta J}{2}\sum_{x\in(\bbZ/\ell\bbZ)^{d-1}}(\tau_x -\tau'_x)^2\Big\} . 
\ee
This operator serves as the ``transfer matrix'' in terms of which the partition function can be presented as a trace: 
\be 
Z_{m,\ell}:= {\rm Tr}(T_\ell^{m})   \,.
\ee 
To express the correlations functions, let us consider the multiplication operators 
\be 
\tau[v]     :=   \sum_{x\in (\Z/\ell\bbZ)^{d-1}}v(x) \tau_x  \, 
\ee
 associated with square summable functions $ v:  (\Z/\ell\bbZ)^{d-1} \longmapsto \bbC$.

In this notation, the correlation function   
of spins 
 at sites (which we write as $(n,x_\perp)\in \bbT_{m,\ell}$) satisfy
\begin{align} \label{eq:212}
\sum_{x_\perp,y_\perp\in (\bbZ/\ell\bbZ)^{d-1}}\overline{v_{y_\perp}}v_{x_\perp} S_{m,\ell;\beta}((n,y_\perp-x_\perp))
&=\frac{{\rm Tr}(T_\ell^{m-n}\, \overline \tau[v]  \, T_\ell^n\, \tau[v])}{{\rm Tr}(T_\ell^{m})}.
\end{align}
We next claim  that for any $\ell<\infty$ the operator $T_\ell $ is: \begin{itemize}[noitemsep,nolistsep]
\item[(i)] self adjoint and compact  (and thus with spectrum which is discrete, except for possible accumulation at $0$);
\item[(ii)] positive definite;
\item[(iii)] non-degenerate at the top of its spectrum, with a strictly positive eigenfunction. 
 \end{itemize} 
Item (i) is implied by the kernel's symmetry and the finiteness of its Hilbert-Schmidt norm:
\be 
{\rm Tr} \,  T_\ell ^*T_\ell \, =\,  \int  \int \rho(d\tau) \rho(d\tau')\,  |T_\ell(\tau,{\tau'})|^2 \, \leq  \, 1  .
\ee   
Positivity (ii) can be deduced by the criteria of \cite{FroSimSpe76} (see also \cite{Bis09}) applied to the reflection symmetry with respect to the hyperplanes passing through mid-edges.  The last assertion (iii) is implied by (i)  combined with the kernel's pointwise positivity  (cf.~Krein-Rutman theorem  \cite{KrRu}).   
 
Rewritten  in terms of the spectral representation of $T_\ell$, \eqref{eq:212} takes the form: 
\be \label{eq:N_inf} 
\sum_{x_\perp,y_\perp\in \Z^{d-1}}v_{x_\perp}\overline{v_{y_\perp}}S_{m,\ell;\beta}((n,x_\perp-y_\perp))
=\frac{\sum_{\lambda_1, \lambda_2\in {\rm Spec}(T_\ell)} {\lambda_1}^{m-n} {\lambda_2}^n\, \bra{e_{\lambda_1}} \overline \tau[v]   \ket{e_{\lambda_2}}\bra{e_{\lambda_2}}   \tau[v]   \ket{e_{\lambda_1}}}{   
\sum_{\lambda\in {\rm Spec}(T_\ell)}
 \lambda^m}, 
\ee 
where $\{\ket{e_\lambda}\}$ is an orthonormal basis of eigenvectors  of $T_\ell $.  
By the structure of the spectrum described above, in the limit $m\to \infty$  only the terms with $\lambda_1= \lambda_{\rm max}$ and  $\lambda= \lambda_{\rm max}$ are of relevance,  and one is left with the single sum:
\begin{eqnarray} \label{101}
\sum_{x_\perp,y_\perp\in \Z^{d-1}}v_{x_\perp}\overline{v_{y_\perp}}S_{\infty,\ell;\beta}((n,x_\perp-y_\perp))
&=&\sum_{\lambda \in {\rm Spec}(P)} (\tfrac{\lambda}{\lambda_{\rm max}})^n  \bra{e_{\lambda_{\rm max}}} \overline T  \ket{e_{\lambda}}\bra{e_{\lambda}}  T  \ket{e_{\lambda_{\rm max}}}  \notag \\[2ex]    
&=: & \int_{0}^\infty e^{-an} d\mu_{v,\beta,\ell} (a),
\end{eqnarray}   
with $ e^{-a} = \lambda/\lambda_{\rm max} $  
and $ \mu_{v,\beta,\ell}$ the above discrete spectral measure (whose support starts at $\xi(\ell,\beta)$, the inverse rate of decay in $x$ of $S_{\infty,\ell}(x)$).   
 
Next we consider the limit $\ell \to \infty$ at fixed $x_\perp-y_\perp$.  It is  known, through the FKG inequality, that the correlation function converges pointwise, i.e.  for any  $(n,x_\perp-y_\perp)$ and 
$\beta$,
\be 
S_{\beta}((n,x_\perp-y_\perp)) \  = \  \lim_{\ell \to \infty} S_{\infty,\ell; \beta}((n,x_\perp-y_\perp)) 
\ee  
Through \eqref{101} this translates into  convergence of the moments of $e^{-a}$ under the measures $\mu_{v,\beta,\ell} $.   The  moment criterion for the convergence of positive measures over  bounded intervals (here  $[0,1]$)  allows to conclude existence of the (weak) limit 
$\lim_{\ell\to \infty} \mu_{v,\beta,\ell} = \mu_{v,\beta}$  
(which need not be a point measure) with which  
the claimed relation  \eqref{eq:spec_rep_app} holds. \end{proof} 

One may observe that the above result applies to all $\beta$.  It may be added that for $\beta <\beta_c$  the  spectral measure's support is bounded away from $0$.  In contrast, for $\beta >\beta_c$ the measures associated with $v$ of non-zero sum  include a point mass there, i.e.~$\mu_{v,\beta} (\{0\}) \neq 0$, which is the spectral representation of the long range order.

\subsection{Intersection properties for random current representation of models in the GS class}

We start with the version of the switching lemma that we will use. Below, $\delta_{uv}$ denotes the current equal to 1 on the edge $uv$ and 0 otherwise.

\begin{lemma}[Coarse switching]\label{prop:coarse switching}
Let $S,T$ be two disjoint sets of vertices. For every event $E$ depending on the sum of two currents and every $x\ne y$, 
\begin{align}\label{eq:ap}{\bf P}^{xy,\emptyset}_\beta[x\stackrel{\n_1+\n_2}{\longleftrightarrow}S,\n_1+\n_2\in E]&\le \beta\sum_{a\in S,b\notin S} J_{a,b}\frac{\langle\sigma_x\sigma_a\rangle\langle\sigma_b\sigma_y\rangle}{\langle\sigma_x\sigma_y\rangle}{\bf P}^{xa,by}_\beta[\n_1+\n_2+\delta_{ab}\in E],\\
\label{eq:app}{\bf P}^{\emptyset,\emptyset}_\beta[S\stackrel{\n_1+\n_2}{\longleftrightarrow}T,\n_1+\n_2\in E]&\le \beta^2\sum_{\substack{a\in S,b\notin S\\ s\in T,t\notin T}}J_{a,b}J_{s,t}\langle\sigma_a\sigma_s\rangle\langle\sigma_t\sigma_b\rangle{\bf P}^{as,tb}_\beta[\n_1+\n_2+\delta_{ab}+\delta_{st}\in E].
\end{align}
\end{lemma}

\begin{proof}
We start with the first inequality. Fix $\Lambda$ finite. By multiplying by the quantity $\langle\sigma_x\sigma_y\rangle_{\beta}4^{-|\Lambda|} Z(\Lambda,J,\beta)^2$, and then making the change of variable $\m=\n_1+\n_2$, $\n_2=\n$, we find that 
\begin{align}(1):&=\sum_{\substack{\partial\n_1=\{x,y\}\\ \partial\n_2=\emptyset}}w_\beta(\n_1)w_\beta(\n_2)\mathbb I[\n_1+\n_2\in E] \mathbb I[x\stackrel{\n_1+\n_2}{\longleftrightarrow}S]\nonumber\\
&=\sum_{\partial\m=\{x,y\}}w_\beta(\m) \mathbb I[\m\in E] \mathbb I[x\stackrel{\m}{\longleftrightarrow}S]\sum_{\substack{\n\le\m\\\partial\n=\emptyset}}\binom{\m}{\n}\nonumber\\
&=2^{-|\Lambda|}\sum_{\partial\m=\{x,y\}}w_{2\beta}(\m)\mathbb I[\m\in E] \mathbb I[x\stackrel{\m}{\longleftrightarrow}S]2^{k(\m)},\end{align}
where in the last line we used that the number of even subgraphs of the multi-graph $\calM$ (see for instance definition in \cite{AizDumTasWar18}) associated with $\m$ is given by $2^{|\m|+k(\m)-|\Lambda|}$, where $|\m|$ means the total sum of $\m$, and $k(\m)$ is the number of connected components. Now, observe that \be
w_{2\beta}(\m)\mathbb I[x\stackrel{\m}{\longleftrightarrow}S]2^{k(\m)}\le \sum_{a\notin S,b\in S}\beta J_{a,b}w_{2\beta}(\m-\delta_{ab})\mathbb I[x\stackrel{\m-\delta_{ab}}{\longleftrightarrow}b]\mathbb I[\m_{ab}\ge1]2^{k(\m-\delta_{ab})}.
\ee
Indeed, we are necessarily in one of the following cases: consider the edges $ab$ with $a\in S$, $b\notin S$, and $a$ connected to $y$ in $\m-\delta_{ab}$. Assume that 
\begin{itemize}[noitemsep]
\item there is an edge $ab$ as above with $\m_{ab}\ge 2$, in such case $k(\m-\delta_{ab})=k(\m)$ and $w_{2\beta}(\m)=\frac{2\beta J_{ab}}{\m_{ab}}w_{2\beta}(\m-\delta_{ab})\le \beta J_{a,b} w_{2\beta}(\m-\delta_{ab})$;
\item there is a loop in the cluster of $x$ in $\m$ which is intersecting the edge-boundary of $S$, in such case there are two edges $ab$ satisfying the property above, with $k(\m-\delta_{ab})=k(\m)$ and $w_{2\beta}(\m)\le 2\beta J_{a,b} w_{2\beta}(\m-\delta_{ab})$;
\item otherwise, there is only one edge $ab$ with $\m_{ab}=1$, in such case $k(\m-\delta_{ab})=k(\m)+1$ and $w_{2\beta}(\m)\le 2\beta J_{a,b}w_{2\beta}(\m-\delta_{ab})$.
\end{itemize}
Injecting the last displayed inequality in the first one, and then making the change of variable $\m'=\m-\delta_{uv}$, we find that 
\begin{align}2^{|\Lambda|}\times(1)&\le\sum_{a\notin S,b\in S}\beta J_{a,b}\sum_{\partial\m'=\{x,y,a,b\}}w_{2\beta}(\m')2^{k(\m')}\mathbb I[x\stackrel{\m'}\longleftrightarrow a]\mathbb  I[\m'+\delta_{ab}\in E]\nonumber\\
&=\sum_{a\notin S,b\in S}\beta J_{a,b}\sum_{\substack{\partial\n_1=\{x,a,b,y\}\\\partial\n_2=\emptyset}}w_{\beta}(\n_1)w_\beta(\n_2)\mathbb I[x\stackrel{\n_1+\n_2}\longleftrightarrow a]\mathbb  I[\n_1+\n_2+\delta_{ab}\in E]\nonumber\\
&=\sum_{a\notin S,b\in S}\beta J_{a,b}\sum_{\substack{\partial\n_1=\{b,y\}\\\partial\n_2=\{x,a\}}}w_{\beta}(\n_1)w_\beta(\n_2)\mathbb  I[\n_1+\n_2+\delta_{ab}\in E],
\end{align}
where in the last line we used the switching lemma. Dividing 
this relation by the factor
$\langle\sigma_x\sigma_y\rangle_{\Lambda,\beta} 4^{-|\Lambda|}Z(\Lambda,J,\beta)^2$ and letting $\Lambda$ tend to the full lattice implies the first claim.

The second claim follows from the same reasoning using pairs of edges $(ab,st)$ with $a\in S$, $b\notin S$, $s\in T$ and $t\notin T$ such that $a$ is connected to $s$ in $\n_1+\n_2$.
\end{proof}

We deduce the following pair of diagrammatic bounds on the connectivity probabilities. 

\begin{proposition}\label{prop:intersection phi4 upper}
For every distinct $x,y,u,v\in\Z^d$ 
\begin{align}
 {\bf P}^{\emptyset,\emptyset}_{\rho,\beta}[\calB_x\stackrel{\n_1+\n_2}{\longleftrightarrow} \calB_y]&\le\sum_{x',y'\in\Z^d}\langle\tau_x\tau_y\rangle\,\beta J_{y,y'}\,\langle\tau_{y'}\tau_{x'}\rangle\,\beta J_{x',x}\,
\label{eq:disconnect source phi4 sourceless},\\
{\bf P}^{xy,\emptyset}_{\rho,\beta}[\partial\n_1\stackrel{\n_1+\n_2}{\longleftrightarrow} \calB_u]&\le\sum_{u'\in\Z^d}\frac{\langle\tau_x\tau_u\rangle\,\beta J_{u,u'}\,\langle\tau_{u'}\tau_x\rangle}{\langle\tau_x\tau_y\rangle}\label{eq:disconnect source phi4}.
\end{align} 
\end{proposition}

\begin{proof}
For the first one, sum \eqref{eq:app} for $E$ being the full event and vertices in $\calB_x$ and $\calB_y$, and use \eqref{eq:prop2b}. For the second one, do the same with \eqref{eq:ap} instead.
\end{proof}

 \paragraph{Acknowledgments}  The work of M.~Aizenman on this project was supported in part by the NSF grant DMS-1613296, and that of H.~Duminil-Copin by the NCCR SwissMAP, the Swiss NSF and an IDEX Chair from Paris-Saclay. This project has received funding from the European Research Council (ERC) under the European Union's Horizon 2020 research and innovation programme (grant agreement No.~757296).  The joint work was advanced through mutual visits to Princeton and Geneva University,  sponsored by a  Princeton-Unige partnership grant.  We thank S.~Goswami, A.~Raoufi, P.-F.~Rodriguez, and F.~Severo for stimulating discussions, and M.~Oulamara, R.~Panis, P.~Wildemann, and an anonymous referee for careful reading of the paper.

\providecommand{\bysame}{\leavevmode\hbox to3em{\hrulefill}\thinspace}
\providecommand{\MR}{\relax\ifhmode\unskip\space\fi MR }
\providecommand{\MRhref}[2]{
  \href{http://www.ams.org/mathscinet-getitem?mr=#1}{#2}
}
\providecommand{\href}[2]{#2}

\end{document}